%% file: ms.tex
\PassOptionsToPackage{unicode}{hyperref}
\PassOptionsToPackage{hyphens}{url}
\documentclass[
  11pt]{ociamthesis}
\usepackage{lmodern}
\usepackage{amssymb,amsmath}
\usepackage{ifxetex,ifluatex}
\ifnum 0\ifxetex 1\fi\ifluatex 1\fi=0 % if pdftex
  \usepackage[T1]{fontenc}
  \usepackage[utf8]{inputenc}
  \usepackage{textcomp} % provide euro and other symbols
\else % if luatex or xetex
  \usepackage{unicode-math}
  \defaultfontfeatures{Scale=MatchLowercase}
  \defaultfontfeatures[\rmfamily]{Ligatures=TeX,Scale=1}
\fi
\IfFileExists{upquote.sty}{\usepackage{upquote}}{}
\IfFileExists{microtype.sty}{% use microtype if available
  \usepackage[]{microtype}
  \UseMicrotypeSet[protrusion]{basicmath} % disable protrusion for tt fonts
}{}
\makeatletter
\@ifundefined{KOMAClassName}{% if non-KOMA class
  \IfFileExists{parskip.sty}{%
    \usepackage{parskip}
  }{% else
    \setlength{\parindent}{0pt}
    \setlength{\parskip}{6pt plus 2pt minus 1pt}}
}{% if KOMA class
  \KOMAoptions{parskip=half}}
\makeatother
\usepackage{xcolor}
\IfFileExists{xurl.sty}{\usepackage{xurl}}{} % add URL line breaks if available
\IfFileExists{bookmark.sty}{\usepackage{bookmark}}{\usepackage{hyperref}}
\hypersetup{
  pdftitle={Understanding Profunctor Optics : a representation theorem},
  pdfauthor={Guillaume Boisseau},
  hidelinks,
  pdfcreator={LaTeX via pandoc}}
\providecommand{\tightlist}{%
  \setlength{\itemsep}{0pt}\setlength{\parskip}{0pt}}
\makeatletter
\@ifundefined{lhs2tex.lhs2tex.sty.read}%
  {\@namedef{lhs2tex.lhs2tex.sty.read}{}%
   \newcommand\SkipToFmtEnd{}%
   \newcommand\EndFmtInput{}%
   \long\def\SkipToFmtEnd#1\EndFmtInput{}%
  }\SkipToFmtEnd

\newcommand\ReadOnlyOnce[1]{\@ifundefined{#1}{\@namedef{#1}{}}\SkipToFmtEnd}
\usepackage{amstext}
\usepackage{amssymb}
\usepackage{stmaryrd}
\DeclareFontFamily{OT1}{cmtex}{}
\DeclareFontShape{OT1}{cmtex}{m}{n}
  {<5><6><7><8>cmtex8
   <9>cmtex9
   <10><10.95><12><14.4><17.28><20.74><24.88>cmtex10}{}
\DeclareFontShape{OT1}{cmtex}{m}{it}
  {<-> ssub * cmtt/m/it}{}

\DeclareFontShape{OT1}{cmtt}{bx}{n}
  {<5><6><7><8>cmtt8
   <9>cmbtt9
   <10><10.95><12><14.4><17.28><20.74><24.88>cmbtt10}{}
\DeclareFontShape{OT1}{cmtex}{bx}{n}
  {<-> ssub * cmtt/bx/n}{}
\newcommand{\Conid}[1]{\mathit{#1}}
\newcommand{\Varid}[1]{\mathit{#1}}
\newcommand{\anonymous}{\kern0.06em \vbox{\hrule\@width.5em}}

\usepackage{polytable}

\@ifundefined{mathindent}%
  {\newdimen\mathindent\mathindent\leftmargini}%
  {}%

\def\resethooks{%
  \global\let\SaveRestoreHook\empty
  \global\let\ColumnHook\empty}
\newcommand*{\savecolumns}[1][default]%
  {\g@addto@macro\SaveRestoreHook{\savecolumns[#1]}}
\newcommand*{\restorecolumns}[1][default]%
  {\g@addto@macro\SaveRestoreHook{\restorecolumns[#1]}}
\newcommand*{\aligncolumn}[2]%
  {\g@addto@macro\ColumnHook{\column{#1}{#2}}}

\resethooks

\newcommand{\onelinecommentchars}{\quad-{}- }
\newcommand{\commentbeginchars}{\enskip\{-}
\newcommand{\commentendchars}{-\}\enskip}

\newcommand{\visiblecomments}{%
  \let\onelinecomment=\onelinecommentchars
  \let\commentbegin=\commentbeginchars
  \let\commentend=\commentendchars}

\newcommand{\invisiblecomments}{%
  \let\onelinecomment=\empty
  \let\commentbegin=\empty
  \let\commentend=\empty}

\visiblecomments

\newlength{\blanklineskip}
\newcommand{\hsindent}[1]{\quad}% default is fixed indentation
\let\hspre\empty
\let\hspost\empty

\EndFmtInput
\makeatother
\ReadOnlyOnce{polycode.fmt}%
\makeatletter

\newcommand{\hsnewpar}[1]%
  {{\parskip=0pt\parindent=0pt\par\vskip #1\noindent}}

\newcommand{\hscodestyle}{}

\newcommand{\sethscode}[1]%
  {\expandafter\let\expandafter\hscode\csname #1\endcsname
   \expandafter\let\expandafter\endhscode\csname end#1\endcsname}

  {\par\noindent
   \advance\leftskip\mathindent
   \hscodestyle
   \let\\=\@normalcr
   \let\hspre\(\let\hspost\)%
   \pboxed}%
  {\endpboxed\)%
   \par\noindent
   \ignorespacesafterend}

  {\hsnewpar\abovedisplayskip
   \advance\leftskip\mathindent
   \hscodestyle
   \let\hspre\(\let\hspost\)%
   \pboxed}%
  {\endpboxed%
   \hsnewpar\belowdisplayskip
   \ignorespacesafterend}

  {\hsnewpar\abovedisplayskip
   \advance\leftskip\mathindent
   \hscodestyle
   \let\\=\@normalcr
   \(\pboxed}%
  {\endpboxed\)%
   \hsnewpar\belowdisplayskip
   \ignorespacesafterend}

\newcommand{\plainhs}{\sethscode{plainhscode}}

\plainhs

  {\hsnewpar\abovedisplayskip
   \advance\leftskip\mathindent
   \hscodestyle
   \let\\=\@normalcr
   \(\parray}%
  {\endparray\)%
   \hsnewpar\belowdisplayskip
   \ignorespacesafterend}

  {\parray}{\endparray}

  {\(\parray}{\endparray\)}

\def\codeframewidth{\arrayrulewidth}
\RequirePackage{calc}

  {\parskip=\abovedisplayskip\par\noindent
   \hscodestyle
   \arrayrulewidth=\codeframewidth
   \tabular{@{}|p{\linewidth-2\arraycolsep-2\arrayrulewidth-2pt}|@{}}%
   \hline\framedhslinecorrect\\{-1.5ex}%
   \let\endoflinesave=\\
   \let\\=\@normalcr
   \(\pboxed}%
  {\endpboxed\)%
   \framedhslinecorrect\endoflinesave{.5ex}\hline
   \endtabular
   \parskip=\belowdisplayskip\par\noindent
   \ignorespacesafterend}

\newcommand{\framedhslinecorrect}[2]%
  {#1[#2]}

  {\(\def\column##1##2{}%
   \let\>\undefined\let\<\undefined\let\\\undefined
   \newcommand\>[1][]{}\newcommand\<[1][]{}\newcommand\\[1][]{}%
   \def\fromto##1##2##3{##3}%
   }{\) }%

  {\let\orighscode=\hscode
   \let\origendhscode=\endhscode
   \def\endhscode{\def\hscode{\endgroup\def\@currenvir{hscode}\\}\begingroup}
   \orighscode\def\hscode{\endgroup\def\@currenvir{hscode}}}%
  {\origendhscode
   \global\let\hscode=\orighscode
   \global\let\endhscode=\origendhscode}%

\makeatother
\EndFmtInput
\ReadOnlyOnce{forall.fmt}%
\makeatletter

\let\HaskellResetHook\empty
\newcommand*{\AtHaskellReset}[1]{%
  \g@addto@macro\HaskellResetHook{#1}}
\newcommand*{\HaskellReset}{\HaskellResetHook}

\newcommand\hsforall{\global\let\hsdot=\hsperiodonce}
\newcommand*\hsperiodonce[2]{#2\global\let\hsdot=\hscompose}
\newcommand*\hscompose[2]{#1}

\AtHaskellReset{\global\let\hsdot=\hscompose}

\HaskellReset

\makeatother
\EndFmtInput
\ReadOnlyOnce{exists.fmt}%
\makeatletter

\newcommand\hsexists{\global\let\hsdot=\hsperiodonce}

\AtHaskellReset{\global\let\hsdot=\hscompose}

\HaskellReset

\makeatother
\EndFmtInput

\def\commentbegin{\quad\{\ }
\def\commentend{\}}

\usepackage[toc,page]{appendix}
\usepackage[style=alphabetic]{biblatex}

\usepackage{amsthm}
\usepackage[nameinlink]{cleveref}
\theoremstyle{plain}
\newtheorem{thm}{Theorem}[chapter]
\newtheorem{lemma}[thm]{Lemma}
\newtheorem{crly}[thm]{Corollary}
\newtheorem{prop}[thm]{Proposition}
\theoremstyle{definition}
\newtheorem{defn}[thm]{Definition}
\newtheorem{expl}[thm]{Example}
\theoremstyle{remark}
\newtheorem{rmk}[thm]{Remark}
\newtheorem{note}[thm]{Note}
\newtheorem{notation}[thm]{Notation}

\renewcommand{\textcite}{\cref}

\college{St Anne's College}
\degree{MSc in Computer Science}
\degreedate{Trinity 2017}

\usepackage{color}
\usepackage{fancyvrb}

\DefineVerbatimEnvironment{Highlighting}{Verbatim}{commandchars=\\\{\}}
\newenvironment{Shaded}{}{}
\newcommand{\KeywordTok}[1]{\textcolor[rgb]{0.00,0.44,0.13}{\textbf{#1}}}
\newcommand{\DataTypeTok}[1]{\textcolor[rgb]{0.56,0.13,0.00}{#1}}

\newcommand{\CommentTok}[1]{\textcolor[rgb]{0.38,0.63,0.69}{\textit{#1}}}

\newcommand{\OtherTok}[1]{\textcolor[rgb]{0.00,0.44,0.13}{#1}}
\newcommand{\FunctionTok}[1]{\textcolor[rgb]{0.02,0.16,0.49}{#1}}

\newcommand{\OperatorTok}[1]{\textcolor[rgb]{0.40,0.40,0.40}{#1}}

\newcommand{\NormalTok}[1]{#1}
\usepackage[style=alphabetic,]{biblatex}
\title{Understanding Profunctor Optics : a representation theorem}
\author{Guillaume Boisseau}
\date{}

\begin{document}
\maketitle

\begin{abstract}

\emph{Optics}, aka \emph{functional references}, are classes of tools
that allow composable access into compound data structures. Usually
defined as programming language libraries, they provide combinators to
manipulate different shapes of data such as sums, products and
collections, that can be composed to operate on larger structures.
Together they form a powerful language to describe transformations of
data.

Among the different approaches to describing optics, one particular type
of optics, called \emph{profunctor optics}, stands out. It describes
alternative but equivalent representations of most of the common
combinators, and enhances them with elegant composability properties via
a higher-order encoding. Notably, it enables easy composition across
different \emph{optic families}.

Unfortunately, profunctor optics are difficult to reason about, and
linking usual optics with an equivalent profunctor representation has so
far been done on a case-by-case basis, with definitions that sometimes
seem very \emph{ad hoc}. This makes it hard both to analyse properties
of existing profunctor optics and to define new ones.

This thesis presents an equivalent representation of profunctor optics,
called \emph{isomorphism optics}, that is both closer to intuition and
easier to reason about. This tool enables powerful theorems to be
derived generically about profunctor optics. Finally, this thesis
develops a framework to ease deriving new profunctor encodings from
concrete optic families.

\end{abstract}

{
\setcounter{tocdepth}{2}
\tableofcontents
}
\hypertarget{introduction}{%
\chapter{Introduction}\label{introduction}}

Compound data structures (structures made from the composition of
records, variants, collections, \ldots) are pervasive in software
development, and even more so in functional programming languages.
Though these structures are inherently modular, accessing their
components -- that is querying properties, extracting values or
modifying them -- is less so. In Haskell for example, modifying a field
of a record requires explicit unwrapping and rewrapping of the values in
the record, which becomes quickly impractical when nesting records.

In the recent years, a number of mechanisms have emerged to solve this
problem and allow composable access into all sorts of data shapes.
Collectively dubbed \emph{functional references}, or \emph{optics}, they
have been developed in the form of libraries of combinators, and are now
becoming an important tool in the functional programmer's toolbox.

\hypertarget{simple-optics}{%
\section{Simple optics}\label{simple-optics}}

The most common of these optics, called a \emph{simple lens}, allows
getting and setting a value of a given type \ensuremath{\Varid{a}} present in a larger
structure of type \ensuremath{\Varid{s}}. In its most basic form, a lens is simply a record
with two functions \ensuremath{\Varid{get}} and \ensuremath{\Varid{put}} (sometimes called \ensuremath{\Varid{view}} and
\ensuremath{\Varid{update}}):

\begin{hscode}\SaveRestoreHook
\column{B}{@{}>{\hspre}l<{\hspost}@{}}%
\column{E}{@{}>{\hspre}l<{\hspost}@{}}%
\>[B]{}\mathbf{data}\;\Conid{Lens}\;\Varid{a}\;\Varid{s}\mathrel{=}\Conid{Lens}\;\{\mskip1.5mu \Varid{get}\mathbin{::}\Varid{s}\to \Varid{a},\Varid{put}\mathbin{::}\Varid{a}\to \Varid{s}\to \Varid{s}\mskip1.5mu\}{}\<[E]%
\ColumnHook
\end{hscode}\resethooks

\ensuremath{\Varid{get}} retrieves the contained value of type \ensuremath{\Varid{a}}, and \ensuremath{\Varid{put}} replaces it
with a new provided value.

For example, here is a simple lens into the left component of a pair:

\begin{hscode}\SaveRestoreHook
\column{B}{@{}>{\hspre}l<{\hspost}@{}}%
\column{3}{@{}>{\hspre}l<{\hspost}@{}}%
\column{5}{@{}>{\hspre}l<{\hspost}@{}}%
\column{E}{@{}>{\hspre}l<{\hspost}@{}}%
\>[B]{}\Varid{first}\mathbin{::}\Conid{Lens}\;\Varid{a}\;(\Varid{a},\Varid{b}){}\<[E]%
\\
\>[B]{}\Varid{first}\mathrel{=}\Conid{Lens}\;\Varid{get}\;\Varid{put}{}\<[E]%
\\
\>[B]{}\hsindent{3}{}\<[3]%
\>[3]{}\mathbf{where}{}\<[E]%
\\
\>[3]{}\hsindent{2}{}\<[5]%
\>[5]{}\Varid{get}\;(\Varid{x},\Varid{y})\mathrel{=}\Varid{x}{}\<[E]%
\\
\>[3]{}\hsindent{2}{}\<[5]%
\>[5]{}\Varid{put}\;\Varid{x'}\;(\anonymous ,\Varid{y})\mathrel{=}(\Varid{x'},\Varid{y}){}\<[E]%
\ColumnHook
\end{hscode}\resethooks

This lens can be used as follows:

\begin{hscode}\SaveRestoreHook
\column{B}{@{}>{\hspre}l<{\hspost}@{}}%
\column{E}{@{}>{\hspre}l<{\hspost}@{}}%
\>[B]{}\Varid{get}\;\Varid{first}\;(\mathrm{4},\text{\ttfamily \char34 hello\char34}){}\<[E]%
\\
\>[B]{}\mbox{\onelinecomment  \text{\ttfamily 4}}{}\<[E]%
\\
\>[B]{}\Varid{put}\;\Varid{first}\;\mathrm{12}\;(\mathrm{4},\text{\ttfamily \char34 hello\char34}){}\<[E]%
\\
\>[B]{}\mbox{\onelinecomment  \text{\ttfamily \char40{}12\char44{}~\char34{}hello\char34{}\char41{}}}{}\<[E]%
\ColumnHook
\end{hscode}\resethooks

So far, we have simply tupled together two accessors to get more generic
code. The real power of these constructions comes from their
composability: we can define a function:

\begin{hscode}\SaveRestoreHook
\column{B}{@{}>{\hspre}l<{\hspost}@{}}%
\column{3}{@{}>{\hspre}l<{\hspost}@{}}%
\column{5}{@{}>{\hspre}l<{\hspost}@{}}%
\column{E}{@{}>{\hspre}l<{\hspost}@{}}%
\>[B]{}\Varid{composeLens}\mathbin{::}\Conid{Lens}\;\Varid{a}\;\Varid{s}\to \Conid{Lens}\;\Varid{x}\;\Varid{a}\to \Conid{Lens}\;\Varid{x}\;\Varid{s}{}\<[E]%
\\
\>[B]{}\Varid{composeLens}\;(\Conid{Lens}\;\Varid{get}_{\mathrm{1}}\;\Varid{put}_{\mathrm{1}})\;(\Conid{Lens}\;\Varid{get}_{\mathrm{2}}\;\Varid{put}_{\mathrm{2}})\mathrel{=}\Conid{Lens}\;\Varid{get}\;\Varid{put}{}\<[E]%
\\
\>[B]{}\hsindent{3}{}\<[3]%
\>[3]{}\mathbf{where}{}\<[E]%
\\
\>[3]{}\hsindent{2}{}\<[5]%
\>[5]{}\Varid{get}\mathbin{::}\Varid{s}\to \Varid{x}{}\<[E]%
\\
\>[3]{}\hsindent{2}{}\<[5]%
\>[5]{}\Varid{get}\mathrel{=}\Varid{get}_{\mathrm{2}}\hsdot{\circ }{.\;}\Varid{get}_{\mathrm{1}}{}\<[E]%
\\
\>[3]{}\hsindent{2}{}\<[5]%
\>[5]{}\Varid{put}\mathbin{::}\Varid{y}\to \Varid{s}\to \Varid{t}{}\<[E]%
\\
\>[3]{}\hsindent{2}{}\<[5]%
\>[5]{}\Varid{put}\;\Varid{y}\;\Varid{s}\mathrel{=}\Varid{put}_{\mathrm{1}}\;(\Varid{put}_{\mathrm{2}}\;\Varid{y}\;(\Varid{get}_{\mathrm{1}}\;\Varid{s}))\;\Varid{s}{}\<[E]%
\ColumnHook
\end{hscode}\resethooks

\ensuremath{\Varid{composeLens}} takes a lens onto an \ensuremath{\Varid{a}} within an \ensuremath{\Varid{s}}, and another lens
onto an \ensuremath{\Varid{x}} within an \ensuremath{\Varid{a}}, and together makes a lens onto the \ensuremath{\Varid{x}} nested
two levels within an \ensuremath{\Varid{s}}.

The composed \ensuremath{\Varid{get}} is simple: it composes the \ensuremath{\Varid{get}} functions of the two
given lenses to get the \ensuremath{\Varid{x}} contained in the \ensuremath{\Varid{a}} contained in the input
\ensuremath{\Varid{s}}. The composed \ensuremath{\Varid{put}} is slightly trickier to define, but its
definition is standard and can be found in any paper about
lenses\autocite{first_lenses}\autocite{mlenses}\autocite{coalg_update_lenses}.

We can now for example combine the lens we have with itself to define a
new lens into the leftmost component of a 4-tuple:

\begin{hscode}\SaveRestoreHook
\column{B}{@{}>{\hspre}l<{\hspost}@{}}%
\column{E}{@{}>{\hspre}l<{\hspost}@{}}%
\>[B]{}\Varid{firstOf4}\mathbin{::}\Conid{Lens}\;\Varid{a}\;(((\Varid{a},\Varid{b}),\Varid{c}),\Varid{d}){}\<[E]%
\\
\>[B]{}\Varid{firstOf4}\mathrel{=}\Varid{first}\mathbin{`\Varid{composeLens}`}\Varid{first}\mathbin{`\Varid{composeLens}`}\Varid{first}{}\<[E]%
\\[\blanklineskip]%
\>[B]{}\Varid{put}\;\Varid{firstOf4}\;\mathrm{42}\;(((\mathrm{1},\mathrm{2}),\text{\ttfamily \char34 hi\char34}),\mathrm{4}){}\<[E]%
\\
\>[B]{}\mbox{\onelinecomment  \text{\ttfamily \char40{}\char40{}\char40{}42\char44{}~2\char41{}\char44{}~\char34{}hi\char34{}\char41{}\char44{}~4\char41{}}}{}\<[E]%
\ColumnHook
\end{hscode}\resethooks

By defining lenses for simple building blocks, we can then easily
construct lenses for any composition of those blocks.

\bigskip

Optics also often include well-behavedness laws that state invariants
that the optic should maintain.

For example, the usual laws for lenses, called very-well-behavedness
laws, are the following:

\begin{itemize}
\tightlist
\item
  (GetPut) \ensuremath{\Varid{get}\;(\Varid{put}\;\Varid{b}\;\Varid{s})\mathrel{=}\Varid{b}}
\item
  (PutGet) \ensuremath{\Varid{put}\;(\Varid{get}\;\Varid{s})\;\Varid{s}\mathrel{=}\Varid{s}}
\item
  (PutPut) \ensuremath{\Varid{put}\;\Varid{b}\hsdot{\circ }{.\;}\Varid{put}\;\Varid{b'}\mathrel{=}\Varid{put}\;\Varid{b}}
\end{itemize}

Composition also preserves these laws, so that well-behavedness of a
composite optic can be defined in terms of its building blocks.

\bigskip
\bigskip

Lenses do not however cover every datatype. They can only be defined if
a value of type \ensuremath{\Varid{s}} contains exactly one of the values of type \ensuremath{\Varid{a}} being
accessed. A number of different kinds of optics -- called \emph{optic
families} -- have been developed to work with various shapes of data one
may encounter.

For example a \emph{simple prism} allows access into a given variant of
a sum type:

\begin{hscode}\SaveRestoreHook
\column{B}{@{}>{\hspre}l<{\hspost}@{}}%
\column{E}{@{}>{\hspre}l<{\hspost}@{}}%
\>[B]{}\mathbf{data}\;\Conid{Prism}\;\Varid{a}\;\Varid{s}\mathrel{=}\Conid{Prism}\;\{\mskip1.5mu \Varid{match}\mathbin{::}\Varid{s}\to \Varid{a}\mathbin{+}\Varid{s},\Varid{build}\mathbin{::}\Varid{a}\to \Varid{s}\mskip1.5mu\}{}\<[E]%
\ColumnHook
\end{hscode}\resethooks

\ensuremath{\Varid{match}\;\Varid{s}} returns the value contained if \ensuremath{\Varid{s}} is of the correct shape and
otherwise returns the original \ensuremath{\Varid{s}}. \ensuremath{\Varid{build}} creates a value of type \ensuremath{\Varid{s}}
of the correct shape from a value of type \ensuremath{\Varid{a}}.

We can define a prism into \ensuremath{\Conid{Maybe}}, one of the most common Haskell
datatypes:

\begin{hscode}\SaveRestoreHook
\column{B}{@{}>{\hspre}l<{\hspost}@{}}%
\column{3}{@{}>{\hspre}l<{\hspost}@{}}%
\column{5}{@{}>{\hspre}l<{\hspost}@{}}%
\column{E}{@{}>{\hspre}l<{\hspost}@{}}%
\>[B]{}\Varid{just}\mathbin{::}\Conid{Prism}\;\Varid{a}\;(\Conid{Maybe}\;\Varid{a}){}\<[E]%
\\
\>[B]{}\Varid{just}\mathrel{=}\Conid{Prism}\;\Varid{match}\;\Varid{build}{}\<[E]%
\\
\>[B]{}\hsindent{3}{}\<[3]%
\>[3]{}\mathbf{where}{}\<[E]%
\\
\>[3]{}\hsindent{2}{}\<[5]%
\>[5]{}\Varid{match}\;\Conid{Nothing}\mathrel{=}\Conid{Left}\;\Conid{Nothing}{}\<[E]%
\\
\>[3]{}\hsindent{2}{}\<[5]%
\>[5]{}\Varid{match}\;(\Conid{Just}\;\Varid{a})\mathrel{=}\Conid{Right}\;\Varid{a}{}\<[E]%
\\
\>[3]{}\hsindent{2}{}\<[5]%
\>[5]{}\Varid{build}\;\Varid{a}\mathrel{=}\Conid{Just}\;\Varid{a}{}\<[E]%
\ColumnHook
\end{hscode}\resethooks

Here, \ensuremath{\Varid{match}\;\Varid{just}\;\Varid{s}} yields \ensuremath{\Conid{Right}\;\Varid{a}} when the target value \ensuremath{\Varid{a}} is
present in the value \ensuremath{\Varid{s}}, and otherwise (i.e.~when \ensuremath{\Varid{s}} is \ensuremath{\Conid{Nothing}})
returns \ensuremath{\Conid{Left}\;\Varid{s}}. \ensuremath{\Varid{build}} creates a value of type \ensuremath{\Conid{Maybe}\;\Varid{a}} from a value
of type \ensuremath{\Varid{a}} by using the \ensuremath{\Conid{Just}} constructor.

For example:

\begin{hscode}\SaveRestoreHook
\column{B}{@{}>{\hspre}l<{\hspost}@{}}%
\column{E}{@{}>{\hspre}l<{\hspost}@{}}%
\>[B]{}\Varid{match}\;\Varid{just}\;(\Conid{Just}\;\mathrm{42}){}\<[E]%
\\
\>[B]{}\mbox{\onelinecomment  \text{\ttfamily Right~42}}{}\<[E]%
\\
\>[B]{}\Varid{match}\;\Varid{just}\;(\Conid{Nothing}){}\<[E]%
\\
\>[B]{}\mbox{\onelinecomment  \text{\ttfamily Left~Nothing}}{}\<[E]%
\ColumnHook
\end{hscode}\resethooks

Assuming we have defined a \ensuremath{\Varid{composePrism}} function mirroring
\ensuremath{\Varid{composeLens}}, we get similar composability for prisms:

\begin{hscode}\SaveRestoreHook
\column{B}{@{}>{\hspre}l<{\hspost}@{}}%
\column{E}{@{}>{\hspre}l<{\hspost}@{}}%
\>[B]{}\Varid{justjust}\mathbin{::}\Conid{Prism}\;\Varid{a}\;(\Conid{Maybe}\;(\Conid{Maybe}\;\Varid{a})){}\<[E]%
\\
\>[B]{}\Varid{justjust}\mathrel{=}\Varid{just}\mathbin{`\Varid{composePrism}`}\Varid{just}{}\<[E]%
\\[\blanklineskip]%
\>[B]{}\Varid{match}\;\Varid{justjust}\;(\Conid{Just}\;\Conid{Nothing}){}\<[E]%
\\
\>[B]{}\mbox{\onelinecomment  \text{\ttfamily Left~\char40{}Just~Nothing\char41{}}}{}\<[E]%
\\
\>[B]{}\Varid{match}\;\Varid{justjust}\;(\Conid{Just}\;(\Conid{Just}\;\mathrm{42})){}\<[E]%
\\
\>[B]{}\mbox{\onelinecomment  \text{\ttfamily Right~42}}{}\<[E]%
\\
\>[B]{}\Varid{build}\;\Varid{justjust}\;\mathrm{42}{}\<[E]%
\\
\>[B]{}\mbox{\onelinecomment  \text{\ttfamily Just~\char40{}Just~42\char41{}}}{}\<[E]%
\ColumnHook
\end{hscode}\resethooks

\hypertarget{polymorphic-optic-families}{%
\section{Polymorphic optic families}\label{polymorphic-optic-families}}

So far, the optics we have seen were \emph{simple} optics, which means
that they could not change the types of the objects they access. For
example, given a function of type \ensuremath{\Varid{a}\to \Varid{b}}, there is not way to use the
\ensuremath{\Varid{first}} lens defined earlier to get from \ensuremath{(\Varid{a},\Varid{c})} to \ensuremath{(\Varid{b},\Varid{c})}.

Lenses can be extended to support such polymorphic updates by relaxing
some of the type parameters\autocite{polymorphic_update_lenses}:

\begin{hscode}\SaveRestoreHook
\column{B}{@{}>{\hspre}l<{\hspost}@{}}%
\column{E}{@{}>{\hspre}l<{\hspost}@{}}%
\>[B]{}\mathbf{data}\;\Conid{Lens}\;\Varid{a}\;\Varid{b}\;\Varid{s}\;\Varid{t}\mathrel{=}\Conid{Lens}\;\{\mskip1.5mu \Varid{get}\mathbin{::}\Varid{s}\to \Varid{a},\Varid{put}\mathbin{::}\Varid{b}\to \Varid{s}\to \Varid{t}\mskip1.5mu\}{}\<[E]%
\ColumnHook
\end{hscode}\resethooks

The two original type parameters have each been split in two to separate
the ``input'' part from the ``output'' part (i.e.~the contravariant and
covariant parts, respectively).

\ensuremath{\Varid{first}} has the same implementation as earlier, but now gets a more
general type:

\begin{hscode}\SaveRestoreHook
\column{B}{@{}>{\hspre}l<{\hspost}@{}}%
\column{3}{@{}>{\hspre}l<{\hspost}@{}}%
\column{5}{@{}>{\hspre}l<{\hspost}@{}}%
\column{E}{@{}>{\hspre}l<{\hspost}@{}}%
\>[B]{}\Varid{first}\mathbin{::}\Conid{Lens}\;\Varid{a}\;\Varid{b}\;(\Varid{a},\Varid{c})\;(\Varid{b},\Varid{c}){}\<[E]%
\\
\>[B]{}\Varid{first}\mathrel{=}\Conid{Lens}\;\Varid{get}\;\Varid{put}{}\<[E]%
\\
\>[B]{}\hsindent{3}{}\<[3]%
\>[3]{}\mathbf{where}{}\<[E]%
\\
\>[3]{}\hsindent{2}{}\<[5]%
\>[5]{}\Varid{get}\;(\Varid{x},\Varid{y})\mathrel{=}\Varid{x}{}\<[E]%
\\
\>[3]{}\hsindent{2}{}\<[5]%
\>[5]{}\Varid{put}\;\Varid{x'}\;(\anonymous ,\Varid{y})\mathrel{=}(\Varid{x'},\Varid{y}){}\<[E]%
\ColumnHook
\end{hscode}\resethooks

And we can write the desired function:

\begin{hscode}\SaveRestoreHook
\column{B}{@{}>{\hspre}l<{\hspost}@{}}%
\column{E}{@{}>{\hspre}l<{\hspost}@{}}%
\>[B]{}\Varid{map}\mathbin{::}(\Varid{a}\to \Varid{b})\to (\Varid{a},\Varid{c})\to (\Varid{b},\Varid{c}){}\<[E]%
\\
\>[B]{}\Varid{map}\;\Varid{f}\;\Varid{x}\mathrel{=}\Varid{put}\;(\Varid{f}\;(\Varid{get}\;\Varid{x}))\;\Varid{x}{}\<[E]%
\ColumnHook
\end{hscode}\resethooks

Prisms can be similarly generalized:

\begin{hscode}\SaveRestoreHook
\column{B}{@{}>{\hspre}l<{\hspost}@{}}%
\column{E}{@{}>{\hspre}l<{\hspost}@{}}%
\>[B]{}\mathbf{data}\;\Conid{Prism}\;\Varid{a}\;\Varid{b}\;\Varid{s}\;\Varid{t}\mathrel{=}\Conid{Prism}\;\{\mskip1.5mu \Varid{match}\mathbin{::}\Varid{s}\to \Varid{a}\mathbin{+}\Varid{t},\Varid{build}\mathbin{::}\Varid{b}\to \Varid{t}\mskip1.5mu\}{}\<[E]%
\ColumnHook
\end{hscode}\resethooks

\bigskip

\begin{rmk}

As pointed out by Edward Kmett\autocite{mirrored_lenses}, the 4 type
parameters should not vary completely independently; for example, when
\ensuremath{\Varid{a}\mathrel{=}\Varid{b}}, we expect that \ensuremath{\Varid{s}\mathrel{=}\Varid{t}}.

In general, we will expect defined optics to have some form of
parametricity. This notably allows laws defined when types match
(i.e.~on simple optics) to still constrain the general case.

E. Kmett expresses this by saying that optics would best be described as
having two type-level functions as parameters, as follows:

\begin{hscode}\SaveRestoreHook
\column{B}{@{}>{\hspre}l<{\hspost}@{}}%
\column{5}{@{}>{\hspre}l<{\hspost}@{}}%
\column{E}{@{}>{\hspre}l<{\hspost}@{}}%
\>[B]{}\mathbf{type}\;\Conid{Lens'}\;\Varid{inner}\;\Varid{outer}\mathrel{=}{}\<[E]%
\\
\>[B]{}\hsindent{5}{}\<[5]%
\>[5]{}\mathbf{forall}\;\Varid{x}\hsforall \;\Varid{y}\hsdot{\circ }{.\;}\Conid{Lens}\;(\Varid{inner}\;\Varid{x})\;(\Varid{inner}\;\Varid{y})\;(\Varid{outer}\;\Varid{x})\;(\Varid{outer}\;\Varid{y}){}\<[E]%
\ColumnHook
\end{hscode}\resethooks

The \ensuremath{\Varid{first}} lens would then have a type similar to:

\begin{hscode}\SaveRestoreHook
\column{B}{@{}>{\hspre}l<{\hspost}@{}}%
\column{E}{@{}>{\hspre}l<{\hspost}@{}}%
\>[B]{}\Varid{first}\mathbin{::}\Conid{Lens'}\;\Varid{id}\;(\Varid{c},\mathbin{-}){}\<[E]%
\ColumnHook
\end{hscode}\resethooks

This is however not doable in Haskell, thus we keep the 4-parameter
description, which is sufficient to express what we need.

\end{rmk}

\hypertarget{profunctor-optics}{%
\section{Profunctor optics}\label{profunctor-optics}}

As we have seen, using lenses and prisms we can manipulate respectively
product-like structures and sum-like structures. What about data
structures that have both sums and products ? The problem is that an
optic into e.g.~\ensuremath{\Varid{a}\mathbin{\mathbf{\times}}\Varid{b}\mathbin{+}\Varid{c}} would be neither a lens nor a prism, since
we would not be able to define \ensuremath{\Varid{get}} nor \ensuremath{\Varid{build}}.

To manipulate such a type, we would need a new optic family that is more
general than both lenses and prisms.

In general, composing across different optic families is not
straightforward and may require not previously known families. We would
then need operations to cast between compatible optic families, and it
would start to get quite cumbersome.

To solve this problem, an alternative representation of optic families,
known as the \emph{profunctor encoding}\autocite{profunctor_optics}, has
been developed.

\hypertarget{profunctors}{%
\subsection{Profunctors}\label{profunctors}}

The profunctor encoding relies on the notion of a \emph{profunctor},
which is a two-argument functor that is contravariant in its first
argument and covariant in its last. Profunctors can often be understood
as being arrow-like objects that describe some notion of transformation
between two types. \bigskip

\begin{defn}[Profunctor]

A profunctor is a type constructor with two parameters that is an
instance of the following typeclass:

\begin{hscode}\SaveRestoreHook
\column{B}{@{}>{\hspre}l<{\hspost}@{}}%
\column{5}{@{}>{\hspre}l<{\hspost}@{}}%
\column{E}{@{}>{\hspre}l<{\hspost}@{}}%
\>[B]{}\mathbf{class}\;\Conid{Profunctor}\;\Varid{p}\;\mathbf{where}{}\<[E]%
\\
\>[B]{}\hsindent{5}{}\<[5]%
\>[5]{}\Varid{dimap}\mathbin{::}(\Varid{a'}\to \Varid{a})\to (\Varid{b}\to \Varid{b'})\to \Varid{p}\;\Varid{a}\;\Varid{b}\to \Varid{p}\;\Varid{a'}\;\Varid{b'}{}\<[E]%
\ColumnHook
\end{hscode}\resethooks

\ensuremath{\Varid{dimap}} is a two-parameter generalization of \ensuremath{\Varid{fmap}}, with the first
parameter ``backwards'' to account for the contravariance in the first
argument.

\ensuremath{\Varid{dimap}} is expected to verify the following well-behavedness laws,
similar to the \ensuremath{\Varid{fmap}} laws:

\begin{itemize}
\tightlist
\item
  \ensuremath{\Varid{dimap}\;\Varid{id}\;\Varid{id}\mathrel{=}\Varid{id}}
\item
  \ensuremath{\Varid{dimap}\;(\Varid{f'}\hsdot{\circ }{.\;}\Varid{f})\;(\Varid{g}\hsdot{\circ }{.\;}\Varid{g'})\mathrel{=}\Varid{dimap}\;\Varid{f}\;\Varid{g}\hsdot{\circ }{.\;}\Varid{dimap}\;\Varid{f'}\;\Varid{g'}}
\end{itemize}

\end{defn}

\bigskip

\begin{expl}

The simplest example of a profunctor is the function arrow \ensuremath{(\to )}
itself:

\begin{hscode}\SaveRestoreHook
\column{B}{@{}>{\hspre}l<{\hspost}@{}}%
\column{5}{@{}>{\hspre}l<{\hspost}@{}}%
\column{E}{@{}>{\hspre}l<{\hspost}@{}}%
\>[B]{}\mathbf{instance}\;\Conid{Profunctor}\;(\to )\;\mathbf{where}{}\<[E]%
\\
\>[B]{}\hsindent{5}{}\<[5]%
\>[5]{}\Varid{dimap}\mathbin{::}(\Varid{a'}\to \Varid{a})\to (\Varid{b}\to \Varid{b'})\to (\Varid{a}\to \Varid{b})\to (\Varid{a'}\to \Varid{b'}){}\<[E]%
\\
\>[B]{}\hsindent{5}{}\<[5]%
\>[5]{}\Varid{dimap}\;\Varid{f}\;\Varid{g}\;\Varid{h}\mathrel{=}\Varid{g}\hsdot{\circ }{.\;}\Varid{h}\hsdot{\circ }{.\;}\Varid{f}{}\<[E]%
\ColumnHook
\end{hscode}\resethooks

The laws are easy to verify:

\begin{itemize}
\tightlist
\item
  \ensuremath{\Varid{dimap}\;\Varid{id}\;\Varid{id}\;\Varid{h}\mathrel{=}\Varid{h}}
\item
  \ensuremath{\Varid{dimap}\;(\Varid{f'}\hsdot{\circ }{.\;}\Varid{f})\;(\Varid{g}\hsdot{\circ }{.\;}\Varid{g'})\;\Varid{h}\mathrel{=}\Varid{g}\hsdot{\circ }{.\;}\Varid{g'}\hsdot{\circ }{.\;}\Varid{h}\hsdot{\circ }{.\;}\Varid{f'}\hsdot{\circ }{.\;}\Varid{f}\mathrel{=}\Varid{dimap}\;\Varid{f}\;\Varid{g}\;(\Varid{dimap}\;\Varid{f'}\;\Varid{g'}\;\Varid{h})}
\end{itemize}

\end{expl}

\bigskip

\begin{expl}

A more interesting example is \ensuremath{\Conid{Lens}\;\Varid{a}\;\Varid{b}}. Indeed, with \ensuremath{\Varid{a}} and \ensuremath{\Varid{b}}
fixed, \ensuremath{\Conid{Lens}\;\Varid{a}\;\Varid{b}\;\Varid{s}\;\Varid{t}} has the right variance. We can write its
profunctor instance:

\begin{hscode}\SaveRestoreHook
\column{B}{@{}>{\hspre}l<{\hspost}@{}}%
\column{5}{@{}>{\hspre}l<{\hspost}@{}}%
\column{13}{@{}>{\hspre}l<{\hspost}@{}}%
\column{E}{@{}>{\hspre}l<{\hspost}@{}}%
\>[B]{}\mathbf{instance}\;\Conid{Profunctor}\;(\Conid{Lens}\;\Varid{a}\;\Varid{b})\;\mathbf{where}{}\<[E]%
\\
\>[B]{}\hsindent{5}{}\<[5]%
\>[5]{}\Varid{dimap}\mathbin{::}(\Varid{s'}\to \Varid{s})\to (\Varid{t}\to \Varid{t'})\to \Conid{Lens}\;\Varid{a}\;\Varid{b}\;\Varid{s}\;\Varid{t}\to \Conid{Lens}\;\Varid{a}\;\Varid{b}\;\Varid{s'}\;\Varid{t'}{}\<[E]%
\\
\>[B]{}\hsindent{5}{}\<[5]%
\>[5]{}\Varid{dimap}\;\Varid{f}\;\Varid{g}\;(\Conid{Lens}\;\Varid{get}\;\Varid{put})\mathrel{=}\Conid{Lens}\;\Varid{get'}\;\Varid{put'}\;\mathbf{where}{}\<[E]%
\\
\>[5]{}\hsindent{8}{}\<[13]%
\>[13]{}\Varid{get'}\mathbin{::}\Varid{s'}\to \Varid{a}{}\<[E]%
\\
\>[5]{}\hsindent{8}{}\<[13]%
\>[13]{}\Varid{get'}\mathrel{=}\Varid{get}\hsdot{\circ }{.\;}\Varid{f}{}\<[E]%
\\
\>[5]{}\hsindent{8}{}\<[13]%
\>[13]{}\Varid{put'}\mathbin{::}\Varid{b}\to \Varid{s'}\to \Varid{t'}{}\<[E]%
\\
\>[5]{}\hsindent{8}{}\<[13]%
\>[13]{}\Varid{put'}\;\Varid{b}\mathrel{=}\Varid{g}\hsdot{\circ }{.\;}\Varid{put}\;\Varid{b}\hsdot{\circ }{.\;}\Varid{f}{}\<[E]%
\ColumnHook
\end{hscode}\resethooks

The proof of the profunctor laws is a routine proof by equational
reasoning.

\ensuremath{\Conid{Lens}\;\Varid{a}\;\Varid{b}\;\Varid{s}\;\Varid{t}} can therefore be seen as a form of transformation from
\ensuremath{\Varid{s}} to \ensuremath{\Varid{t}}, that has more structure than a simpler \ensuremath{\Varid{s}\to \Varid{t}} arrow.

\end{expl}

\bigskip

\hypertarget{profunctor-encoding}{%
\subsection{Profunctor encoding}\label{profunctor-encoding}}

A profunctor encoding is defined on a case-by-case basis for an optic
family. It works as follows: for a given optic family, we define some
typeclass \ensuremath{\Sigma} that has \ensuremath{\Conid{Profunctor}} as a superclass.

A profunctor encoded optic is then a value of the following type:

\begin{hscode}\SaveRestoreHook
\column{B}{@{}>{\hspre}l<{\hspost}@{}}%
\column{E}{@{}>{\hspre}l<{\hspost}@{}}%
\>[B]{}\mathbf{type}\;\Conid{POptic}\;\Varid{a}\;\Varid{b}\;\Varid{s}\;\Varid{t}\mathrel{=}\mathbf{forall}\;\Varid{p}\hsforall \hsdot{\circ }{.\;}\Sigma\;\Varid{p}\Rightarrow \Varid{p}\;\Varid{a}\;\Varid{b}\to \Varid{p}\;\Varid{s}\;\Varid{t}{}\<[E]%
\ColumnHook
\end{hscode}\resethooks

\hypertarget{lenses}{%
\subsubsection{Lenses}\label{lenses}}

In the case of \ensuremath{\Conid{Lens}} for example, a profunctor-encoded lens is an arrow
of type

\begin{hscode}\SaveRestoreHook
\column{B}{@{}>{\hspre}l<{\hspost}@{}}%
\column{E}{@{}>{\hspre}l<{\hspost}@{}}%
\>[B]{}\mathbf{type}\;\Conid{PLens}\;\Varid{a}\;\Varid{b}\;\Varid{s}\;\Varid{t}\mathrel{=}\mathbf{forall}\;\Varid{p}\hsforall \hsdot{\circ }{.\;}\Conid{Cartesian}\;\Varid{p}\Rightarrow \Varid{p}\;\Varid{a}\;\Varid{b}\to \Varid{p}\;\Varid{s}\;\Varid{t}{}\<[E]%
\ColumnHook
\end{hscode}\resethooks

where

\begin{hscode}\SaveRestoreHook
\column{B}{@{}>{\hspre}l<{\hspost}@{}}%
\column{5}{@{}>{\hspre}l<{\hspost}@{}}%
\column{E}{@{}>{\hspre}l<{\hspost}@{}}%
\>[B]{}\mathbf{class}\;\Conid{Profunctor}\;\Varid{p}\Rightarrow \Conid{Cartesian}\;\Varid{p}\;\mathbf{where}{}\<[E]%
\\
\>[B]{}\hsindent{5}{}\<[5]%
\>[5]{}\Varid{second}\mathbin{::}\Varid{p}\;\Varid{a}\;\Varid{b}\to \Varid{p}\;(\Varid{c},\Varid{a})\;(\Varid{c},\Varid{b}){}\<[E]%
\ColumnHook
\end{hscode}\resethooks

\ensuremath{\Varid{second}} must abide by some laws that we do not detail here for brevity.

\ensuremath{\Conid{Cartesian}\;\Varid{p}} expresses that \ensuremath{\Varid{p}\;\Varid{a}\;\Varid{b}} is not only a form of
transformation from \ensuremath{\Varid{a}} to \ensuremath{\Varid{b}} (a profunctor), but also that such a
transformation can be ``lifted'' to pass an additional value of type \ensuremath{\Varid{c}}
along.

Those new lenses can be proved to be isomorphic to the previous
\emph{concrete} lenses, provided we assume very-well-behavedness laws on
the concrete lenses. For a proof of this equivalence, see Pickering et
al.\autocite{profunctor_optics}.

The wonderful property of these new lenses is that the composition
operator is now simply arrow composition instead of the more complicated
\ensuremath{\Varid{composeLens}} function.

\medskip

In this new encoding, the \ensuremath{\Varid{first}} lens is defined as follows:

\begin{hscode}\SaveRestoreHook
\column{B}{@{}>{\hspre}l<{\hspost}@{}}%
\column{5}{@{}>{\hspre}l<{\hspost}@{}}%
\column{E}{@{}>{\hspre}l<{\hspost}@{}}%
\>[B]{}\Varid{first}\mathbin{::}\Conid{PLens}\;\Varid{a}\;\Varid{b}\;(\Varid{a},\Varid{c})\;(\Varid{b},\Varid{c}){}\<[E]%
\\
\>[B]{}\Varid{first}\mathrel{=}\Varid{dimap}\;\Varid{swap}\;\Varid{swap}\hsdot{\circ }{.\;}\Varid{second}{}\<[E]%
\\
\>[B]{}\hsindent{5}{}\<[5]%
\>[5]{}\mathbf{where}\;\Varid{swap}\;(\Varid{x},\Varid{y})\mathrel{=}(\Varid{y},\Varid{x}){}\<[E]%
\ColumnHook
\end{hscode}\resethooks

We can compose \ensuremath{\Varid{first}} as a normal arrow:

\begin{hscode}\SaveRestoreHook
\column{B}{@{}>{\hspre}l<{\hspost}@{}}%
\column{E}{@{}>{\hspre}l<{\hspost}@{}}%
\>[B]{}\Varid{first}\hsdot{\circ }{.\;}\Varid{first}\mathbin{::}\Conid{PLens}\;\Varid{a}\;\Varid{b}\;((\Varid{a},\Varid{c}),\Varid{d})\;((\Varid{b},\Varid{c}),\Varid{d}){}\<[E]%
\ColumnHook
\end{hscode}\resethooks

\ensuremath{\Varid{second}} is also a valid lens. \bigskip

\begin{note}

\ensuremath{\Conid{Cartesian}} is often called \ensuremath{\Conid{Strong}} in profunctor optic libraries.

\end{note}

\hypertarget{prisms}{%
\subsubsection{Prisms}\label{prisms}}

Prisms have a similar encoding:

\begin{hscode}\SaveRestoreHook
\column{B}{@{}>{\hspre}l<{\hspost}@{}}%
\column{E}{@{}>{\hspre}l<{\hspost}@{}}%
\>[B]{}\mathbf{type}\;\Conid{PPrism}\;\Varid{a}\;\Varid{b}\;\Varid{s}\;\Varid{t}\mathrel{=}\mathbf{forall}\;\Varid{p}\hsforall \hsdot{\circ }{.\;}\Conid{Cocartesian}\;\Varid{p}\Rightarrow \Varid{p}\;\Varid{a}\;\Varid{b}\to \Varid{p}\;\Varid{s}\;\Varid{t}{}\<[E]%
\ColumnHook
\end{hscode}\resethooks

with

\begin{hscode}\SaveRestoreHook
\column{B}{@{}>{\hspre}l<{\hspost}@{}}%
\column{5}{@{}>{\hspre}l<{\hspost}@{}}%
\column{E}{@{}>{\hspre}l<{\hspost}@{}}%
\>[B]{}\mathbf{class}\;\Conid{Profunctor}\;\Varid{p}\Rightarrow \Conid{Cocartesian}\;\Varid{p}\;\mathbf{where}{}\<[E]%
\\
\>[B]{}\hsindent{5}{}\<[5]%
\>[5]{}\Varid{right}\mathbin{::}\Varid{p}\;\Varid{a}\;\Varid{b}\to \Varid{p}\;(\Varid{c}\mathbin{+}\Varid{a})\;(\Varid{c}\mathbin{+}\Varid{b}){}\<[E]%
\ColumnHook
\end{hscode}\resethooks

The \ensuremath{\Varid{just}} prism becomes:

\begin{hscode}\SaveRestoreHook
\column{B}{@{}>{\hspre}l<{\hspost}@{}}%
\column{5}{@{}>{\hspre}l<{\hspost}@{}}%
\column{7}{@{}>{\hspre}l<{\hspost}@{}}%
\column{E}{@{}>{\hspre}l<{\hspost}@{}}%
\>[B]{}\Varid{just}\mathbin{::}\Conid{PPrism}\;\Varid{a}\;\Varid{b}\;(\Conid{Maybe}\;\Varid{a})\;(\Conid{Maybe}\;\Varid{b}){}\<[E]%
\\
\>[B]{}\Varid{just}\mathrel{=}\Varid{dimap}\;\Varid{maybeToSum}\;\Varid{sumToMaybe}\hsdot{\circ }{.\;}\Varid{right}{}\<[E]%
\\
\>[B]{}\hsindent{5}{}\<[5]%
\>[5]{}\mathbf{where}{}\<[E]%
\\
\>[5]{}\hsindent{2}{}\<[7]%
\>[7]{}\Varid{sumToMaybe}\mathbin{::}()\mathbin{+}\Varid{a}\to \Conid{Maybe}\;\Varid{a}{}\<[E]%
\\
\>[5]{}\hsindent{2}{}\<[7]%
\>[7]{}\Varid{sumToMaybe}\mathrel{=}\Varid{either}\;(\Varid{const}\;\Conid{Nothing})\;\Conid{Just}{}\<[E]%
\\
\>[5]{}\hsindent{2}{}\<[7]%
\>[7]{}\Varid{maybeToSum}\mathbin{::}\Conid{Maybe}\;\Varid{a}\to ()\mathbin{+}\Varid{a}{}\<[E]%
\\
\>[5]{}\hsindent{2}{}\<[7]%
\>[7]{}\Varid{maybeToSum}\mathrel{=}\Varid{maybe}\;(\Conid{Left}\;())\;\Conid{Right}{}\<[E]%
\ColumnHook
\end{hscode}\resethooks

\hypertarget{operators}{%
\subsubsection{Operators}\label{operators}}

To define operators on a profunctor-encoded optic, one needs to
instantiate the optic to a particular profunctor.

For example, to define operators on profunctor lenses, we need to find
an appropriate \ensuremath{\Conid{Cartesian}} profunctor to instantiate it to.

\ensuremath{\Varid{get}} is defined as follows:

\begin{hscode}\SaveRestoreHook
\column{B}{@{}>{\hspre}l<{\hspost}@{}}%
\column{E}{@{}>{\hspre}l<{\hspost}@{}}%
\>[B]{}\mathbf{data}\;\Conid{Getting}\;\Varid{a}\;\Varid{b}\;\Varid{s}\;\Varid{t}\mathrel{=}\Conid{Getting}\;\{\mskip1.5mu \Varid{unGetting}\mathbin{::}\Varid{s}\to \Varid{a}\mskip1.5mu\}{}\<[E]%
\ColumnHook
\end{hscode}\resethooks

\smallskip

\begin{hscode}\SaveRestoreHook
\column{B}{@{}>{\hspre}l<{\hspost}@{}}%
\column{5}{@{}>{\hspre}l<{\hspost}@{}}%
\column{E}{@{}>{\hspre}l<{\hspost}@{}}%
\>[B]{}\mathbf{instance}\;\Conid{Profunctor}\;(\Conid{Getting}\;\Varid{a}\;\Varid{b})\;\mathbf{where}{}\<[E]%
\\
\>[B]{}\hsindent{5}{}\<[5]%
\>[5]{}\Varid{dimap}\;\Varid{f}\;\Varid{g}\;(\Conid{Getting}\;\Varid{h})\mathrel{=}\Conid{Getting}\;(\Varid{h}\hsdot{\circ }{.\;}\Varid{f}){}\<[E]%
\\[\blanklineskip]%
\>[B]{}\mathbf{instance}\;\Conid{Cartesian}\;(\Conid{Getting}\;\Varid{a}\;\Varid{b})\;\mathbf{where}{}\<[E]%
\\
\>[B]{}\hsindent{5}{}\<[5]%
\>[5]{}\Varid{second}\;(\Conid{Getting}\;\Varid{h})\mathrel{=}\Conid{Getting}\;(\Varid{h}\hsdot{\circ }{.\;}\Varid{snd}){}\<[E]%
\ColumnHook
\end{hscode}\resethooks

\smallskip

\begin{hscode}\SaveRestoreHook
\column{B}{@{}>{\hspre}l<{\hspost}@{}}%
\column{E}{@{}>{\hspre}l<{\hspost}@{}}%
\>[B]{}\Varid{idGetting}\mathbin{::}\Conid{Getting}\;\Varid{a}\;\Varid{b}\;\Varid{a}\;\Varid{b}{}\<[E]%
\\
\>[B]{}\Varid{idGetting}\mathrel{=}\Conid{Getting}\;\Varid{id}{}\<[E]%
\ColumnHook
\end{hscode}\resethooks

\smallskip

\begin{hscode}\SaveRestoreHook
\column{B}{@{}>{\hspre}l<{\hspost}@{}}%
\column{E}{@{}>{\hspre}l<{\hspost}@{}}%
\>[B]{}\Varid{get}\mathbin{::}(\Conid{Getting}\;\Varid{a}\;\Varid{b}\;\Varid{a}\;\Varid{b}\to \Conid{Getting}\;\Varid{a}\;\Varid{b}\;\Varid{s}\;\Varid{t})\to \Varid{s}\to \Varid{a}{}\<[E]%
\\
\>[B]{}\Varid{get}\;\Varid{l}\mathrel{=}\Varid{unGetting}\;(\Varid{l}\;\Varid{idGetting}){}\<[E]%
\ColumnHook
\end{hscode}\resethooks

Then

\begin{hscode}\SaveRestoreHook
\column{B}{@{}>{\hspre}l<{\hspost}@{}}%
\column{E}{@{}>{\hspre}l<{\hspost}@{}}%
\>[B]{}\Varid{get}\;\Varid{first}{}\<[E]%
\\
\>[B]{}\mathrel{=}\Varid{unGetting}\;(\Varid{dimap}\;\Varid{swap}\;\Varid{swap}\;(\Varid{second}\;\Varid{idGetting})){}\<[E]%
\\
\>[B]{}\mathrel{=}\Varid{unGetting}\;(\Varid{dimap}\;\Varid{swap}\;\Varid{swap}\;(\Conid{Getting}\;\Varid{snd})){}\<[E]%
\\
\>[B]{}\mathrel{=}\Varid{unGetting}\;(\Conid{Getting}\;(\Varid{swap}\hsdot{\circ }{.\;}\Varid{snd}\hsdot{\circ }{.\;}\Varid{swap})){}\<[E]%
\\
\>[B]{}\mathrel{=}\Varid{swap}\hsdot{\circ }{.\;}\Varid{snd}\hsdot{\circ }{.\;}\Varid{swap}{}\<[E]%
\\
\>[B]{}\mathrel{=}\lambda (\Varid{a},\Varid{c})\to \Varid{a}{}\<[E]%
\ColumnHook
\end{hscode}\resethooks

as expected. \bigskip

\begin{rmk}\label{profoptics_are_complicated}

We can already see that despite the nice compositional properties of
profunctor lenses, it is more difficult to reason about their properties
than with concrete lenses.

\end{rmk}

\bigskip
\bigskip

Similarly, we can define the \ensuremath{\Varid{match}} function for prisms:

\begin{hscode}\SaveRestoreHook
\column{B}{@{}>{\hspre}l<{\hspost}@{}}%
\column{E}{@{}>{\hspre}l<{\hspost}@{}}%
\>[B]{}\mathbf{data}\;\Conid{Matching}\;\Varid{a}\;\Varid{b}\;\Varid{s}\;\Varid{t}\mathrel{=}\Conid{Matching}\;\{\mskip1.5mu \Varid{unMatching}\mathbin{::}\Varid{s}\to \Varid{t}\mathbin{+}\Varid{a}\mskip1.5mu\}{}\<[E]%
\ColumnHook
\end{hscode}\resethooks

\smallskip

\begin{hscode}\SaveRestoreHook
\column{B}{@{}>{\hspre}l<{\hspost}@{}}%
\column{5}{@{}>{\hspre}l<{\hspost}@{}}%
\column{9}{@{}>{\hspre}l<{\hspost}@{}}%
\column{16}{@{}>{\hspre}l<{\hspost}@{}}%
\column{E}{@{}>{\hspre}l<{\hspost}@{}}%
\>[B]{}\mathbf{instance}\;\Conid{Profunctor}\;(\Conid{Matching}\;\Varid{a}\;\Varid{b})\;\mathbf{where}{}\<[E]%
\\
\>[B]{}\hsindent{5}{}\<[5]%
\>[5]{}\Varid{dimap}\;\Varid{f}\;\Varid{g}\;(\Conid{Matching}\;\Varid{h})\mathrel{=}\Conid{Matching}\;(\Varid{left}\;\Varid{g}\hsdot{\circ }{.\;}\Varid{h}\hsdot{\circ }{.\;}\Varid{f}){}\<[E]%
\\[\blanklineskip]%
\>[B]{}\mathbf{instance}\;\Conid{Cocartesian}\;(\Conid{Matching}\;\Varid{a}\;\Varid{b})\;\mathbf{where}{}\<[E]%
\\
\>[B]{}\hsindent{5}{}\<[5]%
\>[5]{}\Varid{right}\;(\Conid{Matching}\;\Varid{h})\mathrel{=}\Conid{Matching}\;(\Varid{assoc}\hsdot{\circ }{.\;}\Varid{fmap}\;\Varid{h}){}\<[E]%
\\
\>[5]{}\hsindent{4}{}\<[9]%
\>[9]{}\mathbf{where}\;{}\<[16]%
\>[16]{}\Varid{assoc}\mathbin{::}\Varid{c}\mathbin{+}(\Varid{t}\mathbin{+}\Varid{a})\to (\Varid{c}\mathbin{+}\Varid{t})\mathbin{+}\Varid{a}{}\<[E]%
\\
\>[16]{}\Varid{assoc}\;(\Conid{Left}\;\Varid{c})\mathrel{=}\Conid{Left}\;(\Conid{Left}\;\Varid{c}){}\<[E]%
\\
\>[16]{}\Varid{assoc}\;(\Conid{Right}\;(\Conid{Left}\;\Varid{t}))\mathrel{=}\Conid{Left}\;(\Conid{Right}\;\Varid{t}){}\<[E]%
\\
\>[16]{}\Varid{assoc}\;(\Conid{Right}\;(\Conid{Right}\;\Varid{a}))\mathrel{=}\Conid{Right}\;\Varid{a}{}\<[E]%
\ColumnHook
\end{hscode}\resethooks

\smallskip

\begin{hscode}\SaveRestoreHook
\column{B}{@{}>{\hspre}l<{\hspost}@{}}%
\column{E}{@{}>{\hspre}l<{\hspost}@{}}%
\>[B]{}\Varid{idMatching}\mathbin{::}\Conid{Matching}\;\Varid{a}\;\Varid{b}\;\Varid{a}\;\Varid{b}{}\<[E]%
\\
\>[B]{}\Varid{idMatching}\mathrel{=}\Conid{Matching}\;\Conid{Left}{}\<[E]%
\\[\blanklineskip]%
\>[B]{}\Varid{match}\mathbin{::}(\Conid{Matching}\;\Varid{a}\;\Varid{b}\;\Varid{a}\;\Varid{b}\to \Conid{Matching}\;\Varid{a}\;\Varid{b}\;\Varid{s}\;\Varid{t})\to \Varid{s}\to \Varid{t}\mathbin{+}\Varid{a}{}\<[E]%
\\
\>[B]{}\Varid{match}\;\Varid{p}\mathrel{=}\Varid{unMatching}\;(\Varid{p}\;\Varid{idMatching}){}\<[E]%
\ColumnHook
\end{hscode}\resethooks

\hypertarget{cross-family-composition}{%
\subsection{Cross-family composition}\label{cross-family-composition}}

Since both prisms and lenses are now encoded as similar-looking arrows,
we can try composing them:

Recall we had the following optics:

\begin{hscode}\SaveRestoreHook
\column{B}{@{}>{\hspre}l<{\hspost}@{}}%
\column{E}{@{}>{\hspre}l<{\hspost}@{}}%
\>[B]{}\Varid{second}\mathbin{::}\mathbf{forall}\;\Varid{p}\hsforall \hsdot{\circ }{.\;}\Conid{Cartesian}\;\Varid{p}\Rightarrow \Varid{p}\;\Varid{a}\;\Varid{b}\to \Varid{p}\;(\Varid{c},\Varid{a})\;(\Varid{c},\Varid{b}){}\<[E]%
\\
\>[B]{}\Varid{just}\mathbin{::}\mathbf{forall}\;\Varid{p}\hsforall \hsdot{\circ }{.\;}\Conid{Cocartesian}\;\Varid{p}\Rightarrow \Varid{p}\;\Varid{a}\;\Varid{b}\to \Varid{p}\;(\Conid{Maybe}\;\Varid{a})\;(\Conid{Maybe}\;\Varid{b}){}\<[E]%
\ColumnHook
\end{hscode}\resethooks

then

\begin{hscode}\SaveRestoreHook
\column{B}{@{}>{\hspre}l<{\hspost}@{}}%
\column{13}{@{}>{\hspre}l<{\hspost}@{}}%
\column{E}{@{}>{\hspre}l<{\hspost}@{}}%
\>[B]{}\Varid{second}\hsdot{\circ }{.\;}\Varid{just}\mathbin{::}\mathbf{forall}\;\Varid{p}\hsforall \hsdot{\circ }{.\;}(\Conid{Cartesian}\;\Varid{p},\Conid{Cocartesian}\;\Varid{p})\Rightarrow {}\<[E]%
\\
\>[B]{}\hsindent{13}{}\<[13]%
\>[13]{}\Varid{p}\;\Varid{a}\;\Varid{b}\to \Varid{p}\;(\Varid{c},\Conid{Maybe}\;\Varid{a})\;(\Varid{c},\Conid{Maybe}\;\Varid{b}){}\<[E]%
\ColumnHook
\end{hscode}\resethooks

This is neither a lens nor a prism, we have indeed just defined a new
optic family. These are usually called \emph{optionals}, \emph{affine
traversals} or \emph{partial lenses}.

\begin{hscode}\SaveRestoreHook
\column{B}{@{}>{\hspre}l<{\hspost}@{}}%
\column{E}{@{}>{\hspre}l<{\hspost}@{}}%
\>[B]{}\mathbf{type}\;\Conid{POptional}\;\Varid{a}\;\Varid{b}\;\Varid{s}\;\Varid{t}\mathrel{=}\mathbf{forall}\;\Varid{p}\hsforall \hsdot{\circ }{.\;}(\Conid{Cartesian}\;\Varid{p},\Conid{Cocartesian}\;\Varid{p})\Rightarrow \Varid{p}\;\Varid{a}\;\Varid{b}\to \Varid{p}\;\Varid{s}\;\Varid{t}{}\<[E]%
\ColumnHook
\end{hscode}\resethooks

\bigskip

To define an operator on optionals, we need a profunctor that is both
\ensuremath{\Conid{Cartesian}} and \ensuremath{\Conid{Cocartesian}}. Interestingly, \ensuremath{\Conid{Matching}\;\Varid{a}\;\Varid{b}} is one such
profunctor:

\begin{hscode}\SaveRestoreHook
\column{B}{@{}>{\hspre}l<{\hspost}@{}}%
\column{5}{@{}>{\hspre}l<{\hspost}@{}}%
\column{9}{@{}>{\hspre}l<{\hspost}@{}}%
\column{16}{@{}>{\hspre}l<{\hspost}@{}}%
\column{E}{@{}>{\hspre}l<{\hspost}@{}}%
\>[B]{}\mathbf{instance}\;\Conid{Cartesian}\;(\Conid{Matching}\;\Varid{a}\;\Varid{b})\;\mathbf{where}{}\<[E]%
\\
\>[B]{}\hsindent{5}{}\<[5]%
\>[5]{}\Varid{second}\;(\Conid{Matching}\;\Varid{h})\mathrel{=}\Conid{Matching}\;(\Varid{dist}\hsdot{\circ }{.\;}\Varid{fmap}\;\Varid{h}){}\<[E]%
\\
\>[5]{}\hsindent{4}{}\<[9]%
\>[9]{}\mathbf{where}\;{}\<[16]%
\>[16]{}\Varid{dist}\mathbin{::}\Varid{c}\mathbin{\mathbf{\times}}(\Varid{t}\mathbin{+}\Varid{a})\to (\Varid{c}\mathbin{\mathbf{\times}}\Varid{t})\mathbin{+}\Varid{a}{}\<[E]%
\\
\>[16]{}\Varid{dist}\;(\Varid{c},\Conid{Left}\;\Varid{t})\mathrel{=}\Conid{Left}\;(\Varid{c},\Varid{t}){}\<[E]%
\\
\>[16]{}\Varid{dist}\;(\Varid{c},\Conid{Right}\;\Varid{a})\mathrel{=}\Conid{Right}\;\Varid{a}{}\<[E]%
\ColumnHook
\end{hscode}\resethooks

We can then use \ensuremath{\Varid{match}} on optionals too.

\hypertarget{optic-libraries}{%
\subsection{Optic libraries}\label{optic-libraries}}

Optic libraries define several optic families and a plethora of useful
combinators.

The most popular optic library is the Haskell
\text{\ttfamily lens}\autocite{kmett/lens} package, by E. Kmett. It does not define
profunctor optics however.

The two most developed profunctor optic libraries are
\text{\ttfamily mezzolens}\autocite{mezzolens} and
\texttt{purescript\-profunctor-lenses}\autocite{purescript_profunctor_lenses}.

For an extensive reference on the existing profunctor optics, their
operators and their derivation, see \autocite{oleg_glassery}.

\hypertarget{haskell}{%
\section{Haskell}\label{haskell}}

We use an idealized version of Haskell to convey the notions and
derivations of this paper. Familiarity with Haskell and some of its more
advanced features (typeclasses, universal/existential quantification,
rank-n types, \ldots) is expected.

A lot of results of this paper come from parametricity
theorems\autocite{wadler_parametricity}. For a nice introduction to
parametricity theorems, see for example
\autocite{welltyped_parametricity} or \autocite{bartosz_parametricity}.

To make the derivations clearer, we will omit a lot of the
wrapping/unwrapping that would be necessary for our typeclass instances
to not overlap. We will ignore possible ambiguous definitions and
overlapping instances when it is clear which one we are referring to.

This notably means that a lot of the code presented does not compile. A
version of the presented concepts that does compile is available in
appendix \labelcref{working-implementation}. \bigskip

\begin{notation}

We will sometimes write ``\ensuremath{\Varid{f}\;\mathbin{\in}\;\Conid{Functor}}'' in place of ``\ensuremath{\Varid{f}} is an
instance of \ensuremath{\Conid{Functor}}''.

\end{notation}

\hypertarget{layout-and-aims}{%
\section{Layout and aims}\label{layout-and-aims}}

Profunctor optics have very nice compositional properties, but have one
major drawback: it is hard to relate them to concrete optics.

In particular, given a profunctor-encoded optic, we have observed
(\textcite{profoptics_are_complicated}) that it is not always
straightforward to determine if it obeys some desired law. Operators
tend to be a mess of wrapping and unwrapping of the relevant profunctors
which can make checking even simple laws quite cumbersome.

Another difficulty arises when trying to define a profunctor encoding
for a new optic family. Indeed, the link between the record definition
of lenses or prisms and their profunctor encoding is far from obvious.
\bigskip

This thesis aims at solving this problem by providing general theorems
about the structure of profunctor optics that make it easier to study
their properties and derive new instances.

In \textcite{profunctor-optics-generically}, we provide general
definitions of optic families and profunctor encodings.

In \textcite{isomorphism-optics}, we present a new encoding of optics
called \emph{isomorphism optics}, that are easier to reason about than
profunctor optics, and we derive some of their major properties.

In \textcite{two-theorems-about-profunctor-optics}, we study profunctor
encodings generically and derive two important theorems about profunctor
optics.

Finally, in
\textcite{tying-it-all-up-a-new-look-at-profunctor-encodings}, using the
two theorems we redefine some common optic families in our new framework
and derive some of their properties. We then present a case study of how
to derive a new profunctor encoding.

\hypertarget{profunctor-optics-generically}{%
\chapter{Profunctor optics,
generically}\label{profunctor-optics-generically}}

In this chapter, we define precisely the notions of optic families,
profunctor optics, and profunctor encodings between them.

\hypertarget{optic-families}{%
\section{Optic families}\label{optic-families}}

\hypertarget{definition}{%
\subsection{Definition}\label{definition}}

We have seen that \ensuremath{\Conid{Lens}} and \ensuremath{\Conid{Prism}} are instances of what we call
``optic families''. More generally, we call \emph{optic family} a type
constructor \ensuremath{\Varid{op}} that takes 4 type parameters and has the following
properties:

First of all, the characteristic property of optics is that they are
composable. Therefore \ensuremath{\Varid{op}} needs an operator to compose optics:

\begin{hscode}\SaveRestoreHook
\column{B}{@{}>{\hspre}l<{\hspost}@{}}%
\column{E}{@{}>{\hspre}l<{\hspost}@{}}%
\>[B]{}(\circ_{op})\mathbin{::}\Varid{op}\;\Varid{a}\;\Varid{b}\;\Varid{s}\;\Varid{t}\to \Varid{op}\;\Varid{x}\;\Varid{y}\;\Varid{a}\;\Varid{b}\to \Varid{op}\;\Varid{x}\;\Varid{y}\;\Varid{s}\;\Varid{t}{}\<[E]%
\ColumnHook
\end{hscode}\resethooks

We also require this composition operator to have an identity:

\begin{hscode}\SaveRestoreHook
\column{B}{@{}>{\hspre}l<{\hspost}@{}}%
\column{E}{@{}>{\hspre}l<{\hspost}@{}}%
\>[B]{}\Varid{id}_{\Varid{op}}\mathbin{::}\Varid{op}\;\Varid{a}\;\Varid{b}\;\Varid{a}\;\Varid{b}{}\<[E]%
\ColumnHook
\end{hscode}\resethooks

Together, \ensuremath{(\circ_{op})} and \ensuremath{\Varid{id}_{\Varid{op}}} should verify the usual axioms of
composition, i.e.~associativity and right/left identity.

\ensuremath{\Varid{op}} should contain at least the trivial transformations, i.e.~pairs of
arrows:

\begin{hscode}\SaveRestoreHook
\column{B}{@{}>{\hspre}l<{\hspost}@{}}%
\column{E}{@{}>{\hspre}l<{\hspost}@{}}%
\>[B]{}\Varid{injOptic}\mathbin{::}(\Varid{s}\to \Varid{a})\to (\Varid{b}\to \Varid{t})\to \Varid{op}\;\Varid{a}\;\Varid{b}\;\Varid{s}\;\Varid{t}{}\<[E]%
\ColumnHook
\end{hscode}\resethooks

We require \ensuremath{\Varid{injOptic}} to be compatible with composition in the following
sense:

\begin{itemize}
\tightlist
\item
  \ensuremath{\Varid{injOptic}\;\Varid{id}\;\Varid{id}\mathrel{=}\Varid{id}_{\Varid{op}}}
\item
  \ensuremath{\Varid{injOptic}\;(\Varid{f'}\hsdot{\circ }{.\;}\Varid{f})\;(\Varid{g}\hsdot{\circ }{.\;}\Varid{g'})\mathrel{=}\Varid{injOptic}\;\Varid{f}\;\Varid{g}\ \circ_{op}\ \Varid{injOptic}\;\Varid{f'}\;\Varid{g'}}
\end{itemize}

Finally, we require \ensuremath{\Varid{op}\;\Varid{a}\;\Varid{b}\;\Varid{s}\;\Varid{t}} to at least be able to transform an
\ensuremath{\Varid{a}\to \Varid{b}} arrow into an \ensuremath{\Varid{s}\to \Varid{t}} arrow.

\begin{hscode}\SaveRestoreHook
\column{B}{@{}>{\hspre}l<{\hspost}@{}}%
\column{E}{@{}>{\hspre}l<{\hspost}@{}}%
\>[B]{}\Varid{mapOptic}\mathbin{::}\Varid{op}\;\Varid{a}\;\Varid{b}\;\Varid{s}\;\Varid{t}\to (\Varid{a}\to \Varid{b})\to (\Varid{s}\to \Varid{t}){}\<[E]%
\ColumnHook
\end{hscode}\resethooks

\ensuremath{\Varid{mapOptic}} should abide by the following laws:

\begin{itemize}
\tightlist
\item
  \ensuremath{\Varid{mapOptic}\;(\Varid{injOptic}\;\Varid{f}\;\Varid{g})\;\Varid{h}\mathrel{=}\Varid{g}\hsdot{\circ }{.\;}\Varid{h}\hsdot{\circ }{.\;}\Varid{f}}
\item
  \ensuremath{\Varid{mapOptic}\;(\Varid{op}_{\mathrm{1}}\ \circ_{op}\ \Varid{op}_{\mathrm{2}})\mathrel{=}\Varid{mapOptic}\;\Varid{op}_{\mathrm{1}}\hsdot{\circ }{.\;}\Varid{mapOptic}\;\Varid{op}_{\mathrm{2}}}
\end{itemize}

\bigskip

We bundle these requirements in a Haskell typeclass: \bigskip

\begin{defn}[Optic family]

An optic family is a 4-parameter type constructor \ensuremath{\Varid{op}} instance of the
following typeclass:

\begin{hscode}\SaveRestoreHook
\column{B}{@{}>{\hspre}l<{\hspost}@{}}%
\column{5}{@{}>{\hspre}l<{\hspost}@{}}%
\column{E}{@{}>{\hspre}l<{\hspost}@{}}%
\>[B]{}\mathbf{class}\;\Conid{OpticFamily}\;(\Varid{op}\mathbin{::}\mathbin{*}\to \mathbin{*}\to \mathbin{*}\to \mathbin{*}\to \mathbin{*})\;\mathbf{where}{}\<[E]%
\\
\>[B]{}\hsindent{5}{}\<[5]%
\>[5]{}(\circ_{op})\mathbin{::}\Varid{op}\;\Varid{a}\;\Varid{b}\;\Varid{s}\;\Varid{t}\to \Varid{op}\;\Varid{x}\;\Varid{y}\;\Varid{a}\;\Varid{b}\to \Varid{op}\;\Varid{x}\;\Varid{y}\;\Varid{s}\;\Varid{t}{}\<[E]%
\\
\>[B]{}\hsindent{5}{}\<[5]%
\>[5]{}\Varid{id}_{\Varid{op}}\mathbin{::}\Varid{op}\;\Varid{a}\;\Varid{b}\;\Varid{a}\;\Varid{b}{}\<[E]%
\\
\>[B]{}\hsindent{5}{}\<[5]%
\>[5]{}\Varid{injOptic}\mathbin{::}(\Varid{s}\to \Varid{a})\to (\Varid{b}\to \Varid{t})\to \Varid{op}\;\Varid{a}\;\Varid{b}\;\Varid{s}\;\Varid{t}{}\<[E]%
\\
\>[B]{}\hsindent{5}{}\<[5]%
\>[5]{}\Varid{mapOptic}\mathbin{::}\Varid{op}\;\Varid{a}\;\Varid{b}\;\Varid{s}\;\Varid{t}\to (\Varid{a}\to \Varid{b})\to (\Varid{s}\to \Varid{t}){}\<[E]%
\ColumnHook
\end{hscode}\resethooks

that verifies the axioms stated above.

\end{defn}

\bigskip

\begin{rmk}

In a precise sense, an optic family forms a category which acts on pairs
of types, similar to how the usual definition from \ensuremath{\Conid{\Conid{Control}.Category}}
acts on single types. \ensuremath{\Varid{injOptic}} then describes a (bijective on objects)
functor from \ensuremath{\Conid{Hask}^{\Varid{op}}\mathbin{\mathbf{\times}}\Conid{Hask}} to \ensuremath{\Varid{op}}.

\end{rmk}

\bigskip

\begin{rmk}\label{omit_idoptic}

Using the first \ensuremath{\Varid{injOptic}} axiom, we can provide the implementation of
\ensuremath{\Varid{id}_{\Varid{op}}} generically:

\begin{hscode}\SaveRestoreHook
\column{B}{@{}>{\hspre}l<{\hspost}@{}}%
\column{E}{@{}>{\hspre}l<{\hspost}@{}}%
\>[B]{}\Varid{id}_{\Varid{op}}\mathrel{=}\Varid{injOptic}\;\Varid{id}\;\Varid{id}{}\<[E]%
\ColumnHook
\end{hscode}\resethooks

We will therefore omit implementations of \ensuremath{\Varid{id}_{\Varid{op}}}.

\end{rmk}

\bigskip

\begin{defn}[Morphism of optic families]

Given two optic families \ensuremath{\Varid{op}} and \ensuremath{\Varid{op'}}, an arrow
\ensuremath{\theta\mathbin{::}\mathbf{forall}\;\Varid{a}\hsforall \;\Varid{b}\;\Varid{s}\;\Varid{t}\hsdot{\circ }{.\;}\Varid{op}\;\Varid{a}\;\Varid{b}\;\Varid{s}\;\Varid{t}\to \Varid{op'}\;\Varid{a}\;\Varid{b}\;\Varid{s}\;\Varid{t}} is a morphism of
optic families iff it respects the structure of the optic families in
the following sense:

\begin{hscode}\SaveRestoreHook
\column{B}{@{}>{\hspre}l<{\hspost}@{}}%
\column{E}{@{}>{\hspre}l<{\hspost}@{}}%
\>[B]{}\theta\;(\Varid{injOptic}\;\Varid{f}\;\Varid{g})\mathrel{=}\Varid{injOptic}\;\Varid{f}\;\Varid{g}{}\<[E]%
\\
\>[B]{}\theta\;(\Varid{op}_{\mathrm{1}}\ \circ_{op}\ \Varid{op}_{\mathrm{2}})\mathrel{=}\theta\;\Varid{op}_{\mathrm{1}}\ \circ_{op}\ \theta\;\Varid{op}_{\mathrm{2}}{}\<[E]%
\\
\>[B]{}\Varid{mapOptic}\;(\theta\;\Varid{op})\mathrel{=}\Varid{mapOptic}\;\Varid{op}{}\<[E]%
\ColumnHook
\end{hscode}\resethooks

\end{defn}

\bigskip

\begin{notation}

We note morphism of optic families with a squiggly arrow:
\ensuremath{\theta\mathbin{::}\Varid{op}\leadsto \Varid{op'}}.

\end{notation}

\bigskip

\begin{rmk}\label{discuss_opfam}

The choice of the methods of the \ensuremath{\Conid{OpticFamily}} typeclass may seem
arbitrary, but it has been carefully made to be general enough to
encompass all existing optic families and still allow interesting
theorems.

The \ensuremath{\Varid{injOptic}} requirement forces optics to be functorial in their 4
arguments. An equivalent requirement would have been to include
\ensuremath{\Varid{multiMapOptic}} from the proposition below instead, along with relevant
axioms.

\ensuremath{\Varid{mapOptic}} was introduced to have a notion of well-behavedness for
\ensuremath{\Varid{enhanceOp}} (see \textcite{isomorphism-optics}). The requirement that
\ensuremath{\Varid{mapOptic}} be preserved by morphisms between optic families is quite
strong and implies surprising properties like
\textcite{endo_iso_is_singleton}, but it also enables the second
important result of this paper (\textcite{functorization_theorem}). It
is not too strong either since all morphisms studied respect it.

\end{rmk}

\bigskip

\begin{rmk}\label{opfam_and_lawfulness}

Both \ensuremath{\Varid{injOptic}} and \ensuremath{\Varid{mapOptic}} seem to have a deep link with
well-behavedness laws: for all the optic families studied, if \ensuremath{\Varid{f}} and
\ensuremath{\Varid{g}} are mutual inverses, then \ensuremath{\Varid{injOptic}\;\Varid{f}\;\Varid{g}} is lawful; conversely, if
\ensuremath{\Varid{op}} is a lawful optic, then \ensuremath{\Varid{mapOptic}\;\Varid{op}\;\Varid{id}\mathrel{=}\Varid{id}} and
\ensuremath{\Varid{mapOptic}\;\Varid{op}\;(\Varid{f}\hsdot{\circ }{.\;}\Varid{g})\mathrel{=}\Varid{mapOptic}\;\Varid{op}\;\Varid{f}\hsdot{\circ }{.\;}\Varid{mapOptic}\;\Varid{op}\;\Varid{g}}. The determination
of a general definition for optics well-behavedness laws is still an
open problem.

\end{rmk}

\bigskip

\hypertarget{properties}{%
\subsection{Properties}\label{properties}}

\begin{prop}

Optic families are covariant in their second and fourth arguments, and
contravariant in the other two.

\end{prop}

\begin{proof}

We can write the following generalization of \ensuremath{\Varid{fmap}} and \ensuremath{\Varid{dimap}}:

\begin{hscode}\SaveRestoreHook
\column{B}{@{}>{\hspre}l<{\hspost}@{}}%
\column{9}{@{}>{\hspre}l<{\hspost}@{}}%
\column{17}{@{}>{\hspre}l<{\hspost}@{}}%
\column{E}{@{}>{\hspre}l<{\hspost}@{}}%
\>[B]{}\Varid{multiMapOptic}\mathbin{::}\Conid{OpticFamily}\;\Varid{op}\Rightarrow (\Varid{a'}\to \Varid{a})\to (\Varid{b}\to \Varid{b'})\to (\Varid{s'}\to \Varid{s})\to (\Varid{t}\to \Varid{t'})\to {}\<[E]%
\\
\>[B]{}\hsindent{17}{}\<[17]%
\>[17]{}\Varid{op}\;\Varid{a}\;\Varid{b}\;\Varid{s}\;\Varid{t}\to \Varid{op}\;\Varid{a'}\;\Varid{b'}\;\Varid{s'}\;\Varid{t'}{}\<[E]%
\\
\>[B]{}\Varid{multiMapOptic}\;\Varid{fa}\;\Varid{fb}\;\Varid{fs}\;\Varid{ft}\;\Varid{op}\mathrel{=}{}\<[E]%
\\
\>[B]{}\hsindent{9}{}\<[9]%
\>[9]{}\Varid{injOptic}\;\Varid{fs}\;\Varid{ft}\ \circ_{op}\ \Varid{op}\ \circ_{op}\ \Varid{injOptic}\;\Varid{fa}\;\Varid{fb}{}\<[E]%
\ColumnHook
\end{hscode}\resethooks

The \ensuremath{\Conid{OpticFamily}} axioms easily imply adequate versions of the \ensuremath{\Varid{fmap}}
laws for this function.

\end{proof}

\begin{note}

In the rest of the paper, we will make use of the following
specialization:

\begin{hscode}\SaveRestoreHook
\column{B}{@{}>{\hspre}l<{\hspost}@{}}%
\column{E}{@{}>{\hspre}l<{\hspost}@{}}%
\>[B]{}\Varid{dimapOptic}\mathbin{::}\Conid{OpticFamily}\;\Varid{op}\Rightarrow (\Varid{s'}\to \Varid{s})\to (\Varid{t}\to \Varid{t'})\to \Varid{op}\;\Varid{a}\;\Varid{b}\;\Varid{s}\;\Varid{t}\to \Varid{op}\;\Varid{a}\;\Varid{b}\;\Varid{s'}\;\Varid{t'}{}\<[E]%
\\
\>[B]{}\mbox{\onelinecomment  \ensuremath{\Varid{dimapOptic}\;\Varid{f}\;\Varid{g}\mathrel{=}\Varid{multiMapOptic}\;\Varid{id}\;\Varid{id}\;\Varid{f}\;\Varid{g}}}{}\<[E]%
\\
\>[B]{}\Varid{dimapOptic}\;\Varid{f}\;\Varid{g}\;\Varid{op}\mathrel{=}\Varid{injOptic}\;\Varid{f}\;\Varid{g}\ \circ_{op}\ \Varid{op}{}\<[E]%
\ColumnHook
\end{hscode}\resethooks

\ensuremath{\Varid{dimapOptic}} verifies the \ensuremath{\Varid{dimap}} laws.

\end{note}

\bigskip

\begin{prop}\label{invertible_opfam_morphism}

If \ensuremath{\theta\mathbin{::}\Varid{op}\leadsto \Varid{op'}} has an inverse, then its inverse is a morphism
of optic families too.

\end{prop}

\begin{proof}

The \ensuremath{\Conid{OpticFamily}} methods can easily be seen to be transported through
\ensuremath{\theta^{-1}}:

\begin{hscode}\SaveRestoreHook
\column{B}{@{}>{\hspre}l<{\hspost}@{}}%
\column{E}{@{}>{\hspre}l<{\hspost}@{}}%
\>[B]{}\theta^{-1}\;(\Varid{injOptic}\;\Varid{f}\;\Varid{g}){}\<[E]%
\\
\>[B]{}\mathrel{=}\theta^{-1}\;(\theta\;(\Varid{injOptic}\;\Varid{f}\;\Varid{g})){}\<[E]%
\\
\>[B]{}\mathrel{=}\Varid{injOptic}\;\Varid{f}\;\Varid{g}{}\<[E]%
\\[\blanklineskip]%
\>[B]{}\theta^{-1}\;(\Varid{op}_{\mathrm{1}}\ \circ_{op}\ \Varid{op}_{\mathrm{2}}){}\<[E]%
\\
\>[B]{}\mathrel{=}\theta^{-1}\;(\theta\;(\theta^{-1}\;\Varid{op}_{\mathrm{1}})\ \circ_{op}\ \theta\;(\theta^{-1}\;\Varid{op}_{\mathrm{2}})){}\<[E]%
\\
\>[B]{}\mathrel{=}\theta^{-1}\;(\theta\;(\theta^{-1}\;\Varid{op}_{\mathrm{1}}\ \circ_{op}\ \theta^{-1}\;\Varid{op}_{\mathrm{2}})){}\<[E]%
\\
\>[B]{}\mathrel{=}\theta^{-1}\;\Varid{op}_{\mathrm{1}}\ \circ_{op}\ \theta^{-1}\;\Varid{op}_{\mathrm{2}}{}\<[E]%
\\[\blanklineskip]%
\>[B]{}\Varid{mapOptic}\;(\theta^{-1}\;\Varid{op}){}\<[E]%
\\
\>[B]{}\mathrel{=}\Varid{mapOptic}\;(\theta\;(\theta^{-1}\;\Varid{op})){}\<[E]%
\\
\>[B]{}\mathrel{=}\Varid{mapOptic}{}\<[E]%
\ColumnHook
\end{hscode}\resethooks

\end{proof}

\hypertarget{examples}{%
\subsection{Examples}\label{examples}}

\begin{prop}

Pairs of arrows form an optic family (dubbed
\emph{adapters}\autocite{profunctor_optics}, or
\emph{isos}\autocite{kmett/lens}):

\begin{hscode}\SaveRestoreHook
\column{B}{@{}>{\hspre}l<{\hspost}@{}}%
\column{E}{@{}>{\hspre}l<{\hspost}@{}}%
\>[B]{}\mathbf{data}\;\Conid{Adapter}\;\Varid{a}\;\Varid{b}\;\Varid{s}\;\Varid{t}\mathrel{=}\Conid{Adapter}\;(\Varid{s}\to \Varid{a})\;(\Varid{b}\to \Varid{t}){}\<[E]%
\ColumnHook
\end{hscode}\resethooks

\end{prop}

\begin{proof}

Let's implement the \ensuremath{\Conid{OpticFamily}} instance for \ensuremath{\Conid{Adapter}}:

\begin{hscode}\SaveRestoreHook
\column{B}{@{}>{\hspre}l<{\hspost}@{}}%
\column{5}{@{}>{\hspre}l<{\hspost}@{}}%
\column{E}{@{}>{\hspre}l<{\hspost}@{}}%
\>[B]{}\mathbf{instance}\;\Conid{OpticFamily}\;\Conid{Adapter}\;\mathbf{where}{}\<[E]%
\\
\>[B]{}\hsindent{5}{}\<[5]%
\>[5]{}(\Conid{Adapter}\;\Varid{f}\;\Varid{g})\ \circ_{op}\ (\Conid{Adapter}\;\Varid{f'}\;\Varid{g'})\mathrel{=}\Conid{Adapter}\;(\Varid{f'}\hsdot{\circ }{.\;}\Varid{f})\;(\Varid{g}\hsdot{\circ }{.\;}\Varid{g'}){}\<[E]%
\\
\>[B]{}\hsindent{5}{}\<[5]%
\>[5]{}\Varid{injOptic}\;\Varid{f}\;\Varid{g}\mathrel{=}\Conid{Adapter}\;\Varid{f}\;\Varid{g}{}\<[E]%
\\
\>[B]{}\hsindent{5}{}\<[5]%
\>[5]{}\Varid{mapOptic}\;(\Conid{Adapter}\;\Varid{f}\;\Varid{g})\;\Varid{h}\mathrel{=}\Varid{g}\hsdot{\circ }{.\;}\Varid{h}\hsdot{\circ }{.\;}\Varid{f}{}\<[E]%
\ColumnHook
\end{hscode}\resethooks

As per \textcite{omit_idoptic}, we can omit the \ensuremath{\Varid{id}_{\Varid{op}}}
implementation.

The laws follow immediately from the definition.

\end{proof}

\bigskip

\begin{prop}

Lenses form an optic family.

\end{prop}

\begin{proof}

Let's implement the \ensuremath{\Conid{OpticFamily}} instance for \ensuremath{\Conid{Lens}}:

\begin{hscode}\SaveRestoreHook
\column{B}{@{}>{\hspre}l<{\hspost}@{}}%
\column{E}{@{}>{\hspre}l<{\hspost}@{}}%
\>[B]{}\mathbf{instance}\;\Conid{OpticFamily}\;\Conid{Lens}\;\mathbf{where}{}\<[E]%
\ColumnHook
\end{hscode}\resethooks

\ensuremath{\Varid{injOptic}} embeds a trivial transformation into a lens. Its \ensuremath{\Varid{get}} is
straightforward: it uses the \ensuremath{\Varid{s}\to \Varid{a}} arrow provided; its \ensuremath{\Varid{put}} is
slightly more interesting: it discards the previous \ensuremath{\Varid{s}} value and
computes the result using \ensuremath{\Varid{b}} alone. We indeed do not have enough
structure on \ensuremath{\Varid{s}} to know how to put a new \ensuremath{\Varid{b}} inside it.

\begin{hscode}\SaveRestoreHook
\column{B}{@{}>{\hspre}l<{\hspost}@{}}%
\column{3}{@{}>{\hspre}l<{\hspost}@{}}%
\column{7}{@{}>{\hspre}l<{\hspost}@{}}%
\column{9}{@{}>{\hspre}l<{\hspost}@{}}%
\column{E}{@{}>{\hspre}l<{\hspost}@{}}%
\>[3]{}\Varid{injOptic}\mathbin{::}(\Varid{s}\to \Varid{a})\to (\Varid{b}\to \Varid{t})\to \Conid{Lens}\;\Varid{a}\;\Varid{b}\;\Varid{s}\;\Varid{t}{}\<[E]%
\\
\>[3]{}\Varid{injOptic}\;\Varid{f}\;\Varid{g}\mathrel{=}\Conid{Lens}\;\Varid{get}\;\Varid{put}{}\<[E]%
\\
\>[3]{}\hsindent{4}{}\<[7]%
\>[7]{}\mathbf{where}{}\<[E]%
\\
\>[7]{}\hsindent{2}{}\<[9]%
\>[9]{}\Varid{get}\mathbin{::}\Varid{s}\to \Varid{a}{}\<[E]%
\\
\>[7]{}\hsindent{2}{}\<[9]%
\>[9]{}\Varid{get}\mathrel{=}\Varid{f}{}\<[E]%
\\
\>[7]{}\hsindent{2}{}\<[9]%
\>[9]{}\Varid{put}\mathbin{::}\Varid{b}\to \Varid{s}\to \Varid{t}{}\<[E]%
\\
\>[7]{}\hsindent{2}{}\<[9]%
\>[9]{}\Varid{put}\;\Varid{b}\;\anonymous \mathrel{=}\Varid{g}\;\Varid{b}{}\<[E]%
\ColumnHook
\end{hscode}\resethooks

\ensuremath{(\circ_{op})} has exactly the same implementation as \ensuremath{\Varid{composeLens}} seen
in \textcite{simple-optics}.

\begin{hscode}\SaveRestoreHook
\column{B}{@{}>{\hspre}l<{\hspost}@{}}%
\column{3}{@{}>{\hspre}l<{\hspost}@{}}%
\column{5}{@{}>{\hspre}l<{\hspost}@{}}%
\column{7}{@{}>{\hspre}l<{\hspost}@{}}%
\column{E}{@{}>{\hspre}l<{\hspost}@{}}%
\>[3]{}(\circ_{op})\mathbin{::}\Conid{Lens}\;\Varid{a}\;\Varid{b}\;\Varid{s}\;\Varid{t}\to \Conid{Lens}\;\Varid{x}\;\Varid{y}\;\Varid{a}\;\Varid{b}\to \Conid{Lens}\;\Varid{x}\;\Varid{y}\;\Varid{s}\;\Varid{t}{}\<[E]%
\\
\>[3]{}(\circ_{op})\;(\Conid{Lens}\;\Varid{get}_{\mathrm{1}}\;\Varid{put}_{\mathrm{1}})\;(\Conid{Lens}\;\Varid{get}_{\mathrm{2}}\;\Varid{put}_{\mathrm{2}})\mathrel{=}\Conid{Lens}\;\Varid{get}\;\Varid{put}{}\<[E]%
\\
\>[3]{}\hsindent{2}{}\<[5]%
\>[5]{}\mathbf{where}{}\<[E]%
\\
\>[5]{}\hsindent{2}{}\<[7]%
\>[7]{}\Varid{get}\mathbin{::}\Varid{s}\to \Varid{x}{}\<[E]%
\\
\>[5]{}\hsindent{2}{}\<[7]%
\>[7]{}\Varid{get}\mathrel{=}\Varid{get}_{\mathrm{2}}\hsdot{\circ }{.\;}\Varid{get}_{\mathrm{1}}{}\<[E]%
\\
\>[5]{}\hsindent{2}{}\<[7]%
\>[7]{}\Varid{put}\mathbin{::}\Varid{y}\to \Varid{s}\to \Varid{t}{}\<[E]%
\\
\>[5]{}\hsindent{2}{}\<[7]%
\>[7]{}\Varid{put}\;\Varid{y}\;\Varid{s}\mathrel{=}\Varid{put}_{\mathrm{1}}\;(\Varid{put}_{\mathrm{2}}\;\Varid{y}\;(\Varid{get}_{\mathrm{1}}\;\Varid{s}))\;\Varid{s}{}\<[E]%
\ColumnHook
\end{hscode}\resethooks

\ensuremath{\Varid{mapOptic}} modifies the stored value by getting it, passing it through
the provided function, and putting it back.

\begin{hscode}\SaveRestoreHook
\column{B}{@{}>{\hspre}l<{\hspost}@{}}%
\column{3}{@{}>{\hspre}l<{\hspost}@{}}%
\column{E}{@{}>{\hspre}l<{\hspost}@{}}%
\>[3]{}\Varid{mapOptic}\mathbin{::}\Conid{Lens}\;\Varid{a}\;\Varid{b}\;\Varid{s}\;\Varid{t}\to (\Varid{a}\to \Varid{b})\to (\Varid{s}\to \Varid{t}){}\<[E]%
\\
\>[3]{}\Varid{mapOptic}\;(\Conid{Lens}\;\Varid{get}\;\Varid{put})\;\Varid{f}\;\Varid{s}\mathrel{=}\Varid{put}\;(\Varid{f}\;(\Varid{get}\;\Varid{s}))\;\Varid{s}{}\<[E]%
\ColumnHook
\end{hscode}\resethooks

The proof of the laws is rather mechanical so we omit it here.
Importantly, they do not require the lenses to be lawful lenses in order
to hold.

\end{proof}

\hypertarget{profunctor-optics-1}{%
\section{Profunctor optics}\label{profunctor-optics-1}}

In the rest of this paper, we will be manipulating Haskell typeclasses
as first-class values. This is enabled by the GHC \emph{ConstraintKinds}
extension, and works as follows:

Given a type that has a typeclass constraint, such as
\ensuremath{\Varid{fmap}\mathbin{::}\Conid{Functor}\;\Varid{f}\Rightarrow (\Varid{a}\to \Varid{b})\to (\Varid{f}\;\Varid{a}\to \Varid{f}\;\Varid{b})}, the object on the left
of the \ensuremath{\Rightarrow } arrow is called a constraint and has kind \ensuremath{\Conid{Constraint}}. A
typeclass is then a ``constraint constructor'' in the same sense that
\ensuremath{\Conid{Maybe}} is a type constructor: it takes some arguments and returns a
constraint.

For example, \ensuremath{\Conid{Functor}} has kind \ensuremath{(\mathbin{*}\to \mathbin{*})\to \Conid{Constraint}}, which means
that it is a typeclass that takes an object of kind \ensuremath{(\mathbin{*}\to \mathbin{*})} (i.e.~a
type constructor with 1 argument) as its argument.

Constraints and constraint constructors can be passed around in the same
way that types and type constructors can. They can then be used, as a
typeclass would be, on the left of the \ensuremath{\Rightarrow } arrow. See the definition of
\ensuremath{\Conid{ProfOptic}} below for an example. \bigskip \bigskip

\begin{defn}[Identity functor]

We call identity functor the following functor:

\begin{hscode}\SaveRestoreHook
\column{B}{@{}>{\hspre}l<{\hspost}@{}}%
\column{E}{@{}>{\hspre}l<{\hspost}@{}}%
\>[B]{}\mathbf{data}\;\Conid{Id}\;\Varid{a}\mathrel{=}\Conid{Id}\;\{\mskip1.5mu \Varid{unId}\mathbin{::}\Varid{a}\mskip1.5mu\}{}\<[E]%
\ColumnHook
\end{hscode}\resethooks

Here, \ensuremath{\Conid{Id}} is both the type constructor and the actual constructor, and
\ensuremath{\Varid{unId}\mathbin{::}\Conid{Id}\;\Varid{a}\to \Varid{a}} is the inverse of the \ensuremath{\Conid{Id}} constructor.

\begin{hscode}\SaveRestoreHook
\column{B}{@{}>{\hspre}l<{\hspost}@{}}%
\column{5}{@{}>{\hspre}l<{\hspost}@{}}%
\column{E}{@{}>{\hspre}l<{\hspost}@{}}%
\>[B]{}\mathbf{instance}\;\Conid{Functor}\;\Conid{Id}\;\mathbf{where}{}\<[E]%
\\
\>[B]{}\hsindent{5}{}\<[5]%
\>[5]{}\Varid{fmap}\mathbin{::}(\Varid{a}\to \Varid{b})\to (\Conid{Id}\;\Varid{a}\to \Conid{Id}\;\Varid{b}){}\<[E]%
\\
\>[B]{}\hsindent{5}{}\<[5]%
\>[5]{}\Varid{fmap}\;\Varid{f}\mathrel{=}\Conid{Id}\hsdot{\circ }{.\;}\Varid{f}\hsdot{\circ }{.\;}\Varid{unId}{}\<[E]%
\ColumnHook
\end{hscode}\resethooks

\end{defn}

\bigskip

\begin{defn}[Composed functor]

Given two functors \ensuremath{\Varid{f}} and \ensuremath{\Varid{g}}, \ensuremath{\Conid{Compose}\;\Varid{f}\;\Varid{g}} is the functor made from
their composition.

\begin{hscode}\SaveRestoreHook
\column{B}{@{}>{\hspre}l<{\hspost}@{}}%
\column{5}{@{}>{\hspre}l<{\hspost}@{}}%
\column{E}{@{}>{\hspre}l<{\hspost}@{}}%
\>[B]{}\mathbf{data}\;\Conid{Compose}\;\Varid{f}\;\Varid{g}\;\Varid{a}\mathrel{=}\Conid{Compose}\;\{\mskip1.5mu \Varid{unCompose}\mathbin{::}\Varid{f}\;(\Varid{g}\;\Varid{a})\mskip1.5mu\}{}\<[E]%
\\[\blanklineskip]%
\>[B]{}\mathbf{instance}\;(\Conid{Functor}\;\Varid{f},\Conid{Functor}\;\Varid{g})\Rightarrow \Conid{Functor}\;(\Conid{Compose}\;\Varid{f}\;\Varid{g})\;\mathbf{where}{}\<[E]%
\\
\>[B]{}\hsindent{5}{}\<[5]%
\>[5]{}\Varid{fmap}\mathbin{::}(\Varid{a}\to \Varid{b})\to (\Conid{Compose}\;\Varid{f}\;\Varid{g}\;\Varid{a}\to \Conid{Compose}\;\Varid{f}\;\Varid{g}\;\Varid{b}){}\<[E]%
\\
\>[B]{}\hsindent{5}{}\<[5]%
\>[5]{}\Varid{fmap}\;\Varid{f}\mathrel{=}\Conid{Compose}\hsdot{\circ }{.\;}\Varid{fmap}\;(\Varid{fmap}\;\Varid{f})\hsdot{\circ }{.\;}\Varid{unCompose}{}\<[E]%
\ColumnHook
\end{hscode}\resethooks

\end{defn}

\bigskip

\begin{defn}[Functor monoid]

A typeclass constraint \ensuremath{\sigma} of kind \ensuremath{(\mathbin{*}\to \mathbin{*})\to \Conid{Constraint}}
(i.e.~that takes a type constructor as its one argument) is said to be a
functor monoid if it has the following properties:

\begin{itemize}
\tightlist
\item
  every instance of \ensuremath{\sigma} is an instance of the \ensuremath{\Conid{Functor}} typeclass
\item
  the identity functor \ensuremath{\Conid{Id}} is an instance of \ensuremath{\sigma}
\item
  if \ensuremath{\Varid{f}} and \ensuremath{\Varid{g}} are instances of \ensuremath{\sigma}, then so is \ensuremath{\Conid{Compose}\;\Varid{f}\;\Varid{g}}
\end{itemize}

Since those properties cannot be cleanly captured in Haskell, we use a
dummy typeclass to indicate that those properties are verified. Like
typeclass laws, these cannot be enforced by the compiler and must be
checked by the programmer.

\begin{hscode}\SaveRestoreHook
\column{B}{@{}>{\hspre}l<{\hspost}@{}}%
\column{E}{@{}>{\hspre}l<{\hspost}@{}}%
\>[B]{}\mathbf{class}\;\Conid{FunctorMonoid}\;(\sigma\mathbin{::}(\mathbin{*}\to \mathbin{*})\to \Conid{Constraint}){}\<[E]%
\ColumnHook
\end{hscode}\resethooks

This in particular means that code that uses those properties will not
compile. See Appendix \labelcref{working-implementation} for equivalent
definitions that do compile using advanced typeclass machinery.

\end{defn}

\bigskip

\begin{expl}

\ensuremath{\Conid{Functor}} is a functor monoid (in fact, the most general one):

\begin{itemize}
\tightlist
\item
  by definition of the \ensuremath{\Conid{Id}} functor, \ensuremath{\Conid{Id}} is a functor;
\item
  by definition of the composed functor, \ensuremath{\Conid{Compose}\;\Varid{f}\;\Varid{g}} is a functor when
  \ensuremath{\Varid{f}} and \ensuremath{\Varid{g}} are too;
\item
  trivially, every functor is a functor.
\end{itemize}

\end{expl}

\bigskip

\begin{defn}[Profunctor optic]

Given a functor monoid \ensuremath{\sigma}, we call profunctor optic for \ensuremath{\sigma} a
function of the following type:

\begin{hscode}\SaveRestoreHook
\column{B}{@{}>{\hspre}l<{\hspost}@{}}%
\column{E}{@{}>{\hspre}l<{\hspost}@{}}%
\>[B]{}\mathbf{type}\;\Conid{ProfOptic}\;\sigma\;\Varid{a}\;\Varid{b}\;\Varid{s}\;\Varid{t}\mathrel{=}\mathbf{forall}\;\Varid{p}\hsforall \hsdot{\circ }{.\;}\Conid{Enhancing}\;\sigma\;\Varid{p}\Rightarrow \Varid{p}\;\Varid{a}\;\Varid{b}\to \Varid{p}\;\Varid{s}\;\Varid{t}{}\<[E]%
\ColumnHook
\end{hscode}\resethooks

where

\begin{hscode}\SaveRestoreHook
\column{B}{@{}>{\hspre}l<{\hspost}@{}}%
\column{5}{@{}>{\hspre}l<{\hspost}@{}}%
\column{E}{@{}>{\hspre}l<{\hspost}@{}}%
\>[B]{}\mathbf{class}\;(\Conid{FunctorMonoid}\;\sigma,\Conid{Profunctor}\;\Varid{p})\Rightarrow \Conid{Enhancing}\;\sigma\;\Varid{p}\;\mathbf{where}{}\<[E]%
\\
\>[B]{}\hsindent{5}{}\<[5]%
\>[5]{}\Varid{enhance}\mathbin{::}\mathbf{forall}\;\Varid{f}\hsforall \;\Varid{a}\;\Varid{b}\hsdot{\circ }{.\;}\sigma\;\Varid{f}\Rightarrow \Varid{p}\;\Varid{a}\;\Varid{b}\to \Varid{p}\;(\Varid{f}\;\Varid{a})\;(\Varid{f}\;\Varid{b}){}\<[E]%
\ColumnHook
\end{hscode}\resethooks

\ensuremath{\Varid{enhance}} is required to respect the monoidal structure of \ensuremath{\sigma}:

\begin{hscode}\SaveRestoreHook
\column{B}{@{}>{\hspre}l<{\hspost}@{}}%
\column{E}{@{}>{\hspre}l<{\hspost}@{}}%
\>[B]{}\Varid{enhance}\;@\Conid{Id}\mathrel{=}\Varid{dimap}\;\Varid{unId}\;\Conid{Id}{}\<[E]%
\\
\>[B]{}\Varid{enhance}\;@(\Conid{Compose}\;\Varid{f}\;\Varid{g})\mathrel{=}\Varid{dimap}\;\Varid{unCompose}\;\Conid{Compose}\hsdot{\circ }{.\;}\Varid{enhance}\;@\Varid{f}\hsdot{\circ }{.\;}\Varid{enhance}\;@\Varid{g}{}\<[E]%
\ColumnHook
\end{hscode}\resethooks

The ``@'' syntax used here is type application, which is available as a
Haskell extension in the latest versions of GHC. We will be using it
extensively to disambiguate polymorphic function usage.

We also require the following wedge condition on \ensuremath{\Varid{enhance}}: given a
function \ensuremath{\alpha\mathbin{::}\mathbf{forall}\;\Varid{a}\hsforall \hsdot{\circ }{.\;}\Varid{f}\;\Varid{a}\to \Varid{g}\;\Varid{a}},

\begin{hscode}\SaveRestoreHook
\column{B}{@{}>{\hspre}l<{\hspost}@{}}%
\column{E}{@{}>{\hspre}l<{\hspost}@{}}%
\>[B]{}\Varid{dimap}\;\Varid{id}\;\alpha\hsdot{\circ }{.\;}\Varid{enhance}\;@\Varid{f}\mathrel{=}\Varid{dimap}\;\alpha\;\Varid{id}\hsdot{\circ }{.\;}\Varid{enhance}\;@\Varid{g}{}\<[E]%
\ColumnHook
\end{hscode}\resethooks

This is related to the categorical definition of \ensuremath{\Varid{enhance}} via an end,
and can be understood as a generalization of naturality conditions.

Finally, we mention the following law that is required but is always
true by parametricity:

\begin{hscode}\SaveRestoreHook
\column{B}{@{}>{\hspre}l<{\hspost}@{}}%
\column{E}{@{}>{\hspre}l<{\hspost}@{}}%
\>[B]{}\Varid{enhance}\;@\Varid{f}\hsdot{\circ }{.\;}\Varid{dimap}\;\Varid{f}\;\Varid{g}\mathrel{=}\Varid{dimap}\;(\Varid{fmap}\;\Varid{f})\;(\Varid{fmap}\;\Varid{g})\hsdot{\circ }{.\;}\Varid{enhance}\;@\Varid{f}{}\<[E]%
\ColumnHook
\end{hscode}\resethooks

\end{defn}

\bigskip

\begin{rmk}

\ensuremath{\Conid{Enhancing}} is a generalization of several profunctor typeclasses used
to define profunctor optics\autocite{mezzolens}. Where \ensuremath{\Conid{Cartesian}} and
\ensuremath{\Conid{Cocartesian}} express the ability to lift a profunctorial transformation
into a product or a sum, \ensuremath{\Conid{Enhancing}\;\sigma} expresses a similar
capability for a more general class of shapes, captured by the
constraint constructor \ensuremath{\sigma}.

The \ensuremath{\Conid{Enhancing}} laws are likewise generalizations of the laws of those
typeclasses\autocite{kmett/profunctors}.

The \ensuremath{\Conid{FunctorMonoid}} requirement is less obviously a generalization, but
is necessary for the laws to make sense. This fact that the usual
profunctor classes come from functor families that are closed under
composition was pointed out by Russell
O'Connor\autocite{r6_profunctor_hierarchy}.

\end{rmk}

\bigskip

\begin{rmk}

This definition of profunctor optics does not capture all that has been
called ``profunctor optics'' in the literature and folklore. It notably
does not capture ``one-way'' optics such as \ensuremath{\Conid{Fold}}, \ensuremath{\Conid{View}} and \ensuremath{\Conid{Review}}.
It also does not capture \ensuremath{\Conid{Traversal}}s as defined by Pickering et
al.\autocite{profunctor_optics}, but \ensuremath{\Conid{Traversal}}s have an equivalent
representation that fits this definition\autocite{mezzolens}.

More importantly, composing two profunctor optics with different \ensuremath{\sigma}
does not yield a profunctor optic that fits this definition either. This
can be circumvented, but this paper focuses on single optic families,
therefore it will not be needed.

\end{rmk}

\bigskip
\bigskip

\begin{prop}

The \ensuremath{(\to )} profunctor has an instance of \ensuremath{\Conid{Enhancing}} for any functor
monoid.

\end{prop}

\begin{proof}

\ensuremath{\Varid{enhance}} for the \ensuremath{(\to )} profunctor needs to lift an \ensuremath{\Varid{a}\to \Varid{b}} arrow into
an \ensuremath{\Varid{f}\;\Varid{a}\to \Varid{f}\;\Varid{b}} arrow. This is exactly the defining property of
functors. Since \ensuremath{\sigma} is a functor monoid, every instance of \ensuremath{\sigma}
is a functor, so \ensuremath{\Varid{enhance}} is simply \ensuremath{\Varid{fmap}}.

\begin{hscode}\SaveRestoreHook
\column{B}{@{}>{\hspre}l<{\hspost}@{}}%
\column{5}{@{}>{\hspre}l<{\hspost}@{}}%
\column{E}{@{}>{\hspre}l<{\hspost}@{}}%
\>[B]{}\mathbf{instance}\;\Conid{Enhancing}\;\sigma\;(\to )\;\mathbf{where}{}\<[E]%
\\
\>[B]{}\hsindent{5}{}\<[5]%
\>[5]{}\Varid{enhance}\mathbin{::}\sigma\;\Varid{f}\Rightarrow (\Varid{a}\to \Varid{b})\to (\Varid{f}\;\Varid{a}\to \Varid{f}\;\Varid{b}){}\<[E]%
\\
\>[B]{}\hsindent{5}{}\<[5]%
\>[5]{}\Varid{enhance}\mathrel{=}\Varid{fmap}{}\<[E]%
\ColumnHook
\end{hscode}\resethooks

Law 1:

\begin{hscode}\SaveRestoreHook
\column{B}{@{}>{\hspre}l<{\hspost}@{}}%
\column{4}{@{}>{\hspre}l<{\hspost}@{}}%
\column{E}{@{}>{\hspre}l<{\hspost}@{}}%
\>[B]{}\Varid{enhance}\;@(\to )\;@\Conid{Id}\;\Varid{f}{}\<[E]%
\\
\>[B]{}\mathrel{=}\mbox{\commentbegin   \ensuremath{\mathbf{instance}\;\Conid{Enhancing}\;\sigma\;(\to )}   \commentend}{}\<[E]%
\\
\>[B]{}\hsindent{4}{}\<[4]%
\>[4]{}\Varid{fmap}\;@\Conid{Id}\;\Varid{f}{}\<[E]%
\\
\>[B]{}\mathrel{=}\mbox{\commentbegin   \ensuremath{\mathbf{instance}\;\Conid{Functor}\;\Conid{Id}}   \commentend}{}\<[E]%
\\
\>[B]{}\hsindent{4}{}\<[4]%
\>[4]{}\Conid{Id}\hsdot{\circ }{.\;}\Varid{f}\hsdot{\circ }{.\;}\Varid{unId}{}\<[E]%
\\
\>[B]{}\mathrel{=}\mbox{\commentbegin   \ensuremath{\mathbf{instance}\;\Conid{Profunctor}\;(\to )}   \commentend}{}\<[E]%
\\
\>[B]{}\hsindent{4}{}\<[4]%
\>[4]{}\Varid{dimap}\;\Varid{unId}\;\Conid{Id}\;\Varid{f}{}\<[E]%
\ColumnHook
\end{hscode}\resethooks

Law 2:

\begin{hscode}\SaveRestoreHook
\column{B}{@{}>{\hspre}l<{\hspost}@{}}%
\column{4}{@{}>{\hspre}l<{\hspost}@{}}%
\column{E}{@{}>{\hspre}l<{\hspost}@{}}%
\>[B]{}\Varid{enhance}\;@(\to )\;@(\Conid{Compose}\;\Varid{f}\;\Varid{g}){}\<[E]%
\\
\>[B]{}\mathrel{=}\mbox{\commentbegin   \ensuremath{\mathbf{instance}\;\Conid{Enhancing}\;\sigma\;(\to )}   \commentend}{}\<[E]%
\\
\>[B]{}\hsindent{4}{}\<[4]%
\>[4]{}\Varid{fmap}\;@(\Conid{Compose}\;\Varid{f}\;\Varid{g}){}\<[E]%
\\
\>[B]{}\mathrel{=}\mbox{\commentbegin   \ensuremath{\mathbf{instance}\;\Conid{Functor}\;\Conid{Compose}}   \commentend}{}\<[E]%
\\
\>[B]{}\hsindent{4}{}\<[4]%
\>[4]{}\lambda \Varid{h}\to \Conid{Compose}\hsdot{\circ }{.\;}\Varid{fmap}\;@\Varid{f}\;(\Varid{fmap}\;@\Varid{g}\;\Varid{h})\hsdot{\circ }{.\;}\Varid{unCompose}{}\<[E]%
\\
\>[B]{}\mathrel{=}\mbox{\commentbegin   \ensuremath{\mathbf{instance}\;\Conid{Profunctor}\;(\to )}   \commentend}{}\<[E]%
\\
\>[B]{}\hsindent{4}{}\<[4]%
\>[4]{}\lambda \Varid{h}\to \Varid{dimap}\;\Varid{unCompose}\;\Conid{Compose}\;(\Varid{fmap}\;@\Varid{f}\;(\Varid{fmap}\;@\Varid{g}\;\Varid{h})){}\<[E]%
\\
\>[B]{}\mathrel{=}{}\<[4]%
\>[4]{}\Varid{dimap}\;\Varid{unCompose}\;\Conid{Compose}\hsdot{\circ }{.\;}\Varid{fmap}\;@\Varid{f}\hsdot{\circ }{.\;}\Varid{fmap}\;@\Varid{g}{}\<[E]%
\\
\>[B]{}\mathrel{=}\mbox{\commentbegin   \ensuremath{\mathbf{instance}\;\Conid{Enhancing}\;\sigma\;(\to )}   \commentend}{}\<[E]%
\\
\>[B]{}\hsindent{4}{}\<[4]%
\>[4]{}\Varid{dimap}\;\Varid{unCompose}\;\Conid{Compose}\hsdot{\circ }{.\;}\Varid{enhance}\;@(\to )\;@\Varid{f}\hsdot{\circ }{.\;}\Varid{enhance}\;@(\to )\;@\Varid{g}{}\<[E]%
\ColumnHook
\end{hscode}\resethooks

Law 3:

\begin{hscode}\SaveRestoreHook
\column{B}{@{}>{\hspre}l<{\hspost}@{}}%
\column{4}{@{}>{\hspre}l<{\hspost}@{}}%
\column{E}{@{}>{\hspre}l<{\hspost}@{}}%
\>[B]{}\Varid{dimap}\;\Varid{id}\;\alpha\hsdot{\circ }{.\;}\Varid{enhance}\;@(\to )\;@\Varid{f}{}\<[E]%
\\
\>[B]{}\mathrel{=}\mbox{\commentbegin   \ensuremath{\mathbf{instance}\;\Conid{Enhancing}\;\sigma\;(\to )}   \commentend}{}\<[E]%
\\
\>[B]{}\hsindent{4}{}\<[4]%
\>[4]{}\Varid{dimap}\;\Varid{id}\;\alpha\hsdot{\circ }{.\;}\Varid{fmap}\;@\Varid{f}{}\<[E]%
\\
\>[B]{}\mathrel{=}\mbox{\commentbegin   \ensuremath{\mathbf{instance}\;\Conid{Profunctor}\;(\to )}   \commentend}{}\<[E]%
\\
\>[B]{}\hsindent{4}{}\<[4]%
\>[4]{}\lambda \Varid{h}\to \alpha\hsdot{\circ }{.\;}\Varid{fmap}\;@\Varid{f}\;\Varid{h}{}\<[E]%
\\
\>[B]{}\mathrel{=}\mbox{\commentbegin   Naturality of \ensuremath{\alpha}, aka parametricity   \commentend}{}\<[E]%
\\
\>[B]{}\hsindent{4}{}\<[4]%
\>[4]{}\lambda \Varid{h}\to \Varid{fmap}\;@\Varid{g}\;\Varid{h}\hsdot{\circ }{.\;}\alpha{}\<[E]%
\\
\>[B]{}\mathrel{=}\mbox{\commentbegin   \ensuremath{\mathbf{instance}\;\Conid{Profunctor}\;(\to )}   \commentend}{}\<[E]%
\\
\>[B]{}\hsindent{4}{}\<[4]%
\>[4]{}\Varid{dimap}\;\alpha\;\Varid{id}\hsdot{\circ }{.\;}\Varid{fmap}\;@\Varid{g}{}\<[E]%
\\
\>[B]{}\mathrel{=}\mbox{\commentbegin   \ensuremath{\mathbf{instance}\;\Conid{Enhancing}\;\sigma\;(\to )}   \commentend}{}\<[E]%
\\
\>[B]{}\hsindent{4}{}\<[4]%
\>[4]{}\Varid{dimap}\;\alpha\;\Varid{id}\hsdot{\circ }{.\;}\Varid{enhance}\;@(\to )\;@\Varid{g}{}\<[E]%
\ColumnHook
\end{hscode}\resethooks

\end{proof}

\bigskip

\begin{prop}

For a given functor monoid \ensuremath{\sigma}, \ensuremath{\Conid{ProfOptic}\;\sigma} is an optic
family.

\end{prop}

\begin{proof}

We can write the instance:

\begin{hscode}\SaveRestoreHook
\column{B}{@{}>{\hspre}l<{\hspost}@{}}%
\column{5}{@{}>{\hspre}l<{\hspost}@{}}%
\column{E}{@{}>{\hspre}l<{\hspost}@{}}%
\>[B]{}\mathbf{instance}\;\Conid{OpticFamily}\;(\Conid{ProfOptic}\;\sigma)\;\mathbf{where}{}\<[E]%
\\
\>[B]{}\hsindent{5}{}\<[5]%
\>[5]{}\Varid{injOptic}\;\Varid{f}\;\Varid{g}\mathrel{=}\Varid{dimap}\;\Varid{f}\;\Varid{g}{}\<[E]%
\\
\>[B]{}\hsindent{5}{}\<[5]%
\>[5]{}\Varid{l}_{\mathrm{1}}\ \circ_{op}\ \Varid{l}_{\mathrm{2}}\mathrel{=}\Varid{l}_{\mathrm{1}}\hsdot{\circ }{.\;}\Varid{l}_{\mathrm{2}}{}\<[E]%
\\
\>[B]{}\hsindent{5}{}\<[5]%
\>[5]{}\Varid{mapOptic}\;\Varid{l}\mathrel{=}\Varid{l}\;@(\to ){}\<[E]%
\ColumnHook
\end{hscode}\resethooks

The laws are immediate.

\end{proof}

\bigskip

\begin{prop}\label{lens_is_profop_isproduct}

\ensuremath{\Conid{PLens}\cong \Conid{ProfOptic}\;\Conid{IsProduct}}, where:

\begin{hscode}\SaveRestoreHook
\column{B}{@{}>{\hspre}l<{\hspost}@{}}%
\column{5}{@{}>{\hspre}l<{\hspost}@{}}%
\column{E}{@{}>{\hspre}l<{\hspost}@{}}%
\>[B]{}\mathbf{class}\;\Conid{Functor}\;\Varid{f}\Rightarrow \Conid{IsProduct}\;\Varid{f}\;\mathbf{where}{}\<[E]%
\\
\>[B]{}\hsindent{5}{}\<[5]%
\>[5]{}\Varid{toProduct}\mathbin{::}\Varid{f}\;\Varid{a}\to (\Varid{f}\;(),\Varid{a}){}\<[E]%
\\
\>[B]{}\hsindent{5}{}\<[5]%
\>[5]{}\Varid{fromProduct}\mathbin{::}(\Varid{f}\;(),\Varid{a})\to \Varid{f}\;\Varid{a}{}\<[E]%
\\
\>[B]{}\hsindent{5}{}\<[5]%
\>[5]{}\mbox{\onelinecomment  toProduct and fromProduct should be mutual inverses}{}\<[E]%
\ColumnHook
\end{hscode}\resethooks

\end{prop}

\begin{proof}

First, \ensuremath{\Conid{IsProduct}} is a functor monoid:

\begin{hscode}\SaveRestoreHook
\column{B}{@{}>{\hspre}l<{\hspost}@{}}%
\column{5}{@{}>{\hspre}l<{\hspost}@{}}%
\column{9}{@{}>{\hspre}l<{\hspost}@{}}%
\column{14}{@{}>{\hspre}l<{\hspost}@{}}%
\column{E}{@{}>{\hspre}l<{\hspost}@{}}%
\>[B]{}\mathbf{instance}\;\Conid{IsProduct}\;\Conid{Id}\;\mathbf{where}{}\<[E]%
\\
\>[B]{}\hsindent{5}{}\<[5]%
\>[5]{}\Varid{toProduct}\;(\Conid{Id}\;\Varid{x})\mathrel{=}(\Conid{Id}\;(),\Varid{x}){}\<[E]%
\\
\>[B]{}\hsindent{5}{}\<[5]%
\>[5]{}\Varid{fromProduct}\;(\Conid{Id}\;(),\Varid{x})\mathrel{=}(\Conid{Id}\;\Varid{x}){}\<[E]%
\\[\blanklineskip]%
\>[B]{}\mathbf{instance}\;(\Conid{IsProduct}\;\Varid{f},\Conid{IsProduct}\;\Varid{g})\Rightarrow \Conid{IsProduct}\;(\Conid{Compose}\;\Varid{f}\;\Varid{g})\;\mathbf{where}{}\<[E]%
\\
\>[B]{}\hsindent{5}{}\<[5]%
\>[5]{}\Varid{toProduct}\;(\Conid{Compose}\;\Varid{fgx})\mathrel{=}{}\<[E]%
\\
\>[5]{}\hsindent{4}{}\<[9]%
\>[9]{}\mathbf{let}\;{}\<[14]%
\>[14]{}(\Varid{f1},(\Varid{g1},\Varid{x}))\mathrel{=}\Varid{fmap}\;\Varid{toProduct}\;(\Varid{toProduct}\;\Varid{fgx}){}\<[E]%
\\
\>[14]{}\Varid{fg1}\mathrel{=}\Varid{fromProduct}\;(\Varid{f1},\Varid{g1})\;\mathbf{in}{}\<[E]%
\\
\>[5]{}\hsindent{4}{}\<[9]%
\>[9]{}(\Conid{Compose}\;\Varid{fg1},\Varid{x}){}\<[E]%
\\
\>[B]{}\hsindent{5}{}\<[5]%
\>[5]{}\Varid{fromProduct}\;(\Conid{Compose}\;\Varid{fg1},\Varid{x})\mathrel{=}{}\<[E]%
\\
\>[5]{}\hsindent{4}{}\<[9]%
\>[9]{}\mathbf{let}\;{}\<[14]%
\>[14]{}(\Varid{f1},\Varid{g1})\mathrel{=}\Varid{toProduct}\;\Varid{fg1}{}\<[E]%
\\
\>[14]{}\Varid{fgx}\mathrel{=}\Varid{fromProduct}\;(\Varid{f1},\Varid{fromProduct}\;(\Varid{g1},\Varid{x}))\;\mathbf{in}{}\<[E]%
\\
\>[5]{}\hsindent{4}{}\<[9]%
\>[9]{}(\Conid{Compose}\;\Varid{fgx}){}\<[E]%
\ColumnHook
\end{hscode}\resethooks

The \ensuremath{(\Varid{r},\mathbin{-})} functor is trivially an instance of \ensuremath{\Conid{IsProduct}}:

\begin{hscode}\SaveRestoreHook
\column{B}{@{}>{\hspre}l<{\hspost}@{}}%
\column{5}{@{}>{\hspre}l<{\hspost}@{}}%
\column{E}{@{}>{\hspre}l<{\hspost}@{}}%
\>[B]{}\mathbf{instance}\;\Conid{IsProduct}\;((,)\;\Varid{r})\;\mathbf{where}{}\<[E]%
\\
\>[B]{}\hsindent{5}{}\<[5]%
\>[5]{}\Varid{toProduct}\mathbin{::}(\Varid{r},\Varid{a})\to ((\Varid{r},()),\Varid{a}){}\<[E]%
\\
\>[B]{}\hsindent{5}{}\<[5]%
\>[5]{}\Varid{toProduct}\;(\Varid{r},\Varid{a})\mathrel{=}((\Varid{r},()),\Varid{a}){}\<[E]%
\\
\>[B]{}\hsindent{5}{}\<[5]%
\>[5]{}\Varid{fromProduct}\mathbin{::}((\Varid{r},()),\Varid{a})\to (\Varid{r},\Varid{a}){}\<[E]%
\\
\>[B]{}\hsindent{5}{}\<[5]%
\>[5]{}\Varid{fromProduct}\;((\Varid{r},()),\Varid{a})\mathrel{=}(\Varid{r},\Varid{a}){}\<[E]%
\ColumnHook
\end{hscode}\resethooks

We can now write the isomorphism:

\begin{hscode}\SaveRestoreHook
\column{B}{@{}>{\hspre}l<{\hspost}@{}}%
\column{5}{@{}>{\hspre}l<{\hspost}@{}}%
\column{E}{@{}>{\hspre}l<{\hspost}@{}}%
\>[B]{}\mathbf{instance}\;\Conid{Cartesian}\;\Varid{p}\Rightarrow \Conid{Enhancing}\;\Conid{IsProduct}\;\Varid{p}\;\mathbf{where}{}\<[E]%
\\
\>[B]{}\hsindent{5}{}\<[5]%
\>[5]{}\Varid{enhance}\mathbin{::}\mathbf{forall}\;\Varid{f}\hsforall \hsdot{\circ }{.\;}\Conid{IsProduct}\;\Varid{f}\Rightarrow \Varid{p}\;\Varid{a}\;\Varid{b}\to \Varid{p}\;(\Varid{f}\;\Varid{a})\;(\Varid{f}\;\Varid{b}){}\<[E]%
\\
\>[B]{}\hsindent{5}{}\<[5]%
\>[5]{}\Varid{enhance}\mathrel{=}\Varid{dimap}\;\Varid{fromProduct}\;\Varid{toProduct}\hsdot{\circ }{.\;}\Varid{second}{}\<[E]%
\\[\blanklineskip]%
\>[B]{}\mathbf{instance}\;\Conid{Enhancing}\;\Conid{IsProduct}\;\Varid{p}\Rightarrow \Conid{Cartesian}\;\Varid{p}\;\mathbf{where}{}\<[E]%
\\
\>[B]{}\hsindent{5}{}\<[5]%
\>[5]{}\Varid{second}\mathbin{::}\mathbf{forall}\;\Varid{c}\hsforall \hsdot{\circ }{.\;}\Varid{p}\;\Varid{a}\;\Varid{b}\to \Varid{p}\;(\Varid{c},\Varid{a})\;(\Varid{c},\Varid{b}){}\<[E]%
\\
\>[B]{}\hsindent{5}{}\<[5]%
\>[5]{}\Varid{second}\mathrel{=}\Varid{enhance}\;@(\Varid{c},\mathbin{-}){}\<[E]%
\ColumnHook
\end{hscode}\resethooks

These operations are mutual inverses; see
\textcite{proof_lens_is_profop_isproduct} for the proof.

Finally, we need to verify that the instances are lawful. Not by
accident, the \ensuremath{\Conid{Cartesian}} laws are exactly a specialization of the
\ensuremath{\Conid{Enhancing}} laws to the \ensuremath{(\Varid{c},\mathbin{-})} case. We therefore omit proving that
this isomorphism produces lawful instances here.

Hence \ensuremath{\Conid{Cartesian}\cong \Conid{Enhancing}\;\Conid{IsProduct}}, and therefore
\ensuremath{\Conid{PLens}\cong \Conid{ProfOptic}\;\Conid{IsProduct}}.

\end{proof}

\bigskip

\hypertarget{profunctor-encoding-1}{%
\section{Profunctor encoding}\label{profunctor-encoding-1}}

\begin{defn}[Profunctor encoding]

An optic family is said to have a profunctor encoding if there is a
profunctor optic isomorphic to it.

We capture this definition in the following typeclass:

\begin{hscode}\SaveRestoreHook
\column{B}{@{}>{\hspre}l<{\hspost}@{}}%
\column{5}{@{}>{\hspre}l<{\hspost}@{}}%
\column{E}{@{}>{\hspre}l<{\hspost}@{}}%
\>[B]{}\mathbf{class}\;(\Conid{OpticFamily}\;\Varid{op},\Conid{FunctorMonoid}\;\sigma)\Rightarrow \Conid{ProfEncoding}\;\Varid{op}\;\sigma\;\mathbf{where}{}\<[E]%
\\
\>[B]{}\hsindent{5}{}\<[5]%
\>[5]{}\Varid{decodeProfOptic}\mathbin{::}\Conid{ProfOptic}\;\sigma\leadsto \Varid{op}{}\<[E]%
\\
\>[B]{}\hsindent{5}{}\<[5]%
\>[5]{}\Varid{encodeProfOptic}\mathbin{::}\Varid{op}\leadsto \Conid{ProfOptic}\;\sigma{}\<[E]%
\ColumnHook
\end{hscode}\resethooks

\ensuremath{\Varid{encodeProfOptic}} and \ensuremath{\Varid{decodeProfOptic}} should be mutual inverses.

Note that \ensuremath{\Varid{decodeProfOptic}} and \ensuremath{\Varid{encodeProfOptic}} are required to be
morphisms of optic families.

\end{defn}

\bigskip

\begin{expl}

\ensuremath{\Conid{Lens}} has a profunctor encoding, as mentioned in the introduction. We
do not detail the encoding here, but we derive it in
\textcite{rederiving-profunctor-encodings-for-common-optics}.

\end{expl}

\hypertarget{isomorphism-optics}{%
\chapter{Isomorphism optics}\label{isomorphism-optics}}

As remarked by Russell O'Connor\autocite{r6_profunctor_hierarchy} and
others\autocite{isomorphism_lenses}, a key to understanding (profunctor)
optics is to formulate optic families as certain classes of
isomorphisms. Indeed, for most common optic families, an optic can be
seen to denote (when the types match) an isomorphism between the
``external'' type (\ensuremath{\Varid{s}}) and a functor of a certain shape applied to the
``internal type'' (\ensuremath{\Varid{a}}).

The ``shape'' of the functor determines the optic family. For
(very-well-behaved) lenses for example, the relevant family of functors
are the functors \ensuremath{(\Varid{r},\mathbin{-})} for some type \ensuremath{\Varid{r}}. This is expressed as
follows:

\begin{hscode}\SaveRestoreHook
\column{B}{@{}>{\hspre}l<{\hspost}@{}}%
\column{E}{@{}>{\hspre}l<{\hspost}@{}}%
\>[B]{}\Conid{Lens}\;\Varid{a}\;\Varid{a}\;\Varid{s}\;\Varid{s}\cong \mathbf{exists}\;\Varid{r}\hsexists \hsdot{\circ }{.\;}\Varid{s}\approx \Varid{r}\mathbin{\mathbf{\times}}\Varid{a}{}\<[E]%
\ColumnHook
\end{hscode}\resethooks

Other optics have similar expressions:

\begin{hscode}\SaveRestoreHook
\column{B}{@{}>{\hspre}l<{\hspost}@{}}%
\column{E}{@{}>{\hspre}l<{\hspost}@{}}%
\>[B]{}\Conid{Prism}\;\Varid{a}\;\Varid{a}\;\Varid{s}\;\Varid{s}\cong \mathbf{exists}\;\Varid{r}\hsexists \hsdot{\circ }{.\;}\Varid{s}\approx \Varid{r}\mathbin{+}\Varid{a}{}\<[E]%
\\
\>[B]{}\Conid{Traversal}\;\Varid{a}\;\Varid{a}\;\Varid{s}\;\Varid{s}\cong \mathbf{exists}\;\Varid{f}\hsexists \hsdot{\circ }{.\;}\Conid{Traversable}\;\Varid{f}\Rightarrow \Varid{s}\approx \Varid{f}\;\Varid{a}{}\<[E]%
\\
\>[B]{}\Conid{Setter}\;\Varid{a}\;\Varid{a}\;\Varid{s}\;\Varid{s}\cong \mathbf{exists}\;\Varid{f}\hsexists \hsdot{\circ }{.\;}\Conid{Functor}\;\Varid{f}\Rightarrow \Varid{s}\approx \Varid{f}\;\Varid{a}{}\<[E]%
\ColumnHook
\end{hscode}\resethooks

\bigskip

We can generalize his insight to cases when the types don't match by
splitting the isomorphisms into two functions:

\begin{hscode}\SaveRestoreHook
\column{B}{@{}>{\hspre}l<{\hspost}@{}}%
\column{43}{@{}>{\hspre}l<{\hspost}@{}}%
\column{E}{@{}>{\hspre}l<{\hspost}@{}}%
\>[B]{}\Conid{Lens}\;\Varid{a}\;\Varid{b}\;\Varid{s}\;\Varid{t}\cong \mathbf{exists}\;\Varid{r}\hsexists \hsdot{\circ }{.\;}(\Varid{s}\to \Varid{r}\mathbin{\mathbf{\times}}\Varid{a},{}\<[43]%
\>[43]{}\Varid{r}\mathbin{\mathbf{\times}}\Varid{b}\to \Varid{t}){}\<[E]%
\ColumnHook
\end{hscode}\resethooks

etc.

\bigskip

R. O'Connor also notes a fundamental property of those classes of
functors: they are closed under composition. We again capture this
property by requiring the relevant families of functors to be functor
monoids.

This alternative expression of optics forms a new encoding that, like
profunctor optics, reveals interesting new properties about the usual
optic families. We call optics of this form \emph{isomorphism optics},
or \emph{iso optics} for short. \bigskip

\hypertarget{definition-1}{%
\section{Definition}\label{definition-1}}

\begin{defn}[Isomorphism optic]

Given a functor monoid \ensuremath{\sigma}, we call isomorphism optic for \ensuremath{\sigma} a
value of the following datatype:

\begin{hscode}\SaveRestoreHook
\column{B}{@{}>{\hspre}l<{\hspost}@{}}%
\column{E}{@{}>{\hspre}l<{\hspost}@{}}%
\>[B]{}\mathbf{data}\;\Conid{IsoOptic}\;\sigma\;\Varid{a}\;\Varid{b}\;\Varid{s}\;\Varid{t}\mathrel{=}\mathbf{forall}\;\Varid{f}\hsforall \hsdot{\circ }{.\;}\sigma\;\Varid{f}\Rightarrow \Conid{IsoOptic}\;(\Varid{s}\to \Varid{f}\;\Varid{a})\;(\Varid{f}\;\Varid{b}\to \Varid{t}){}\<[E]%
\ColumnHook
\end{hscode}\resethooks

\end{defn}

\bigskip

\begin{rmk}

Notice that the \ensuremath{\Varid{f}} used in the definition is not present in the type
parameters of the datatype. This feature is called \emph{existential
quantification} and is a special case of the more general notion of a
\emph{GADT} (Generalized Algebraic DataType).

Admittedly, using the \text{\ttfamily forall} keyword to define existential
quantification can be surprising. This usage comes from the following
fundamental property of quantification:

\begin{hscode}\SaveRestoreHook
\column{B}{@{}>{\hspre}l<{\hspost}@{}}%
\column{E}{@{}>{\hspre}l<{\hspost}@{}}%
\>[B]{}(\mathbf{exists}\;\Varid{x}\hsexists \hsdot{\circ }{.\;}\Varid{f}_{\Varid{x}})\to \Varid{y}\cong \mathbf{forall}\;\Varid{x}\hsforall \hsdot{\circ }{.\;}(\Varid{f}_{\Varid{x}}\to \Varid{y}){}\<[E]%
\ColumnHook
\end{hscode}\resethooks

where \ensuremath{\Varid{f}_{\Varid{x}}} is an expression involving \ensuremath{\Varid{x}}.

Using this property, we can see that the type of the \ensuremath{\Conid{IsoOptic}}
constructor becomes:

\begin{hscode}\SaveRestoreHook
\column{B}{@{}>{\hspre}l<{\hspost}@{}}%
\column{E}{@{}>{\hspre}l<{\hspost}@{}}%
\>[B]{}\Conid{IsoOptic}\mathbin{::}\mathbf{forall}\;\Varid{f}\hsforall \hsdot{\circ }{.\;}\sigma\;\Varid{f}\Rightarrow (\Varid{s}\to \Varid{f}\;\Varid{a},\Varid{f}\;\Varid{b}\to \Varid{t})\to \Conid{IsoOptic}\;\sigma\;\Varid{a}\;\Varid{b}\;\Varid{s}\;\Varid{t}{}\<[E]%
\\
\>[B]{}\cong (\mathbf{exists}\;\Varid{f}\hsexists \hsdot{\circ }{.\;}\sigma\;\Varid{f}\Rightarrow (\Varid{s}\to \Varid{f}\;\Varid{a},\Varid{f}\;\Varid{b}\to \Varid{t}))\to \Conid{IsoOptic}\;\sigma\;\Varid{a}\;\Varid{b}\;\Varid{s}\;\Varid{t}{}\<[E]%
\ColumnHook
\end{hscode}\resethooks

Thus this datatype should be understood as follows: a value of type
\ensuremath{\Conid{IsoOptic}\;\sigma\;\Varid{a}\;\Varid{b}\;\Varid{s}\;\Varid{t}} is a triplet carrying a type constructor \ensuremath{\Varid{f}}
which is an instance of \ensuremath{\sigma}, as well as two functions, respectively
of type \ensuremath{\Varid{s}\to \Varid{f}\;\Varid{a}} and \ensuremath{\Varid{f}\;\Varid{b}\to \Varid{t}}.

\begin{hscode}\SaveRestoreHook
\column{B}{@{}>{\hspre}l<{\hspost}@{}}%
\column{E}{@{}>{\hspre}l<{\hspost}@{}}%
\>[B]{}\Conid{IsoOptic}\;\sigma\;\Varid{a}\;\Varid{b}\;\Varid{s}\;\Varid{t}\cong \mathbf{exists}\;\Varid{f}\hsexists \hsdot{\circ }{.\;}\sigma\;\Varid{f}\Rightarrow (\Varid{s}\to \Varid{f}\;\Varid{a},\Varid{f}\;\Varid{b}\to \Varid{t}){}\<[E]%
\ColumnHook
\end{hscode}\resethooks

The specific \ensuremath{\Varid{f}} used to define a specific \ensuremath{\Conid{IsoOptic}} is not observable:
a function that takes an \ensuremath{\Conid{IsoOptic}} as argument has no information on
the particular \ensuremath{\Varid{f}} except that it is an instance of \ensuremath{\sigma}.

\end{rmk}

\bigskip

\begin{rmk}

Existential quantification only happens in a datatype declaration and
when the \text{\ttfamily forall} quantifier is before the constructor. In particular,
the definition of \ensuremath{\Conid{ProfOptic}} is not existentially quantified.

\end{rmk}

\bigskip
\bigskip

\begin{expl}

For example, here is a simple function that consumes an
\ensuremath{\Conid{IsoOptic}\;\Conid{Functor}} (with type annotations added to show the hidden \ensuremath{\Varid{f}}):

\begin{hscode}\SaveRestoreHook
\column{B}{@{}>{\hspre}l<{\hspost}@{}}%
\column{9}{@{}>{\hspre}l<{\hspost}@{}}%
\column{E}{@{}>{\hspre}l<{\hspost}@{}}%
\>[B]{}\Varid{isoOpFunctorToMap}\mathbin{::}\Conid{IsoOptic}\;\Conid{Functor}\;\Varid{a}\;\Varid{b}\;\Varid{s}\;\Varid{t}\to (\Varid{a}\to \Varid{b})\to (\Varid{s}\to \Varid{t}){}\<[E]%
\\
\>[B]{}\Varid{isoOpFunctorToMap}\;(\Conid{IsoOptic}\;(\alpha\mathbin{::}\Varid{s}\to \Varid{f}\;\Varid{a})\;(\beta\mathbin{::}\Varid{f}\;\Varid{b}\to \Varid{t}))\;(\Varid{g}\mathbin{::}\Varid{a}\to \Varid{b})\mathrel{=}{}\<[E]%
\\
\>[B]{}\hsindent{9}{}\<[9]%
\>[9]{}\beta\hsdot{\circ }{.\;}\Varid{fmap}\;\Varid{g}\hsdot{\circ }{.\;}\alpha{}\<[E]%
\ColumnHook
\end{hscode}\resethooks

Since the hidden \ensuremath{\Varid{f}} is a \ensuremath{\Conid{Functor}} here, we know we can \ensuremath{\Varid{fmap}} on it.

In fact, this can be applied to any functor monoid \ensuremath{\sigma}, since
\ensuremath{\sigma\Rightarrow \Conid{Functor}}:

\begin{hscode}\SaveRestoreHook
\column{B}{@{}>{\hspre}l<{\hspost}@{}}%
\column{9}{@{}>{\hspre}l<{\hspost}@{}}%
\column{E}{@{}>{\hspre}l<{\hspost}@{}}%
\>[B]{}\Varid{isoOpToMap}\mathbin{::}\Conid{FunctorMonoid}\;\sigma\Rightarrow \Conid{IsoOptic}\;\sigma\;\Varid{a}\;\Varid{b}\;\Varid{s}\;\Varid{t}\to (\Varid{a}\to \Varid{b})\to (\Varid{s}\to \Varid{t}){}\<[E]%
\\
\>[B]{}\Varid{isoOpToMap}\;(\Conid{IsoOptic}\;(\alpha\mathbin{::}\Varid{s}\to \Varid{f}\;\Varid{a})\;(\beta\mathbin{::}\Varid{f}\;\Varid{b}\to \Varid{t}))\;(\Varid{g}\mathbin{::}\Varid{a}\to \Varid{b})\mathrel{=}{}\<[E]%
\\
\>[B]{}\hsindent{9}{}\<[9]%
\>[9]{}\beta\hsdot{\circ }{.\;}\Varid{fmap}\;\Varid{g}\hsdot{\circ }{.\;}\alpha{}\<[E]%
\ColumnHook
\end{hscode}\resethooks

\end{expl}

\bigskip

\begin{rmk}

Since the functor \ensuremath{\Varid{f}} carried by a given iso optic is not directly
observable, comparing iso optics for equality cannot reduce to comparing
the two carried functions for equality since their types may not match.
On the other hand, comparing the carried \ensuremath{\Varid{f}}s for equality would be too
restrictive, since in particular two iso optics carrying isomorphic \ensuremath{\Varid{f}}s
cannot be distinguished.

A proper definition of equality for iso optics requires a better
understanding of exactly which iso optics can be operationally
distinguished. The answer lies in the universal property of a
\emph{coend}, which is the categorical notion that corresponds to the
existential used to define iso optics.

The universal property of coends implies that, given
\ensuremath{\epsilon\mathbin{::}\Conid{IsoOptic}\;\sigma\;\Varid{a}\;\Varid{b}\;\Varid{s}\;\Varid{t}\to \Varid{x}} and
\ensuremath{\phi\mathbin{::}\mathbf{forall}\;\Varid{a}\hsforall \hsdot{\circ }{.\;}\Varid{f}\;\Varid{a}\to \Varid{g}\;\Varid{a}}, then
\ensuremath{\epsilon\;(\Conid{IsoOptic}\;(\phi\hsdot{\circ }{.\;}\alpha)\;\beta)\mathrel{=}\epsilon\;(\Conid{IsoOptic}\;\alpha\;(\beta\hsdot{\circ }{.\;}\phi))}
when the types match.

This motivates the following definition for the equality between iso
optics:

\end{rmk}

\bigskip

\begin{defn}\label{iso_equality}

Two iso optics are considered equal if the functions they carry are
equal up to a natural transformation between the carried functors.

In symbols:

\begin{hscode}\SaveRestoreHook
\column{B}{@{}>{\hspre}l<{\hspost}@{}}%
\column{E}{@{}>{\hspre}l<{\hspost}@{}}%
\>[B]{}\Conid{IsoOptic}\;(\alpha\mathbin{::}\Varid{s}\to \Varid{f}\;\Varid{a})\;(\beta\mathbin{::}\Varid{f}\;\Varid{b}\to \Varid{t})\mathrel{=}\Conid{IsoOptic}\;(\alpha'\mathbin{::}\Varid{s}\to \Varid{g}\;\Varid{a})\;(\beta'\mathbin{::}\Varid{g}\;\Varid{b}\to \Varid{t}){}\<[E]%
\\
\>[B]{}\mathbin{\mathbf{\Leftrightarrow}}\exists\;\phi\mathbin{::}\mathbf{forall}\;\Varid{a}\hsforall \hsdot{\circ }{.\;}\Varid{f}\;\Varid{a}\to \Varid{g}\;\Varid{a},\Varid{such}\;\Varid{that}\;\phi\hsdot{\circ }{.\;}\alpha\mathrel{=}\alpha'\;\Varid{and}\;\beta\mathrel{=}\beta'\hsdot{\circ }{.\;}\phi{}\<[E]%
\ColumnHook
\end{hscode}\resethooks

\end{defn}

\bigskip

\begin{rmk}\label{iso_equality_rmk}

We will be using this property mostly in the following form: if
\ensuremath{\phi\mathbin{::}\mathbf{forall}\;\Varid{a}\hsforall \hsdot{\circ }{.\;}\Varid{f}\;\Varid{a}\to \Varid{g}\;\Varid{a}}, then
\ensuremath{\Conid{IsoOptic}\;\alpha\;(\beta\hsdot{\circ }{.\;}\phi)\mathrel{=}\Conid{IsoOptic}\;(\phi\hsdot{\circ }{.\;}\alpha)\;\beta}.

\end{rmk}

\bigskip

\hypertarget{properties-1}{%
\section{Properties}\label{properties-1}}

\begin{prop}\label{iso_is_opfam}

For a given functor monoid \ensuremath{\sigma}, \ensuremath{\Conid{IsoOptic}\;\sigma} is an optic family.

\end{prop}

\begin{proof}

We can write an \ensuremath{\Conid{OpticFamily}} instance for \ensuremath{\Conid{IsoOptic}\;\sigma}:

\begin{hscode}\SaveRestoreHook
\column{B}{@{}>{\hspre}l<{\hspost}@{}}%
\column{E}{@{}>{\hspre}l<{\hspost}@{}}%
\>[B]{}\mathbf{instance}\;\Conid{FunctorMonoid}\;\sigma\Rightarrow \Conid{OpticFamily}\;(\Conid{IsoOptic}\;\sigma)\;\mathbf{where}{}\<[E]%
\ColumnHook
\end{hscode}\resethooks

\ensuremath{\Conid{Id}} is an instance of \ensuremath{\sigma} as per the axioms of functor monoids:

\begin{hscode}\SaveRestoreHook
\column{B}{@{}>{\hspre}l<{\hspost}@{}}%
\column{5}{@{}>{\hspre}l<{\hspost}@{}}%
\column{9}{@{}>{\hspre}l<{\hspost}@{}}%
\column{13}{@{}>{\hspre}l<{\hspost}@{}}%
\column{E}{@{}>{\hspre}l<{\hspost}@{}}%
\>[5]{}\Varid{injOptic}\mathbin{::}(\Varid{s}\to \Varid{a})\to (\Varid{b}\to \Varid{t})\to \Conid{IsoOptic}\;\sigma\;\Varid{a}\;\Varid{b}\;\Varid{s}\;\Varid{t}{}\<[E]%
\\
\>[5]{}\Varid{injOptic}\;\Varid{f}\;\Varid{g}\mathrel{=}\Conid{IsoOptic}\;\alpha\;\beta{}\<[E]%
\\
\>[5]{}\hsindent{4}{}\<[9]%
\>[9]{}\mathbf{where}{}\<[E]%
\\
\>[9]{}\hsindent{4}{}\<[13]%
\>[13]{}\alpha\mathbin{::}\Varid{s}\to \Conid{Id}\;\Varid{a}{}\<[E]%
\\
\>[9]{}\hsindent{4}{}\<[13]%
\>[13]{}\alpha\mathrel{=}\Conid{Id}\hsdot{\circ }{.\;}\Varid{f}{}\<[E]%
\\
\>[9]{}\hsindent{4}{}\<[13]%
\>[13]{}\beta\mathbin{::}\Conid{Id}\;\Varid{b}\to \Varid{t}{}\<[E]%
\\
\>[9]{}\hsindent{4}{}\<[13]%
\>[13]{}\beta\mathrel{=}\Varid{g}\hsdot{\circ }{.\;}\Varid{unId}{}\<[E]%
\ColumnHook
\end{hscode}\resethooks

Functor monoids are closed under composition:

\begin{hscode}\SaveRestoreHook
\column{B}{@{}>{\hspre}l<{\hspost}@{}}%
\column{5}{@{}>{\hspre}l<{\hspost}@{}}%
\column{9}{@{}>{\hspre}l<{\hspost}@{}}%
\column{13}{@{}>{\hspre}l<{\hspost}@{}}%
\column{17}{@{}>{\hspre}l<{\hspost}@{}}%
\column{21}{@{}>{\hspre}l<{\hspost}@{}}%
\column{E}{@{}>{\hspre}l<{\hspost}@{}}%
\>[5]{}(\circ_{op})\mathbin{::}\Conid{IsoOptic}\;\sigma\;\Varid{a}\;\Varid{b}\;\Varid{s}\;\Varid{t}\to \Conid{IsoOptic}\;\sigma\;\Varid{x}\;\Varid{y}\;\Varid{a}\;\Varid{b}\to \Conid{IsoOptic}\;\sigma\;\Varid{x}\;\Varid{y}\;\Varid{s}\;\Varid{t}{}\<[E]%
\\
\>[5]{}(\circ_{op})\;{}\<[E]%
\\
\>[5]{}\hsindent{4}{}\<[9]%
\>[9]{}(\Conid{IsoOptic}\;(\alpha_{\mathrm{1}}\mathbin{::}\Varid{s}\to \Varid{f}\;\Varid{a})\;(\beta_{\mathrm{1}}\mathbin{::}\Varid{f}\;\Varid{b}\to \Varid{t}))\;{}\<[E]%
\\
\>[5]{}\hsindent{4}{}\<[9]%
\>[9]{}(\Conid{IsoOptic}\;(\alpha_{\mathrm{2}}\mathbin{::}\Varid{a}\to \Varid{g}\;\Varid{x})\;(\beta_{\mathrm{2}}\mathbin{::}\Varid{g}\;\Varid{y}\to \Varid{b}))\mathrel{=}{}\<[E]%
\\
\>[9]{}\hsindent{4}{}\<[13]%
\>[13]{}\Conid{IsoOptic}\;\alpha\;\beta{}\<[E]%
\\
\>[13]{}\hsindent{4}{}\<[17]%
\>[17]{}\mathbf{where}{}\<[E]%
\\
\>[17]{}\hsindent{4}{}\<[21]%
\>[21]{}\alpha\mathbin{::}\Varid{s}\to \Conid{Compose}\;\Varid{f}\;\Varid{g}\;\Varid{x}{}\<[E]%
\\
\>[17]{}\hsindent{4}{}\<[21]%
\>[21]{}\alpha\mathrel{=}\Conid{Compose}\hsdot{\circ }{.\;}\Varid{fmap}\;@\Varid{f}\;\alpha_{\mathrm{2}}\hsdot{\circ }{.\;}\alpha_{\mathrm{1}}{}\<[E]%
\\
\>[17]{}\hsindent{4}{}\<[21]%
\>[21]{}\beta\mathbin{::}\Conid{Compose}\;\Varid{f}\;\Varid{g}\;\Varid{y}\to \Varid{t}{}\<[E]%
\\
\>[17]{}\hsindent{4}{}\<[21]%
\>[21]{}\beta\mathrel{=}\beta_{\mathrm{1}}\hsdot{\circ }{.\;}\Varid{fmap}\;@\Varid{f}\;\beta_{\mathrm{2}}\hsdot{\circ }{.\;}\Varid{unCompose}{}\<[E]%
\ColumnHook
\end{hscode}\resethooks

\ensuremath{\Varid{mapOptic}} is simply \ensuremath{\Varid{isoOpToMap}} as defined above:

\begin{hscode}\SaveRestoreHook
\column{B}{@{}>{\hspre}l<{\hspost}@{}}%
\column{5}{@{}>{\hspre}l<{\hspost}@{}}%
\column{E}{@{}>{\hspre}l<{\hspost}@{}}%
\>[5]{}\Varid{mapOptic}\mathbin{::}\Conid{IsoOptic}\;\sigma\;\Varid{a}\;\Varid{b}\;\Varid{s}\;\Varid{t}\to (\Varid{a}\to \Varid{b})\to (\Varid{s}\to \Varid{t}){}\<[E]%
\\
\>[5]{}\Varid{mapOptic}\;(\Conid{IsoOptic}\;\alpha\;\beta)\;\Varid{g}\mathrel{=}\beta\hsdot{\circ }{.\;}\Varid{fmap}\;\Varid{g}\hsdot{\circ }{.\;}\alpha{}\<[E]%
\ColumnHook
\end{hscode}\resethooks

See \textcite{proof_iso_is_opfam} for a proof of the laws.

\end{proof}

\bigskip

\begin{defn}

For \ensuremath{\Varid{f}\;\mathbin{\in}\;\sigma}, we define

\begin{hscode}\SaveRestoreHook
\column{B}{@{}>{\hspre}l<{\hspost}@{}}%
\column{E}{@{}>{\hspre}l<{\hspost}@{}}%
\>[B]{}\Varid{enhanceIso}\mathbin{::}\sigma\;\Varid{f}\Rightarrow \Conid{IsoOptic}\;\sigma\;\Varid{a}\;\Varid{b}\;(\Varid{f}\;\Varid{a})\;(\Varid{f}\;\Varid{b}){}\<[E]%
\\
\>[B]{}\Varid{enhanceIso}\mathrel{=}\Conid{IsoOptic}\;\Varid{id}\;\Varid{id}{}\<[E]%
\ColumnHook
\end{hscode}\resethooks

\end{defn}

\bigskip

\begin{prop}\label{iso_normal_form}

Let
\ensuremath{\Varid{l}\mathrel{=}\Conid{IsoOptic}\;(\alpha\mathbin{::}\Varid{s}\to \Varid{f}\;\Varid{a})\;(\beta\mathbin{::}\Varid{f}\;\Varid{b}\to \Varid{t})\mathbin{::}\Conid{IsoOptic}\;\sigma\;\Varid{a}\;\Varid{b}\;\Varid{s}\;\Varid{t}}.
Then \ensuremath{\Varid{l}\mathrel{=}\Varid{injOptic}\;\alpha\;\beta\ \circ_{op}\ \Varid{enhanceIso}\;@\Varid{f}}.

\end{prop}

\begin{proof}

The derivation relies on the unusual equality law of iso optics (see
\textcite{iso_equality_rmk}):

\begin{hscode}\SaveRestoreHook
\column{B}{@{}>{\hspre}l<{\hspost}@{}}%
\column{4}{@{}>{\hspre}l<{\hspost}@{}}%
\column{14}{@{}>{\hspre}l<{\hspost}@{}}%
\column{E}{@{}>{\hspre}l<{\hspost}@{}}%
\>[B]{}\Varid{injOptic}\;\alpha\;\beta\ \circ_{op}\ \Varid{enhanceIso}\;@\Varid{f}{}\<[E]%
\\
\>[B]{}\mathrel{=}{}\<[4]%
\>[4]{}\Varid{injOptic}\;\alpha\;\beta\ \circ_{op}\ \Conid{IsoOptic}\;\Varid{id}\;\Varid{id}{}\<[E]%
\\
\>[B]{}\mathrel{=}{}\<[4]%
\>[4]{}\Conid{IsoOptic}\;(\Conid{Id}\hsdot{\circ }{.\;}\alpha)\;(\beta\hsdot{\circ }{.\;}\Varid{unId})\ \circ_{op}\ \Conid{IsoOptic}\;\Varid{id}\;\Varid{id}{}\<[E]%
\\
\>[B]{}\mathrel{=}{}\<[4]%
\>[4]{}\Conid{IsoOptic}\;{}\<[14]%
\>[14]{}(\Conid{Compose}\hsdot{\circ }{.\;}\Varid{fmap}\;@\Conid{Id}\;\Varid{id}\hsdot{\circ }{.\;}\Conid{Id}\hsdot{\circ }{.\;}\alpha)\;{}\<[E]%
\\
\>[14]{}(\beta\hsdot{\circ }{.\;}\Varid{unId}\hsdot{\circ }{.\;}\Varid{fmap}\;@\Conid{Id}\;\Varid{id}\hsdot{\circ }{.\;}\Varid{unCompose}){}\<[E]%
\\
\>[B]{}\mathrel{=}{}\<[4]%
\>[4]{}\Conid{IsoOptic}\;(\Conid{Compose}\hsdot{\circ }{.\;}\Conid{Id}\hsdot{\circ }{.\;}\alpha)\;(\beta\hsdot{\circ }{.\;}\Varid{unId}\hsdot{\circ }{.\;}\Varid{unCompose}){}\<[E]%
\\
\>[B]{}\mathrel{=}\mbox{\commentbegin   \ensuremath{\Varid{unId}\hsdot{\circ }{.\;}\Varid{unCompose}\mathbin{::}\mathbf{forall}\;\Varid{a}\hsforall \hsdot{\circ }{.\;}\Conid{Compose}\;\Conid{Id}\;\Varid{f}\;\Varid{a}\to \Varid{f}\;\Varid{a}}   \commentend}{}\<[E]%
\\
\>[B]{}\hsindent{4}{}\<[4]%
\>[4]{}\Conid{IsoOptic}\;(\Varid{unId}\hsdot{\circ }{.\;}\Varid{unCompose}\hsdot{\circ }{.\;}\Conid{Compose}\hsdot{\circ }{.\;}\Conid{Id}\hsdot{\circ }{.\;}\alpha)\;\beta{}\<[E]%
\\
\>[B]{}\mathrel{=}{}\<[4]%
\>[4]{}\Conid{IsoOptic}\;\alpha\;\beta{}\<[E]%
\\
\>[B]{}\mathrel{=}{}\<[4]%
\>[4]{}\Varid{l}{}\<[E]%
\ColumnHook
\end{hscode}\resethooks

\end{proof}

\bigskip

\begin{prop}\label{enhance_is_retraction}

If \ensuremath{\Varid{l}\mathrel{=}\Conid{IsoOptic}\;\alpha\;\beta\mathbin{::}\Conid{IsoOptic}\;\sigma\;\Varid{a}\;\Varid{b}\;(\Varid{f}\;\Varid{a})\;(\Varid{f}\;\Varid{b})} and
\ensuremath{\Varid{f}\;\mathbin{\in}\;\sigma}, then \ensuremath{\Varid{l}\mathrel{=}\Varid{enhanceIso}\;@\Varid{f}\mathbin{\mathbf{\Leftrightarrow}}\beta\hsdot{\circ }{.\;}\alpha\mathrel{=}\Varid{id}}

\end{prop}

\begin{proof}

Immediate from the definition of equality.

\end{proof}

\bigskip

\begin{prop}\label{compose_enhanceiso}

\ensuremath{\Varid{enhanceIso}\;@\Varid{f}\ \circ_{op}\ \Varid{enhanceIso}\;@\Varid{g}\mathrel{=}\Varid{injOptic}\;\Conid{Compose}\;\Varid{unCompose}\ \circ_{op}\ \Varid{enhanceIso}\;@(\Conid{Compose}\;\Varid{f}\;\Varid{g})}

\end{prop}

\begin{proof}

Straightforward from the definitions:

\begin{hscode}\SaveRestoreHook
\column{B}{@{}>{\hspre}l<{\hspost}@{}}%
\column{4}{@{}>{\hspre}l<{\hspost}@{}}%
\column{E}{@{}>{\hspre}l<{\hspost}@{}}%
\>[B]{}\Varid{enhanceIso}\;@\Varid{f}\ \circ_{op}\ \Varid{enhanceIso}\;@\Varid{g}{}\<[E]%
\\
\>[B]{}\mathrel{=}{}\<[4]%
\>[4]{}(\Conid{IsoOptic}\;\Varid{id}\;\Varid{id})\ \circ_{op}\ (\Conid{IsoOptic}\;\Varid{id}\;\Varid{id}){}\<[E]%
\\
\>[B]{}\mathrel{=}{}\<[4]%
\>[4]{}\Conid{IsoOptic}\;(\Conid{Compose}\hsdot{\circ }{.\;}\Varid{fmap}\;\Varid{id}\hsdot{\circ }{.\;}\Varid{id})\;(\Varid{id}\hsdot{\circ }{.\;}\Varid{fmap}\;\Varid{id}\hsdot{\circ }{.\;}\Varid{unCompose}){}\<[E]%
\\
\>[B]{}\mathrel{=}{}\<[4]%
\>[4]{}\Conid{IsoOptic}\;\Conid{Compose}\;\Varid{unCompose}{}\<[E]%
\\
\>[B]{}\mathrel{=}\mbox{\commentbegin   \Cref{iso_normal_form}   \commentend}{}\<[E]%
\\
\>[B]{}\hsindent{4}{}\<[4]%
\>[4]{}\Varid{injOptic}\;\Conid{Compose}\;\Varid{unCompose}\ \circ_{op}\ \Varid{enhanceIso}\;@(\Conid{Compose}\;\Varid{f}\;\Varid{g}){}\<[E]%
\ColumnHook
\end{hscode}\resethooks

\end{proof}

\bigskip

\hypertarget{arrows-of-isomorphism-optics}{%
\section{Arrows of isomorphism
optics}\label{arrows-of-isomorphism-optics}}

Since iso optics are defined using existentials (i.e.~coends), we expect
arrows from iso optics to have interesting structure. Indeed, an arrow
\ensuremath{\theta\mathbin{::}\Conid{IsoOptic}\;\sigma\leadsto \Varid{op}} is uniquely determined by its image of
\ensuremath{\Varid{enhanceIso}}. We explore some results related to this fact. \bigskip
\bigskip

\begin{prop}\label{theta_normal_form}

If \ensuremath{\theta\mathbin{::}\Conid{IsoOptic}\;\sigma\leadsto \Varid{op}}, then
\ensuremath{\theta\;(\Conid{IsoOptic}\;(\alpha\mathbin{::}\Varid{s}\to \Varid{f}\;\Varid{a})\;(\beta\mathbin{::}\Varid{f}\;\Varid{b}\to \Varid{t}))\mathrel{=}\Varid{injOptic}\;\alpha\;\beta\hsdot{\circ }{.\;}\theta\;(\Varid{enhanceIso}\;@\Varid{f})}

\end{prop}

\begin{proof}

Let
\ensuremath{\Varid{l}\mathrel{=}\Conid{IsoOptic}\;(\alpha\mathbin{::}\Varid{s}\to \Varid{f}\;\Varid{a})\;(\beta\mathbin{::}\Varid{f}\;\Varid{b}\to \Varid{t})\mathbin{::}\Conid{IsoOptic}\;\sigma\;\Varid{a}\;\Varid{b}\;\Varid{s}\;\Varid{t}}.

\begin{hscode}\SaveRestoreHook
\column{B}{@{}>{\hspre}l<{\hspost}@{}}%
\column{3}{@{}>{\hspre}l<{\hspost}@{}}%
\column{E}{@{}>{\hspre}l<{\hspost}@{}}%
\>[B]{}\theta\;\Varid{l}{}\<[E]%
\\
\>[B]{}\mathrel{=}\mbox{\commentbegin   \Cref{iso_normal_form}   \commentend}{}\<[E]%
\\
\>[B]{}\hsindent{3}{}\<[3]%
\>[3]{}\theta\;(\Varid{injOptic}\;\alpha\;\beta\ \circ_{op}\ \Varid{enhanceIso}\;@\Varid{f}){}\<[E]%
\\
\>[B]{}\mathrel{=}\mbox{\commentbegin   \ensuremath{\theta} is a morphism of optic families   \commentend}{}\<[E]%
\\
\>[B]{}\hsindent{3}{}\<[3]%
\>[3]{}\theta\;(\Varid{injOptic}\;\alpha\;\beta)\ \circ_{op}\ \theta\;(\Varid{enhanceIso}\;@\Varid{f}){}\<[E]%
\\
\>[B]{}\mathrel{=}\mbox{\commentbegin   \ensuremath{\theta} is a morphism of optic families   \commentend}{}\<[E]%
\\
\>[B]{}\hsindent{3}{}\<[3]%
\>[3]{}\Varid{injOptic}\;\alpha\;\beta\ \circ_{op}\ \theta\;(\Varid{enhanceIso}\;@\Varid{f}){}\<[E]%
\ColumnHook
\end{hscode}\resethooks

\end{proof}

\bigskip

\begin{crly}\label{match_enhanceiso_is_enough}

If \ensuremath{\theta,\theta'\mathbin{::}\Conid{IsoOptic}\;\sigma\leadsto \Varid{op}} and for all \ensuremath{\Varid{f}\;\mathbin{\in}\;\sigma},
\ensuremath{\theta\;(\Varid{enhanceIso}\;@\Varid{f})\mathrel{=}\theta'\;(\Varid{enhanceIso}\;@\Varid{f})}, then \ensuremath{\theta\mathrel{=}\theta'}.

\end{crly}

\begin{proof}

Trivially from \textcite{theta_normal_form} above.

\end{proof}

\bigskip

\begin{prop}\label{theta_preserve_enhanceiso}

If \ensuremath{\theta\mathbin{::}\Conid{IsoOptic}\;\sigma\leadsto \Conid{IsoOptic}\;\sigma'} and
\ensuremath{\Varid{f}\;\mathbin{\in}\;\sigma\;\cap\;\sigma'}, then \ensuremath{\theta\;(\Varid{enhanceIso}\;@\Varid{f})\mathrel{=}\Varid{enhanceIso}\;@\Varid{f}}.

\end{prop}

\begin{proof}

Let \ensuremath{\Conid{IsoOptic}\;\alpha\;\beta\mathrel{=}\theta\;(\Varid{enhanceIso}\;@\Varid{f})}.

\begin{hscode}\SaveRestoreHook
\column{B}{@{}>{\hspre}l<{\hspost}@{}}%
\column{E}{@{}>{\hspre}l<{\hspost}@{}}%
\>[B]{}\Varid{mapOptic}\;(\theta\;(\Varid{enhanceIso}\;@\Varid{f}))\;\Varid{id}{}\<[E]%
\\
\>[B]{}\mathrel{=}\beta\hsdot{\circ }{.\;}\Varid{fmap}\;@\Varid{g}\;\Varid{id}\hsdot{\circ }{.\;}\alpha{}\<[E]%
\\
\>[B]{}\mathrel{=}\beta\hsdot{\circ }{.\;}\alpha{}\<[E]%
\ColumnHook
\end{hscode}\resethooks

\ensuremath{\theta} preserves \ensuremath{\Varid{mapOptic}}, therefore:

\begin{hscode}\SaveRestoreHook
\column{B}{@{}>{\hspre}l<{\hspost}@{}}%
\column{E}{@{}>{\hspre}l<{\hspost}@{}}%
\>[B]{}\Varid{mapOptic}\;(\theta\;(\Varid{enhanceIso}\;@\Varid{f}))\;\Varid{id}{}\<[E]%
\\
\>[B]{}\mathrel{=}\Varid{mapOptic}\;(\Varid{enhanceIso}\;@\Varid{f})\;\Varid{id}{}\<[E]%
\\
\>[B]{}\mathrel{=}\Varid{id}\hsdot{\circ }{.\;}\Varid{fmap}\;@\Varid{f}\;\Varid{id}\hsdot{\circ }{.\;}\Varid{id}{}\<[E]%
\\
\>[B]{}\mathrel{=}\Varid{id}{}\<[E]%
\ColumnHook
\end{hscode}\resethooks

Therefore \ensuremath{\beta\hsdot{\circ }{.\;}\alpha\mathrel{=}\Varid{id}}, and by \textcite{enhance_is_retraction},
\ensuremath{\theta\;(\Varid{enhanceIso}\;@\Varid{f})\mathrel{=}\Varid{enhanceIso}\;@\Varid{f}}.

\end{proof}

\bigskip

\begin{crly}\label{endo_iso_is_singleton}

If \ensuremath{\theta\mathbin{::}\Conid{IsoOptic}\;\sigma\leadsto \Conid{IsoOptic}\;\sigma}, then \ensuremath{\theta\mathrel{=}\Varid{id}}.

\end{crly}

\begin{proof}

For all \ensuremath{\Varid{f}\;\mathbin{\in}\;\sigma}, using \textcite{theta_preserve_enhanceiso},
\ensuremath{\theta\;(\Varid{enhanceIso}\;@\Varid{f})\mathrel{=}\Varid{enhanceIso}\;@\Varid{f}}. Using
\textcite{match_enhanceiso_is_enough}, we get \ensuremath{\theta\mathrel{=}\Varid{id}}.

\end{proof}

\bigskip
\bigskip

The fact that such arrows are determined by their image of \ensuremath{\Varid{enhanceIso}}
can also be seen by some equational reasoning on the type of those
arrows:

\begin{hscode}\SaveRestoreHook
\column{B}{@{}>{\hspre}l<{\hspost}@{}}%
\column{4}{@{}>{\hspre}l<{\hspost}@{}}%
\column{E}{@{}>{\hspre}l<{\hspost}@{}}%
\>[B]{}\mathbf{forall}\;\Varid{a}\hsforall \;\Varid{b}\;\Varid{s}\;\Varid{t}\hsdot{\circ }{.\;}\Conid{IsoOptic}\;\sigma\;\Varid{a}\;\Varid{b}\;\Varid{s}\;\Varid{t}\to \Varid{op}\;\Varid{a}\;\Varid{b}\;\Varid{s}\;\Varid{t}{}\<[E]%
\\
\>[B]{}\cong \mbox{\commentbegin   definition   \commentend}{}\<[E]%
\\
\>[B]{}\hsindent{4}{}\<[4]%
\>[4]{}\mathbf{forall}\;\Varid{a}\hsforall \;\Varid{b}\;\Varid{s}\;\Varid{t}\hsdot{\circ }{.\;}(\mathbf{exists}\;\Varid{f}\hsexists \hsdot{\circ }{.\;}\sigma\;\Varid{f}\Rightarrow (\Varid{s}\to \Varid{f}\;\Varid{a},\Varid{f}\;\Varid{b}\to \Varid{t}))\to \Varid{op}\;\Varid{a}\;\Varid{b}\;\Varid{s}\;\Varid{t}{}\<[E]%
\\
\>[B]{}\cong \mbox{\commentbegin   fundamental property of quantification   \commentend}{}\<[E]%
\\
\>[B]{}\hsindent{4}{}\<[4]%
\>[4]{}\mathbf{forall}\;\Varid{f}\hsforall \;\Varid{a}\;\Varid{b}\;\Varid{s}\;\Varid{t}\hsdot{\circ }{.\;}\sigma\;\Varid{f}\Rightarrow (\Varid{s}\to \Varid{f}\;\Varid{a},\Varid{f}\;\Varid{b}\to \Varid{t})\to \Varid{op}\;\Varid{a}\;\Varid{b}\;\Varid{s}\;\Varid{t}{}\<[E]%
\\
\>[B]{}\cong \mbox{\commentbegin   Yoneda lemma, since \ensuremath{\Varid{op}\;\Varid{a}\;\Varid{b}\;\Varid{s}\;\Varid{t}} is covariant in \ensuremath{\Varid{t}}   \commentend}{}\<[E]%
\\
\>[B]{}\hsindent{4}{}\<[4]%
\>[4]{}\mathbf{forall}\;\Varid{f}\hsforall \;\Varid{a}\;\Varid{b}\;\Varid{s}\hsdot{\circ }{.\;}\sigma\;\Varid{f}\Rightarrow (\Varid{s}\to \Varid{f}\;\Varid{a})\to \Varid{op}\;\Varid{a}\;\Varid{b}\;\Varid{s}\;(\Varid{f}\;\Varid{b}){}\<[E]%
\\
\>[B]{}\cong \mbox{\commentbegin   Yoneda lemma, since \ensuremath{\Varid{op}\;\Varid{a}\;\Varid{b}\;\Varid{s}\;\Varid{t}} is contravariant in \ensuremath{\Varid{s}}  \commentend}{}\<[E]%
\\
\>[B]{}\hsindent{4}{}\<[4]%
\>[4]{}\mathbf{forall}\;\Varid{f}\hsforall \;\Varid{a}\;\Varid{b}\hsdot{\circ }{.\;}\sigma\;\Varid{f}\Rightarrow \Varid{op}\;\Varid{a}\;\Varid{b}\;(\Varid{f}\;\Varid{a})\;(\Varid{f}\;\Varid{b}){}\<[E]%
\ColumnHook
\end{hscode}\resethooks

This equivalence states that the type of arrows \ensuremath{\Conid{IsoOptic}\;\sigma\leadsto \Varid{op}}
is isomorphic to the type of the image of \ensuremath{\Varid{enhanceIso}} by such an arrow.

This last type can also be read as there being an optic that can zoom
into the contents of all the functors of the family \ensuremath{\sigma}. We dub this
property ``enhanceability'' of \ensuremath{\Varid{op}} by \ensuremath{\sigma}. \bigskip

\begin{defn}[Enhanceability]

An optic family \ensuremath{\Varid{op}} is said to be \emph{enhanceable} by a functor
monoid \ensuremath{\sigma} if there exists a value:

\begin{hscode}\SaveRestoreHook
\column{B}{@{}>{\hspre}l<{\hspost}@{}}%
\column{E}{@{}>{\hspre}l<{\hspost}@{}}%
\>[B]{}\Varid{enhanceOp}\mathbin{::}\mathbf{forall}\;\Varid{f}\hsforall \;\Varid{a}\;\Varid{b}\hsdot{\circ }{.\;}\sigma\;\Varid{f}\Rightarrow \Varid{op}\;\Varid{a}\;\Varid{b}\;(\Varid{f}\;\Varid{a})\;(\Varid{f}\;\Varid{b}){}\<[E]%
\ColumnHook
\end{hscode}\resethooks

that verifies laws similar to the \ensuremath{\Conid{Enhancing}} laws:

\begin{itemize}
\tightlist
\item
  \ensuremath{\Varid{enhanceOp}\;@\Conid{Id}\mathrel{=}\Varid{injOptic}\;\Varid{unId}\;\Conid{Id}}
\item
  \ensuremath{\Varid{enhanceOp}\;@(\Conid{Compose}\;\Varid{f}\;\Varid{g})\mathrel{=}\Varid{injOptic}\;\Varid{unCompose}\;\Conid{Compose}\ \circ_{op}\ \Varid{enhanceOp}\;@\Varid{f}\ \circ_{op}\ \Varid{enhanceOp}\;@\Varid{g}}
\item
  \ensuremath{\alpha\mathbin{::}\mathbf{forall}\;\Varid{a}\hsforall \hsdot{\circ }{.\;}\Varid{f}\;\Varid{a}\to \Varid{g}\;\Varid{a}\Rightarrow \Varid{injOptic}\;\Varid{id}\;\alpha\ \circ_{op}\ \Varid{enhanceOp}\;@\Varid{f}\mathrel{=}\Varid{injOptic}\;\alpha\;\Varid{id}\ \circ_{op}\ \Varid{enhanceOp}\;@\Varid{g}}
\item
  \ensuremath{\Varid{enhanceOp}\;@\Varid{f}\ \circ_{op}\ \Varid{injOptic}\;\Varid{f}\;\Varid{g}\mathrel{=}\Varid{injOptic}\;(\Varid{fmap}\;\Varid{f})\;(\Varid{fmap}\;\Varid{g})\ \circ_{op}\ \Varid{enhanceOp}\;@\Varid{f}}
\end{itemize}

as well as an additional law:

\begin{itemize}
\tightlist
\item
  \ensuremath{\Varid{mapOptic}\;\Varid{enhanceOp}\mathrel{=}\Varid{fmap}}
\end{itemize}

We capture this notion in the \ensuremath{\Conid{Enhanceable}} typeclass:

\begin{hscode}\SaveRestoreHook
\column{B}{@{}>{\hspre}l<{\hspost}@{}}%
\column{5}{@{}>{\hspre}l<{\hspost}@{}}%
\column{E}{@{}>{\hspre}l<{\hspost}@{}}%
\>[B]{}\mathbf{class}\;(\Conid{OpticFamily}\;\Varid{op},\Conid{FunctorMonoid}\;\sigma)\Rightarrow \Conid{Enhanceable}\;\sigma\;\Varid{op}\;\mathbf{where}{}\<[E]%
\\
\>[B]{}\hsindent{5}{}\<[5]%
\>[5]{}\Varid{enhanceOp}\mathbin{::}\mathbf{forall}\;\Varid{f}\hsforall \;\Varid{a}\;\Varid{b}\hsdot{\circ }{.\;}\sigma\;\Varid{f}\Rightarrow \Varid{op}\;\Varid{a}\;\Varid{b}\;(\Varid{f}\;\Varid{a})\;(\Varid{f}\;\Varid{b}){}\<[E]%
\ColumnHook
\end{hscode}\resethooks

\end{defn}

\bigskip

\begin{rmk}\label{map_enhance_id}

By parametricity, the last law is equivalent to requiring only
\ensuremath{\Varid{mapOptic}\;\Varid{enhanceOp}\;\Varid{id}\mathrel{=}\Varid{id}}.

\end{rmk}

\bigskip

\begin{prop}\label{enhanceable_preserved_by_opfam_maps}

If \ensuremath{\theta\mathbin{::}\Varid{op}\leadsto \Varid{op'}} and \ensuremath{\Varid{op}\;\mathbin{\in}\;\Conid{Enhanceable}\;\sigma}, then
\ensuremath{\theta\;\Varid{enhanceOp}} defines a lawful instance of \ensuremath{\Conid{Enhanceable}\;\sigma\;\Varid{op'}}.

\end{prop}

\begin{proof}

It is easy to see that all the \ensuremath{\Conid{Enhanceable}} laws are preserved by
morphisms of optic families.

\end{proof}

\bigskip

\begin{prop}\label{enhanceable_isooptic}

\ensuremath{\Varid{enhanceIso}} defines a lawful instance of
\ensuremath{\Conid{Enhanceable}\;\sigma\;(\Conid{IsoOptic}\;\sigma)}.

\end{prop}

\begin{proof}

Straightforward using \textcite{compose_enhanceiso}.

\end{proof}

\bigskip

\begin{lemma}\label{enhanceable_profoptic}

\ensuremath{\Varid{enhance}} defines a lawful instance of
\ensuremath{\Conid{Enhanceable}\;\sigma\;(\Conid{ProfOptic}\;\sigma)}.

\end{lemma}

\begin{proof}

The first four of the \ensuremath{\Conid{Enhanceable}} laws correspond exactly the
\ensuremath{\Conid{Enhancing}} laws.

We only have to check the last law:

\begin{hscode}\SaveRestoreHook
\column{B}{@{}>{\hspre}l<{\hspost}@{}}%
\column{4}{@{}>{\hspre}l<{\hspost}@{}}%
\column{E}{@{}>{\hspre}l<{\hspost}@{}}%
\>[B]{}\Varid{mapOptic}\;\Varid{enhance}{}\<[E]%
\\
\>[B]{}\mathrel{=}\mbox{\commentbegin   \ensuremath{\mathbf{instance}\;\Conid{OpticFamily}\;(\Conid{ProfOptic}\;\sigma)}   \commentend}{}\<[E]%
\\
\>[B]{}\hsindent{4}{}\<[4]%
\>[4]{}\Varid{enhance}\;@(\to ){}\<[E]%
\\
\>[B]{}\mathrel{=}\mbox{\commentbegin   \ensuremath{\mathbf{instance}\;\Conid{Enhancing}\;\sigma\;(\to )}   \commentend}{}\<[E]%
\\
\>[B]{}\hsindent{4}{}\<[4]%
\>[4]{}\Varid{fmap}{}\<[E]%
\ColumnHook
\end{hscode}\resethooks

\end{proof}

\bigskip

If \ensuremath{\Varid{op}} is enhanceable by \ensuremath{\sigma}, we can recover an
\ensuremath{\Conid{IsoOptic}\;\sigma\leadsto \Varid{op}} arrow as follows:

\begin{hscode}\SaveRestoreHook
\column{B}{@{}>{\hspre}l<{\hspost}@{}}%
\column{E}{@{}>{\hspre}l<{\hspost}@{}}%
\>[B]{}\Varid{enhanceToArrow}\mathbin{::}\Conid{Enhanceable}\;\sigma\;\Varid{op}\Rightarrow \Conid{IsoOptic}\;\sigma\leadsto \Varid{op}{}\<[E]%
\\
\>[B]{}\Varid{enhanceToArrow}\;(\Conid{IsoOptic}\;\alpha\;\beta)\mathrel{=}\Varid{injOptic}\;\alpha\;\beta\ \circ_{op}\ \Varid{enhanceOp}{}\<[E]%
\ColumnHook
\end{hscode}\resethooks

\bigskip

\begin{prop}\label{enhancetoarrow_is_opfam_morphism}

\ensuremath{\Varid{enhanceToArrow}} is a morphism of optic families.

\end{prop}

\begin{proof}

See \textcite{proof_enhancetoarrow_is_opfam_morphism}. The proof
crucially depends on the \ensuremath{\Conid{Enhancing}} laws.

\end{proof}

\bigskip

\hypertarget{examples-1}{%
\section{Examples}\label{examples-1}}

\begin{prop}

\ensuremath{\Conid{IsoOptic}\;((\mathord{\sim})\;\Conid{Id})} is isomorphic to \ensuremath{\Conid{Adapter}}.

\end{prop}

\begin{proof}

Trivially,
\ensuremath{\Conid{IsoOptic}\;((\mathord{\sim})\;\Conid{Id})\;\Varid{a}\;\Varid{b}\;\Varid{s}\;\Varid{t}\cong (\Varid{s}\to \Conid{Id}\;\Varid{a},\Conid{Id}\;\Varid{b}\to \Varid{t})\cong (\Varid{s}\to \Varid{a},\Varid{b}\to \Varid{t})}.

\end{proof}

\bigskip

\begin{prop}

\ensuremath{\Conid{IsoOptic}\;\Conid{IsProduct}} is isomorphic to \ensuremath{\Conid{Lens}}.

\end{prop}

\begin{proof}

The isomorphism is as follows:

\begin{hscode}\SaveRestoreHook
\column{B}{@{}>{\hspre}l<{\hspost}@{}}%
\column{5}{@{}>{\hspre}l<{\hspost}@{}}%
\column{9}{@{}>{\hspre}l<{\hspost}@{}}%
\column{E}{@{}>{\hspre}l<{\hspost}@{}}%
\>[B]{}\Varid{lensToIso}\mathbin{::}\Conid{Lens}\;\Varid{a}\;\Varid{b}\;\Varid{s}\;\Varid{t}\to \Conid{IsoOptic}\;\Conid{IsProduct}\;\Varid{a}\;\Varid{b}\;\Varid{s}\;\Varid{t}{}\<[E]%
\\
\>[B]{}\Varid{lensToIso}\;(\Conid{Lens}\;\Varid{get}\;\Varid{put})\mathrel{=}\Conid{IsoOptic}\;\alpha\;\beta{}\<[E]%
\\
\>[B]{}\hsindent{5}{}\<[5]%
\>[5]{}\mathbf{where}{}\<[E]%
\\
\>[5]{}\hsindent{4}{}\<[9]%
\>[9]{}\alpha\mathbin{::}\Varid{s}\to (\Varid{s},\Varid{a}){}\<[E]%
\\
\>[5]{}\hsindent{4}{}\<[9]%
\>[9]{}\alpha\;\Varid{s}\mathrel{=}(\Varid{s},\Varid{get}\;\Varid{s}){}\<[E]%
\\
\>[5]{}\hsindent{4}{}\<[9]%
\>[9]{}\beta\mathbin{::}(\Varid{s},\Varid{b})\to \Varid{t}{}\<[E]%
\\
\>[5]{}\hsindent{4}{}\<[9]%
\>[9]{}\beta\;(\Varid{s},\Varid{b})\mathrel{=}\Varid{put}\;\Varid{b}\;\Varid{s}{}\<[E]%
\ColumnHook
\end{hscode}\resethooks

We use here the \ensuremath{(\Varid{s},\mathbin{-})} functor, which is a member of \ensuremath{\Conid{IsProduct}}.

\begin{hscode}\SaveRestoreHook
\column{B}{@{}>{\hspre}l<{\hspost}@{}}%
\column{5}{@{}>{\hspre}l<{\hspost}@{}}%
\column{9}{@{}>{\hspre}l<{\hspost}@{}}%
\column{17}{@{}>{\hspre}l<{\hspost}@{}}%
\column{E}{@{}>{\hspre}l<{\hspost}@{}}%
\>[B]{}\Varid{isoToLens}\mathbin{::}\Conid{IsoOptic}\;\Conid{IsProduct}\;\Varid{a}\;\Varid{b}\;\Varid{s}\;\Varid{t}\to \Conid{Lens}\;\Varid{a}\;\Varid{b}\;\Varid{s}\;\Varid{t}{}\<[E]%
\\
\>[B]{}\Varid{isoToLens}\;(\Conid{IsoOptic}\;\alpha\;\beta)\mathrel{=}\Conid{Lens}\;\Varid{get}\;\Varid{put}{}\<[E]%
\\
\>[B]{}\hsindent{5}{}\<[5]%
\>[5]{}\mathbf{where}{}\<[E]%
\\
\>[5]{}\hsindent{4}{}\<[9]%
\>[9]{}\Varid{get}\mathbin{::}\Varid{s}\to \Varid{a}{}\<[E]%
\\
\>[5]{}\hsindent{4}{}\<[9]%
\>[9]{}\Varid{get}\mathrel{=}\Varid{snd}\hsdot{\circ }{.\;}\Varid{toProduct}\hsdot{\circ }{.\;}\alpha{}\<[E]%
\\
\>[5]{}\hsindent{4}{}\<[9]%
\>[9]{}\Varid{put}\mathbin{::}{}\<[17]%
\>[17]{}\Varid{b}\to \Varid{s}\to \Varid{t}{}\<[E]%
\\
\>[5]{}\hsindent{4}{}\<[9]%
\>[9]{}\Varid{put}\;\Varid{b}\mathrel{=}\beta\hsdot{\circ }{.\;}\Varid{fromProduct}\hsdot{\circ }{.\;}(\lambda (\Varid{x},\anonymous )\to (\Varid{x},\Varid{b}))\hsdot{\circ }{.\;}\Varid{toProduct}\hsdot{\circ }{.\;}\alpha{}\<[E]%
\ColumnHook
\end{hscode}\resethooks

We omit the rest of the proof; see
\textcite{rederiving-profunctor-encodings-for-common-optics} for a more
complete derivation.

\end{proof}

\bigskip

\hypertarget{two-theorems-about-profunctor-optics}{%
\chapter{Two theorems about profunctor
optics}\label{two-theorems-about-profunctor-optics}}

Using the properties of isomorphism optics, we are now ready to prove
two important theorems about the structure of profunctor optics: first,
\ensuremath{\Conid{ProfOptic}\;\sigma} and \ensuremath{\Conid{IsoOptic}\;\sigma} are isomorphic as optic families
(\textcite{representation_theorem}); second, for a given optic family
there is one particular functor monoid of interest when trying to derive
a profunctor encoding (\textcite{functorization_theorem}).

\hypertarget{the-representation-theorem}{%
\section{The representation theorem}\label{the-representation-theorem}}

\begin{lemma}\label{enhanceable_implies_enhancing}

If \ensuremath{\Varid{op}\;\mathbin{\in}\;\Conid{Enhanceable}\;\sigma}, then \ensuremath{\Varid{op}\;\Varid{a}\;\Varid{b}\;\mathbin{\in}\;\Conid{Enhancing}\;\sigma}.

\end{lemma}

\begin{proof}

We already know that \ensuremath{\Varid{op}\;\Varid{a}\;\Varid{b}} is a profunctor, from the \ensuremath{\Conid{OpticFamily}}
instance:

\begin{hscode}\SaveRestoreHook
\column{B}{@{}>{\hspre}l<{\hspost}@{}}%
\column{5}{@{}>{\hspre}l<{\hspost}@{}}%
\column{E}{@{}>{\hspre}l<{\hspost}@{}}%
\>[B]{}\mathbf{instance}\;\Conid{OpticFamily}\;\Varid{op}\Rightarrow \Conid{Profunctor}\;(\Varid{op}\;\Varid{a}\;\Varid{b})\;\mathbf{where}{}\<[E]%
\\
\>[B]{}\hsindent{5}{}\<[5]%
\>[5]{}\Varid{dimap}\mathrel{=}\Varid{dimapOptic}{}\<[E]%
\ColumnHook
\end{hscode}\resethooks

We define \ensuremath{\Varid{enhance}} using \ensuremath{\Varid{enhanceOp}} as follows:

\begin{hscode}\SaveRestoreHook
\column{B}{@{}>{\hspre}l<{\hspost}@{}}%
\column{5}{@{}>{\hspre}l<{\hspost}@{}}%
\column{E}{@{}>{\hspre}l<{\hspost}@{}}%
\>[B]{}\mathbf{instance}\;\Conid{Enhanceable}\;\sigma\;\Varid{op}\Rightarrow \Conid{Enhancing}\;\sigma\;(\Varid{op}\;\Varid{a}\;\Varid{b})\;\mathbf{where}{}\<[E]%
\\
\>[B]{}\hsindent{5}{}\<[5]%
\>[5]{}\Varid{enhance}\mathbin{::}\sigma\;\Varid{f}\Rightarrow \Varid{op}\;\Varid{a}\;\Varid{b}\;\Varid{s}\;\Varid{t}\to \Varid{op}\;\Varid{a}\;\Varid{b}\;(\Varid{f}\;\Varid{s})\;(\Varid{f}\;\Varid{t}){}\<[E]%
\\
\>[B]{}\hsindent{5}{}\<[5]%
\>[5]{}\Varid{enhance}\;\Varid{l}\mathrel{=}\Varid{enhanceOp}\ \circ_{op}\ \Varid{l}{}\<[E]%
\ColumnHook
\end{hscode}\resethooks

The laws are straightforward from the \ensuremath{\Conid{Enhanceable}} laws:

\begin{hscode}\SaveRestoreHook
\column{B}{@{}>{\hspre}l<{\hspost}@{}}%
\column{4}{@{}>{\hspre}l<{\hspost}@{}}%
\column{E}{@{}>{\hspre}l<{\hspost}@{}}%
\>[B]{}\Varid{enhance}\;@\Conid{Id}\;\Varid{l}{}\<[E]%
\\
\>[B]{}\mathrel{=}{}\<[4]%
\>[4]{}\Varid{enhanceOp}\;@\Conid{Id}\ \circ_{op}\ \Varid{l}{}\<[E]%
\\
\>[B]{}\mathrel{=}\mbox{\commentbegin   \ensuremath{\Conid{Enhanceable}} law   \commentend}{}\<[E]%
\\
\>[B]{}\hsindent{4}{}\<[4]%
\>[4]{}\Varid{injOptic}\;\Varid{unId}\;\Conid{Id}\ \circ_{op}\ \Varid{l}{}\<[E]%
\\
\>[B]{}\mathrel{=}{}\<[4]%
\>[4]{}\Varid{dimapOptic}\;\Varid{unId}\;\Conid{Id}\;\Varid{l}{}\<[E]%
\\
\>[B]{}\mathrel{=}{}\<[4]%
\>[4]{}\Varid{dimap}\;\Varid{unId}\;\Conid{Id}\;\Varid{l}{}\<[E]%
\ColumnHook
\end{hscode}\resethooks

And similarly for the other two laws.

\end{proof}

\bigskip

\begin{thm}[Representation theorem for profunctor optics]\label{representation_theorem}

For a given functor monoid \ensuremath{\sigma}, there exists an isomorphism of optic
families between \ensuremath{\Conid{ProfOptic}\;\sigma} and \ensuremath{\Conid{IsoOptic}\;\sigma}.

\end{thm}

\begin{proof}

From \textcite{enhanceable_profoptic},
\ensuremath{\Conid{ProfOptic}\;\sigma\;\mathbin{\in}\;\Conid{Enhanceable}\;\sigma} so we get a
\ensuremath{\Conid{IsoOptic}\;\sigma\leadsto \Conid{ProfOptic}\;\sigma} arrow:

\begin{hscode}\SaveRestoreHook
\column{B}{@{}>{\hspre}l<{\hspost}@{}}%
\column{E}{@{}>{\hspre}l<{\hspost}@{}}%
\>[B]{}\Varid{isoToProf}\mathbin{::}\Conid{IsoOptic}\;\sigma\leadsto \Conid{ProfOptic}\;\sigma{}\<[E]%
\\
\>[B]{}\Varid{isoToProf}\mathrel{=}\Varid{enhanceToArrow}{}\<[E]%
\ColumnHook
\end{hscode}\resethooks

To construct the reverse transformation, we are given an
\ensuremath{\Varid{l}\mathbin{::}\Conid{ProfOptic}\;\sigma\;\Varid{a}\;\Varid{b}\;\Varid{s}\;\Varid{t}}. To use it, we need to instantiate it to
an appropriate instance of \ensuremath{\Conid{Enhancing}\;\sigma}.

Using \textcite{enhanceable_isooptic} and
\textcite{enhanceable_implies_enhancing}, we deduce that
\ensuremath{\Conid{IsoOptic}\;\sigma\;\Varid{a}\;\Varid{b}\;\mathbin{\in}\;\Conid{Enhancing}\;\sigma}. We can thus specialize the
type of \ensuremath{\Varid{l}} to \ensuremath{\Varid{l}\mathbin{::}\Conid{IsoOptic}\;\sigma\;\Varid{a}\;\Varid{b}\;\Varid{a}\;\Varid{b}\to \Conid{IsoOptic}\;\sigma\;\Varid{a}\;\Varid{b}\;\Varid{s}\;\Varid{t}}.

We can therefore apply it to \ensuremath{\Varid{id}_{\Varid{op}}\mathbin{::}\Conid{IsoOptic}\;\sigma\;\Varid{a}\;\Varid{b}\;\Varid{a}\;\Varid{b}}:

\begin{hscode}\SaveRestoreHook
\column{B}{@{}>{\hspre}l<{\hspost}@{}}%
\column{E}{@{}>{\hspre}l<{\hspost}@{}}%
\>[B]{}\Varid{profToIso}\mathbin{::}\Conid{ProfOptic}\;\sigma\;\Varid{a}\;\Varid{b}\;\Varid{s}\;\Varid{t}\to \Conid{IsoOptic}\;\sigma\;\Varid{a}\;\Varid{b}\;\Varid{s}\;\Varid{t}{}\<[E]%
\\
\>[B]{}\Varid{profToIso}\;\Varid{l}\mathrel{=}\Varid{l}\;\Varid{id}_{\Varid{op}}{}\<[E]%
\ColumnHook
\end{hscode}\resethooks

We have not proved that \ensuremath{\Varid{profToIso}} is a morphism of optic families yet,
hence we cannot use results such as \textcite{endo_iso_is_singleton} to
prove that \ensuremath{\Varid{profToIso}} and \ensuremath{\Varid{isoToProf}} are inverses. \medskip

Let \ensuremath{\Varid{pab}\mathbin{::}\Varid{p}\;\Varid{a}\;\Varid{b}} for some \ensuremath{\Varid{p}\;\mathbin{\in}\;\Conid{Enhancing}\;\sigma}.

By parametricity,
\ensuremath{\Varid{flip}\;\Varid{isoToProf}\;\Varid{pab}\mathbin{::}\mathbf{forall}\;\Varid{s}\hsforall \;\Varid{t}\hsdot{\circ }{.\;}\Conid{IsoOptic}\;\sigma\;\Varid{a}\;\Varid{b}\;\Varid{s}\;\Varid{t}\to \Varid{p}\;\Varid{s}\;\Varid{t}} is a
natural transformation between the profunctors \ensuremath{\Conid{IsoOptic}\;\sigma\;\Varid{a}\;\Varid{b}} and
\ensuremath{\Varid{p}}. It is also a natural transformation between \ensuremath{\Conid{Enhancing}\;\sigma}
profunctors because it preserves \ensuremath{\Varid{enhance}}:

\begin{hscode}\SaveRestoreHook
\column{B}{@{}>{\hspre}l<{\hspost}@{}}%
\column{4}{@{}>{\hspre}l<{\hspost}@{}}%
\column{E}{@{}>{\hspre}l<{\hspost}@{}}%
\>[B]{}(\Varid{flip}\;\Varid{isoToProf}\;\Varid{pab}\hsdot{\circ }{.\;}\Varid{enhance}\;@\Varid{f})\;\Varid{l}{}\<[E]%
\\
\>[B]{}\mathrel{=}{}\<[4]%
\>[4]{}\Varid{flip}\;\Varid{isoToProf}\;\Varid{pab}\;(\Varid{enhance}\;@\Varid{f}\;\Varid{l}){}\<[E]%
\\
\>[B]{}\mathrel{=}\mbox{\commentbegin   \ensuremath{\mathbf{instance}\;\Conid{Enhancing}\;\sigma\;(\Conid{IsoOptic}\;\sigma)}   \commentend}{}\<[E]%
\\
\>[B]{}\hsindent{4}{}\<[4]%
\>[4]{}\Varid{isoToProf}\;(\Varid{enhanceIso}\;@\Varid{f}\ \circ_{op}\ \Varid{l})\;\Varid{pab}{}\<[E]%
\\
\>[B]{}\mathrel{=}\mbox{\commentbegin   \ensuremath{\Varid{isoToProf}} is a morphism of optic families   \commentend}{}\<[E]%
\\
\>[B]{}\hsindent{4}{}\<[4]%
\>[4]{}(\Varid{isoToProf}\;(\Varid{enhanceIso}\;@\Varid{f})\ \circ_{op}\ \Varid{isoToProf}\;\Varid{l})\;\Varid{pab}{}\<[E]%
\\
\>[B]{}\mathrel{=}\mbox{\commentbegin   \ensuremath{\mathbf{instance}\;\Conid{OpticFamily}\;(\Conid{ProfOptic}\;\sigma)}   \commentend}{}\<[E]%
\\
\>[B]{}\hsindent{4}{}\<[4]%
\>[4]{}(\Varid{isoToProf}\;(\Varid{enhanceIso}\;@\Varid{f})\hsdot{\circ }{.\;}\Varid{isoToProf}\;\Varid{l})\;\Varid{pab}{}\<[E]%
\\
\>[B]{}\mathrel{=}{}\<[4]%
\>[4]{}(\Varid{isoToProf}\;(\Conid{IsoOptic}\;\Varid{id}\;\Varid{id})\hsdot{\circ }{.\;}\Varid{isoToProf}\;\Varid{l})\;\Varid{pab}{}\<[E]%
\\
\>[B]{}\mathrel{=}\mbox{\commentbegin   definition of \ensuremath{\Varid{isoToProf}}   \commentend}{}\<[E]%
\\
\>[B]{}\hsindent{4}{}\<[4]%
\>[4]{}(\Varid{enhance}\;@\Varid{f}\hsdot{\circ }{.\;}\Varid{isoToProf}\;\Varid{l})\;\Varid{pab}{}\<[E]%
\\
\>[B]{}\mathrel{=}{}\<[4]%
\>[4]{}(\Varid{enhance}\;@\Varid{f}\hsdot{\circ }{.\;}\Varid{flip}\;\Varid{isoToProf}\;\Varid{pab})\;\Varid{l}{}\<[E]%
\ColumnHook
\end{hscode}\resethooks

Therefore, for any \ensuremath{\Varid{l}\mathbin{::}\mathbf{forall}\;\Varid{p}\hsforall \hsdot{\circ }{.\;}\Conid{Enhancing}\;\sigma\;\Varid{p}\Rightarrow \Varid{p}\;\Varid{a}\;\Varid{b}\to \Varid{p}\;\Varid{s}\;\Varid{t}},
\ensuremath{\Varid{flip}\;\Varid{isoToProf}\;\Varid{pab}\hsdot{\circ }{.\;}\Varid{l}\mathrel{=}\Varid{l}\hsdot{\circ }{.\;}\Varid{flip}\;\Varid{isoToProf}\;\Varid{pab}}.

Consequently:

\begin{hscode}\SaveRestoreHook
\column{B}{@{}>{\hspre}l<{\hspost}@{}}%
\column{4}{@{}>{\hspre}l<{\hspost}@{}}%
\column{E}{@{}>{\hspre}l<{\hspost}@{}}%
\>[B]{}\Varid{isoToProf}\;(\Varid{profToIso}\;\Varid{l})\;\Varid{pab}{}\<[E]%
\\
\>[B]{}\mathrel{=}{}\<[4]%
\>[4]{}\Varid{isoToProf}\;(\Varid{l}\;\Varid{id}_{\Varid{op}})\;\Varid{pab}{}\<[E]%
\\
\>[B]{}\mathrel{=}{}\<[4]%
\>[4]{}(\Varid{flip}\;\Varid{isoToProf}\;\Varid{pab}\hsdot{\circ }{.\;}\Varid{l})\;\Varid{id}_{\Varid{op}}{}\<[E]%
\\
\>[B]{}\mathrel{=}\mbox{\commentbegin   parametricity   \commentend}{}\<[E]%
\\
\>[B]{}\hsindent{4}{}\<[4]%
\>[4]{}(\Varid{l}\hsdot{\circ }{.\;}\Varid{flip}\;\Varid{isoToProf}\;\Varid{pab})\;\Varid{id}_{\Varid{op}}{}\<[E]%
\\
\>[B]{}\mathrel{=}{}\<[4]%
\>[4]{}(\Varid{l}\hsdot{\circ }{.\;}\Varid{isoToProf}\;\Varid{id}_{\Varid{op}})\;\Varid{pab}{}\<[E]%
\\
\>[B]{}\mathrel{=}{}\<[4]%
\>[4]{}(\Varid{l}\hsdot{\circ }{.\;}\Varid{isoToProf}\;(\Conid{IsoOptic}\;\Conid{Id}\;\Varid{unId}))\;\Varid{pab}{}\<[E]%
\\
\>[B]{}\mathrel{=}\mbox{\commentbegin   definition of \ensuremath{\Varid{isoToProf}}   \commentend}{}\<[E]%
\\
\>[B]{}\hsindent{4}{}\<[4]%
\>[4]{}(\Varid{l}\hsdot{\circ }{.\;}\Varid{dimap}\;\Conid{Id}\;\Varid{unId}\hsdot{\circ }{.\;}\Varid{enhance}\;@\Conid{Id})\;\Varid{pab}{}\<[E]%
\\
\>[B]{}\mathrel{=}\mbox{\commentbegin   \ensuremath{\Conid{Enhancing}} law   \commentend}{}\<[E]%
\\
\>[B]{}\hsindent{4}{}\<[4]%
\>[4]{}(\Varid{l}\hsdot{\circ }{.\;}\Varid{dimap}\;\Conid{Id}\;\Varid{unId}\hsdot{\circ }{.\;}\Varid{dimap}\;\Varid{unId}\;\Conid{Id})\;\Varid{pab}{}\<[E]%
\\
\>[B]{}\mathrel{=}{}\<[4]%
\>[4]{}\Varid{l}\;\Varid{pab}{}\<[E]%
\ColumnHook
\end{hscode}\resethooks

\medskip

Conversely:

\begin{hscode}\SaveRestoreHook
\column{B}{@{}>{\hspre}l<{\hspost}@{}}%
\column{4}{@{}>{\hspre}l<{\hspost}@{}}%
\column{E}{@{}>{\hspre}l<{\hspost}@{}}%
\>[B]{}\Varid{profToIso}\;(\Varid{isoToProf}\;(\Conid{IsoOptic}\;\alpha\;\beta)){}\<[E]%
\\
\>[B]{}\mathrel{=}{}\<[4]%
\>[4]{}\Varid{isoToProf}\;(\Conid{IsoOptic}\;\alpha\;\beta)\;\Varid{id}_{\Varid{op}}{}\<[E]%
\\
\>[B]{}\mathrel{=}{}\<[4]%
\>[4]{}(\Varid{injOptic}\;\alpha\;\beta\ \circ_{op}\ \Varid{enhanceOp})\;\Varid{id}_{\Varid{op}}{}\<[E]%
\\
\>[B]{}\mathrel{=}\mbox{\commentbegin   \ensuremath{\mathbf{instance}\;\Conid{OpticFamily}\;(\Conid{ProfOptic}\;\sigma)}   \commentend}{}\<[E]%
\\
\>[B]{}\hsindent{4}{}\<[4]%
\>[4]{}(\Varid{dimap}\;\alpha\;\beta\hsdot{\circ }{.\;}\Varid{enhance})\;\Varid{id}_{\Varid{op}}{}\<[E]%
\\
\>[B]{}\mathrel{=}\mbox{\commentbegin   \Cref{enhanceable_implies_enhancing}   \commentend}{}\<[E]%
\\
\>[B]{}\hsindent{4}{}\<[4]%
\>[4]{}\Varid{dimap}\;\alpha\;\beta\;(\Varid{enhanceOp}\ \circ_{op}\ \Varid{id}_{\Varid{op}}){}\<[E]%
\\
\>[B]{}\mathrel{=}{}\<[4]%
\>[4]{}\Varid{dimap}\;\alpha\;\beta\;\Varid{enhanceOp}{}\<[E]%
\\
\>[B]{}\mathrel{=}\mbox{\commentbegin   \ensuremath{\mathbf{instance}\;\Conid{Enhanceable}\;\sigma\;(\Conid{IsoOptic}\;\sigma)}   \commentend}{}\<[E]%
\\
\>[B]{}\hsindent{4}{}\<[4]%
\>[4]{}\Varid{dimap}\;\alpha\;\beta\;\Varid{enhanceIso}{}\<[E]%
\\
\>[B]{}\mathrel{=}{}\<[4]%
\>[4]{}\Varid{dimapOp}\;\alpha\;\beta\;\Varid{enhanceIso}{}\<[E]%
\\
\>[B]{}\mathrel{=}{}\<[4]%
\>[4]{}\Conid{IsoOptic}\;\alpha\;\beta{}\<[E]%
\ColumnHook
\end{hscode}\resethooks

\medskip

Hence \ensuremath{\Varid{isoToProf}} and \ensuremath{\Varid{profToIso}} are inverses. \medskip

Since \ensuremath{\Varid{isoToProf}} and \ensuremath{\Varid{profToIso}} are inverses and \ensuremath{\Varid{isoToProf}} is a
morphism of optic families, then by \textcite{invertible_opfam_morphism}
so is \ensuremath{\Varid{profToIso}}.

Therefore \ensuremath{\Conid{IsoOptic}\;\sigma} and \ensuremath{\Conid{ProfOptic}\;\sigma} are isomorphic as optic
families.

\end{proof}

\bigskip

\begin{rmk}

Intuitively, a value of type
\ensuremath{\mathbf{forall}\;\Varid{p}\hsforall \hsdot{\circ }{.\;}\Conid{Enhancing}\;\sigma\;\Varid{p}\Rightarrow \Varid{p}\;\Varid{a}\;\Varid{b}\to \Varid{p}\;\Varid{s}\;\Varid{t}} has to work for any
\ensuremath{\Varid{p}\;\mathbin{\in}\;\Conid{Enhancing}\;\sigma}, thus has to be defined using only \ensuremath{\Varid{dimap}} and
\ensuremath{\Varid{enhance}}.

Furthermore, by the profunctor and enhancing laws we know that
\ensuremath{\Varid{dimap}\;\Varid{f}\;\Varid{g}\hsdot{\circ }{.\;}\Varid{dimap}\;\Varid{k}\;\Varid{l}\mathrel{=}\Varid{dimap}\;(\Varid{k}\hsdot{\circ }{.\;}\Varid{f})\;(\Varid{g}\hsdot{\circ }{.\;}\Varid{l})},
\ensuremath{\Varid{enhance}\hsdot{\circ }{.\;}\Varid{enhance}\mathrel{=}\Varid{enhance}} (modulo wrapping/unwrapping), and
\ensuremath{\Varid{enhance}\hsdot{\circ }{.\;}\Varid{dimap}\;\Varid{f}\;\Varid{g}\mathrel{=}\Varid{dimap}\;(\Varid{fmap}\;\Varid{f})\;(\Varid{fmap}\;\Varid{g})\hsdot{\circ }{.\;}\Varid{enhance}}.

Using those three properties, it seems intuitive that given a (finite)
composition of \ensuremath{\Varid{enhance}}s and \ensuremath{\Varid{dimap}}s, we can put all the \ensuremath{\Varid{dimap}}s
after the \ensuremath{\Varid{enhance}}s, and then group them into one big
\ensuremath{\Varid{dimap}\;\Varid{f}\;\Varid{g}\hsdot{\circ }{.\;}\Varid{enhance}}.

This is exactly what the representation theorem proves: given a
profunctor optic \ensuremath{\Varid{l}}, the representation theorem states that
\ensuremath{\Varid{l}\mathrel{=}\Varid{isoToProf}\;(\Varid{profToIso}\;\Varid{l})}, which, given the definition of
\ensuremath{\Varid{isoToProf}}, means that there exists functions \ensuremath{\alpha} and \ensuremath{\beta} such
that \ensuremath{\Varid{l}\mathrel{=}\Varid{dimap}\;\alpha\;\beta\hsdot{\circ }{.\;}\Varid{enhance}}.

\end{rmk}

\bigskip

\hypertarget{the-derivation-theorem}{%
\section{The derivation theorem}\label{the-derivation-theorem}}

In \textcite{arrows-of-isomorphism-optics}, we have seen that for a
given functor monoid \ensuremath{\sigma}, the optic families that can be enhanced by
\ensuremath{\sigma} are of particular interest. If we reverse the point of view,
given an optic family \ensuremath{\Varid{op}}, we can also characterize which functor
monoids \ensuremath{\sigma} can enhance it.

From the definition of \ensuremath{\Conid{Enhancing}}, it is apparent that such a \ensuremath{\sigma}
must contain only functors for which there exists a value of type
\ensuremath{\mathbf{forall}\;\Varid{a}\hsforall \;\Varid{b}\hsdot{\circ }{.\;}\Varid{op}\;\Varid{a}\;\Varid{b}\;(\Varid{f}\;\Varid{a})\;(\Varid{f}\;\Varid{b})}. We capture those functors in the
following typeclass: \bigskip

\begin{defn}[Functorization of an optic family]

For a given optic family \ensuremath{\Varid{op}}, we call \emph{functorization} of \ensuremath{\Varid{op}} the
functor family \ensuremath{\Conid{Functorize}\;\Varid{op}} defined as follows:

\begin{hscode}\SaveRestoreHook
\column{B}{@{}>{\hspre}l<{\hspost}@{}}%
\column{5}{@{}>{\hspre}l<{\hspost}@{}}%
\column{E}{@{}>{\hspre}l<{\hspost}@{}}%
\>[B]{}\mathbf{class}\;(\Conid{OpticFamily}\;\Varid{op},\Conid{Functor}\;\Varid{f})\Rightarrow \Conid{Functorize}\;\Varid{op}\;\Varid{f}\;\mathbf{where}{}\<[E]%
\\
\>[B]{}\hsindent{5}{}\<[5]%
\>[5]{}\Varid{enhanceFop}\mathbin{::}\Varid{op}\;\Varid{a}\;\Varid{b}\;(\Varid{f}\;\Varid{a})\;(\Varid{f}\;\Varid{b}){}\<[E]%
\ColumnHook
\end{hscode}\resethooks

with the following laws:

\begin{itemize}
\tightlist
\item
  \ensuremath{\Varid{mapOptic}\;\Varid{enhanceFop}\mathrel{=}\Varid{fmap}}
\item
  \ensuremath{\Varid{enhanceFop}\;@\Varid{f}\ \circ_{op}\ \Varid{injOptic}\;\Varid{f}\;\Varid{g}\mathrel{=}\Varid{injOptic}\;(\Varid{fmap}\;\Varid{f})\;(\Varid{fmap}\;\Varid{g})\ \circ_{op}\ \Varid{enhanceFop}\;@\Varid{f}}
\item
  \ensuremath{\alpha\mathbin{::}\mathbf{forall}\;\Varid{a}\hsforall \hsdot{\circ }{.\;}\Varid{f}\;\Varid{a}\to \Varid{g}\;\Varid{a}\Rightarrow \Varid{injOptic}\;\Varid{id}\;\alpha\ \circ_{op}\ \Varid{enhanceFop}\;@\Varid{f}\mathrel{=}\Varid{injOptic}\;\alpha\;\Varid{id}\ \circ_{op}\ \Varid{enhanceFop}\;@\Varid{g}}
\end{itemize}

\end{defn}

\bigskip

\begin{prop}\label{functorize_is_fmonoid}

\ensuremath{\Conid{Functorize}\;\Varid{op}} is a functor monoid.

\end{prop}

\begin{proof}

By definition, every \ensuremath{\Varid{f}} in \ensuremath{\Conid{Functorize}\;\Varid{op}} is a functor.

\ensuremath{\Conid{Id}} is in \ensuremath{\Conid{Functorize}\;\Varid{op}}:

\begin{hscode}\SaveRestoreHook
\column{B}{@{}>{\hspre}l<{\hspost}@{}}%
\column{5}{@{}>{\hspre}l<{\hspost}@{}}%
\column{E}{@{}>{\hspre}l<{\hspost}@{}}%
\>[B]{}\mathbf{instance}\;\Conid{Functorize}\;\Varid{op}\;\Conid{Id}\;\mathbf{where}{}\<[E]%
\\
\>[B]{}\hsindent{5}{}\<[5]%
\>[5]{}\Varid{enhanceFop}\mathbin{::}\Varid{op}\;\Varid{a}\;\Varid{b}\;(\Conid{Id}\;\Varid{a})\;(\Conid{Id}\;\Varid{b}){}\<[E]%
\\
\>[B]{}\hsindent{5}{}\<[5]%
\>[5]{}\Varid{enhanceFop}\mathrel{=}\Varid{injOp}\;\Varid{unId}\;\Conid{Id}{}\<[E]%
\ColumnHook
\end{hscode}\resethooks

\ensuremath{\Conid{Functorize}\;\Varid{op}} is closed under composition:

\begin{hscode}\SaveRestoreHook
\column{B}{@{}>{\hspre}l<{\hspost}@{}}%
\column{5}{@{}>{\hspre}l<{\hspost}@{}}%
\column{9}{@{}>{\hspre}l<{\hspost}@{}}%
\column{11}{@{}>{\hspre}l<{\hspost}@{}}%
\column{E}{@{}>{\hspre}l<{\hspost}@{}}%
\>[B]{}\mathbf{instance}\;{}\<[11]%
\>[11]{}(\Conid{Functorize}\;\Varid{op}\;\Varid{f},\Conid{Functorize}\;\Varid{op}\;\Varid{g})\Rightarrow {}\<[E]%
\\
\>[11]{}\Conid{Functorize}\;\Varid{op}\;(\Conid{Compose}\;\Varid{f}\;\Varid{g})\;\mathbf{where}{}\<[E]%
\\
\>[B]{}\hsindent{5}{}\<[5]%
\>[5]{}\Varid{enhanceFop}\mathbin{::}\Varid{op}\;\Varid{a}\;\Varid{b}\;(\Conid{Compose}\;\Varid{f}\;\Varid{g}\;\Varid{a})\;(\Conid{Compose}\;\Varid{f}\;\Varid{g}\;\Varid{b}){}\<[E]%
\\
\>[B]{}\hsindent{5}{}\<[5]%
\>[5]{}\Varid{enhanceFop}\mathrel{=}\Varid{injOptic}\;\Varid{unCompose}\;\Conid{Compose}\ \circ_{op}\ {}\<[E]%
\\
\>[5]{}\hsindent{4}{}\<[9]%
\>[9]{}\Varid{enhanceFop}\;@\Varid{f}\ \circ_{op}\ \Varid{enhanceFop}\;@\Varid{g}{}\<[E]%
\ColumnHook
\end{hscode}\resethooks

See \textcite{proof_functorize_is_fmonoid} for the first two
\ensuremath{\Varid{enhanceFop}} laws.

Regarding the wedge condition (the third law), we observe that the \ensuremath{\Conid{Id}}
and \ensuremath{\Conid{Compose}} cases are compatible. Without knowing more about \ensuremath{\Varid{op}} it
is however impossible to prove that it holds generically.
Implementations for specific \ensuremath{\Varid{op}} will be required to check the law. In
practice, the wedge condition usually holds automatically by
parametricity.

\end{proof}

\bigskip

\begin{prop}

If \ensuremath{\Varid{op}\;\mathbin{\in}\;\Conid{Enhanceable}\;\sigma}, then \ensuremath{\sigma\;\subset\;\Conid{Functorize}\;\Varid{op}}

\end{prop}

\begin{proof}

Given \ensuremath{\Varid{f}\;\mathbin{\in}\;\sigma}, we easily have a \ensuremath{\Conid{Functorize}\;\Varid{op}\;\Varid{f}} instance:

\begin{hscode}\SaveRestoreHook
\column{B}{@{}>{\hspre}l<{\hspost}@{}}%
\column{5}{@{}>{\hspre}l<{\hspost}@{}}%
\column{E}{@{}>{\hspre}l<{\hspost}@{}}%
\>[B]{}\mathbf{instance}\;(\Conid{Enhanceable}\;\sigma\;\Varid{op},\sigma\;\Varid{f})\Rightarrow \Conid{Functorize}\;\Varid{op}\;\Varid{f}\;\mathbf{where}{}\<[E]%
\\
\>[B]{}\hsindent{5}{}\<[5]%
\>[5]{}\Varid{enhanceFop}\mathbin{::}\Varid{op}\;\Varid{a}\;\Varid{b}\;(\Varid{f}\;\Varid{a})\;(\Varid{f}\;\Varid{b}){}\<[E]%
\\
\>[B]{}\hsindent{5}{}\<[5]%
\>[5]{}\Varid{enhanceFop}\mathrel{=}\Varid{enhanceOp}{}\<[E]%
\ColumnHook
\end{hscode}\resethooks

This instance is valid thanks to the \ensuremath{\Varid{enhanceOp}} laws, that are a
superset of the \ensuremath{\Varid{enhanceFop}} laws.

\end{proof}

This instance allows us to define the following embedding:

\begin{hscode}\SaveRestoreHook
\column{B}{@{}>{\hspre}l<{\hspost}@{}}%
\column{E}{@{}>{\hspre}l<{\hspost}@{}}%
\>[B]{}\eta\mathbin{::}\Conid{Enhanceable}\;\sigma\;\Varid{op}\Rightarrow \Conid{IsoOptic}\;\sigma\leadsto \Conid{IsoOptic}\;(\Conid{Functorize}\;\Varid{op}){}\<[E]%
\\
\>[B]{}\eta\;(\Conid{IsoOptic}\;\alpha\;\beta)\mathrel{=}\Conid{IsoOptic}\;\alpha\;\beta{}\<[E]%
\ColumnHook
\end{hscode}\resethooks

\bigskip

\ensuremath{\Conid{Functorize}\;\Varid{op}} is therefore the most general functor monoid that could
provide a profunctor encoding for \ensuremath{\Varid{op}}.

In fact, as the next results show, \ensuremath{\Conid{Functorize}\;\Varid{op}} is always the right
choice of functor monoid to look for a profunctor encoding for \ensuremath{\Varid{op}}.
\bigskip \bigskip

\begin{prop}

\ensuremath{\Varid{enhanceFop}} defines a lawful instance of
\ensuremath{\Conid{Enhanceable}\;(\Conid{Functorize}\;\Varid{op})\;\Varid{op}}.

\end{prop}

\begin{proof}

The first two of the \ensuremath{\Conid{Enhanceable}} laws are satisfied by the definition
of \ensuremath{\mathbf{instance}\;\Conid{Functorize}\;\Varid{op}\;\Conid{Id}} and
\ensuremath{\mathbf{instance}\;\Conid{Functorize}\;\Varid{op}\;(\Conid{Compose}\;\Varid{f}\;\Varid{g})}. The last three of the
\ensuremath{\Conid{Enhanceable}} laws are enforced by the \ensuremath{\Conid{Functorize}} laws.

\end{proof}

We call \ensuremath{\Varid{unfunctorize}} the specialization of \ensuremath{\Varid{enhanceToArrow}} to the
\ensuremath{\Conid{Functorize}\;\Varid{op}} functor monoid:

\begin{hscode}\SaveRestoreHook
\column{B}{@{}>{\hspre}l<{\hspost}@{}}%
\column{E}{@{}>{\hspre}l<{\hspost}@{}}%
\>[B]{}\Varid{unfunctorize}\mathbin{::}\Conid{IsoOptic}\;(\Conid{Functorize}\;\Varid{op})\leadsto \Varid{op}{}\<[E]%
\\
\>[B]{}\Varid{unfunctorize}\mathrel{=}\Varid{enhanceToArrow}{}\<[E]%
\\
\>[B]{}\mbox{\onelinecomment  \ensuremath{\Varid{unfunctorize}\;(\Conid{IsoOptic}\;\alpha\;\beta)\mathrel{=}\Varid{injOptic}\;\alpha\;\beta\ \circ_{op}\ \Varid{enhanceFop}}}{}\<[E]%
\ColumnHook
\end{hscode}\resethooks

\bigskip
\bigskip

\begin{thm}[Functorization theorem]\label{functorization_theorem}

If \ensuremath{\Varid{op}} has a profunctor encoding, then \ensuremath{\Varid{op}} has a profunctor encoding
to \ensuremath{\Conid{Functorize}\;\Varid{op}} too.

\end{thm}

\begin{proof}

Assume \ensuremath{\Varid{op}} has a profunctor encoding via the functor monoid \ensuremath{\sigma}.
\ensuremath{\Varid{op}} is isomorphic (as optic families) to \ensuremath{\Conid{ProfOptic}\;\sigma}, thus to
\ensuremath{\Conid{IsoOptic}\;\sigma} using the representation theorem.

Let \ensuremath{\theta\mathbin{::}\Conid{IsoOptic}\;\sigma\leadsto \Varid{op}} be this isomorphism.

\ensuremath{\eta\hsdot{\circ }{.\;}\theta^{-1}\hsdot{\circ }{.\;}\Varid{unfunctorize}\mathbin{::}\Conid{IsoOptic}\;(\Conid{Functorize}\;\Varid{op})\leadsto \Conid{IsoOptic}\;(\Conid{Functorize}\;\Varid{op})}.
By \textcite{endo_iso_is_singleton},
\ensuremath{(\eta\hsdot{\circ }{.\;}\theta^{-1})\hsdot{\circ }{.\;}\Varid{unfunctorize}\mathrel{=}\Varid{id}}.

Similarly,
\ensuremath{\theta^{-1}\hsdot{\circ }{.\;}\Varid{unfunctorize}\hsdot{\circ }{.\;}\eta\mathbin{::}\Conid{IsoOptic}\;\sigma\leadsto \Conid{IsoOptic}\;\sigma},
therefore \ensuremath{\theta^{-1}\hsdot{\circ }{.\;}\Varid{unfunctorize}\hsdot{\circ }{.\;}\eta\mathrel{=}\Varid{id}}

Thus \ensuremath{\Varid{unfunctorize}\hsdot{\circ }{.\;}\eta\mathrel{=}\theta}, and
\ensuremath{\Varid{unfunctorize}\hsdot{\circ }{.\;}(\eta\hsdot{\circ }{.\;}\theta^{-1})\mathrel{=}\Varid{id}}.

So \ensuremath{\Varid{unfunctorize}} is an isomorphism, and
\ensuremath{\Varid{op}\cong \Conid{IsoOptic}\;(\Conid{Functorize}\;\Varid{op})\cong \Conid{ProfOptic}\;(\Conid{Functorize}\;\Varid{op})}.

\end{proof}

\bigskip

\begin{thm}[Derivation theorem]\label{derivation_theorem}

\ensuremath{\Varid{op}} has a profunctor encoding if and only if
\ensuremath{\Varid{unfunctorize}\mathbin{::}\Conid{IsoOptic}\;(\Conid{Functorize}\;\Varid{op})\leadsto \Varid{op}} has an inverse (that
need not be a morphism of optic families).

\end{thm}

\begin{proof}

By \textcite{functorization_theorem}, if \ensuremath{\Varid{op}} has a profunctor encoding
then \ensuremath{\Varid{unfunctorize}} has an inverse.

By \textcite{invertible_opfam_morphism}, if \ensuremath{\Varid{unfunctorize}} has an
inverse it is automatically a morphism of optic families. Hence
\ensuremath{\Varid{op}\cong \Conid{IsoOptic}\;(\Conid{Functorize}\;\Varid{op})\cong \Conid{ProfOptic}\;(\Conid{Functorize}\;\Varid{op})}.

\end{proof}

\hypertarget{tying-it-all-up-a-new-look-at-profunctor-encodings}{%
\chapter{Tying it all up: a new look at profunctor
encodings}\label{tying-it-all-up-a-new-look-at-profunctor-encodings}}

The theorems derived in
\cref{isomorphism-optics,two-theorems-about-profunctor-optics} provide a
framework to study properties of optic families. In this chapter, we
apply this framework to (re)derive properties of some common optic
families, as well as of a new one.

Interestingly, most of the derivations are done by simply ``following
the types''. The main remaining non-mechanical step is the definition of
the inverse to \ensuremath{\Varid{unfunctorize}}, which requires the choice of a general
enough functor.

The overly mechanical proofs will be omitted for brevity.

\hypertarget{rederiving-profunctor-encodings-for-common-optics}{%
\section{Rederiving profunctor encodings for common
optics}\label{rederiving-profunctor-encodings-for-common-optics}}

In this section, we look at some common optic families using the new
framework.

We rederive their usual profunctor encoding in a semi-mechanical way and
get some properties for free, notably that the encoding preserves the
optic family structure. We also get useful equivalent representations
thanks to iso optics.

\hypertarget{adapter}{%
\subsection{Adapter}\label{adapter}}

\hypertarget{definition-2}{%
\subsubsection{Definition}\label{definition-2}}

\begin{hscode}\SaveRestoreHook
\column{B}{@{}>{\hspre}l<{\hspost}@{}}%
\column{5}{@{}>{\hspre}l<{\hspost}@{}}%
\column{E}{@{}>{\hspre}l<{\hspost}@{}}%
\>[B]{}\mathbf{data}\;\Conid{Adapter}\;\Varid{a}\;\Varid{b}\;\Varid{s}\;\Varid{t}\mathrel{=}\Conid{Adapter}\;(\Varid{s}\to \Varid{a})\;(\Varid{b}\to \Varid{t}){}\<[E]%
\\[\blanklineskip]%
\>[B]{}\mathbf{instance}\;\Conid{OpticFamily}\;\Conid{Adapter}\;\mathbf{where}{}\<[E]%
\\
\>[B]{}\hsindent{5}{}\<[5]%
\>[5]{}(\Conid{Adapter}\;\Varid{f}\;\Varid{g})\ \circ_{op}\ (\Conid{Adapter}\;\Varid{f'}\;\Varid{g'})\mathrel{=}\Conid{Adapter}\;(\Varid{f'}\hsdot{\circ }{.\;}\Varid{f})\;(\Varid{g}\hsdot{\circ }{.\;}\Varid{g'}){}\<[E]%
\\
\>[B]{}\hsindent{5}{}\<[5]%
\>[5]{}\Varid{injOptic}\;\Varid{f}\;\Varid{g}\mathrel{=}\Conid{Adapter}\;\Varid{f}\;\Varid{g}{}\<[E]%
\\
\>[B]{}\hsindent{5}{}\<[5]%
\>[5]{}\Varid{mapOptic}\;(\Conid{Adapter}\;\Varid{f}\;\Varid{g})\;\Varid{h}\mathrel{=}\Varid{g}\hsdot{\circ }{.\;}\Varid{h}\hsdot{\circ }{.\;}\Varid{f}{}\<[E]%
\ColumnHook
\end{hscode}\resethooks

\hypertarget{analyzing-the-functorization}{%
\subsubsection{Analyzing the
functorization}\label{analyzing-the-functorization}}

Let's study \ensuremath{\Conid{Functorize}\;\Conid{Adapter}}:

Let \ensuremath{\Varid{f}\;\mathbin{\in}\;\Conid{Functorize}\;\Conid{Adapter}} and \ensuremath{\Conid{Adapter}\;\alpha\;\beta\mathrel{=}\Varid{enhanceFop}\;@\Varid{f}}.

\begin{hscode}\SaveRestoreHook
\column{B}{@{}>{\hspre}l<{\hspost}@{}}%
\column{E}{@{}>{\hspre}l<{\hspost}@{}}%
\>[B]{}\Varid{enhanceOp}\;@\Varid{f}\mathbin{::}\mathbf{forall}\;\Varid{a}\hsforall \;\Varid{b}\hsdot{\circ }{.\;}\Conid{Adapter}\;\Varid{a}\;\Varid{b}\;(\Varid{f}\;\Varid{a})\;(\Varid{f}\;\Varid{b}){}\<[E]%
\\
\>[B]{}\alpha\mathbin{::}\mathbf{forall}\;\Varid{a}\hsforall \hsdot{\circ }{.\;}\Varid{f}\;\Varid{a}\to \Varid{a}{}\<[E]%
\\
\>[B]{}\beta\mathbin{::}\mathbf{forall}\;\Varid{b}\hsforall \hsdot{\circ }{.\;}\Varid{b}\to \Varid{f}\;\Varid{b}{}\<[E]%
\ColumnHook
\end{hscode}\resethooks

Since \ensuremath{\Varid{mapOptic}\;(\Varid{enhanceOp}\;@\Varid{f})\mathrel{=}\Varid{fmap}}, we get:

\begin{hscode}\SaveRestoreHook
\column{B}{@{}>{\hspre}l<{\hspost}@{}}%
\column{E}{@{}>{\hspre}l<{\hspost}@{}}%
\>[B]{}\Varid{fmap}\;\Varid{id}{}\<[E]%
\\
\>[B]{}\mathrel{=}\Varid{mapOptic}\;(\Varid{enhanceOp}\;@\Varid{f})\;\Varid{id}{}\<[E]%
\\
\>[B]{}\mathrel{=}\beta\hsdot{\circ }{.\;}\Varid{id}\hsdot{\circ }{.\;}\alpha{}\<[E]%
\ColumnHook
\end{hscode}\resethooks

So \ensuremath{\beta\hsdot{\circ }{.\;}\alpha\mathrel{=}\Varid{id}}.

On the other hand, \ensuremath{\alpha\hsdot{\circ }{.\;}\beta\mathbin{::}\mathbf{forall}\;\Varid{a}\hsforall \hsdot{\circ }{.\;}\Varid{a}\to \Varid{a}}. By parametricity,
the only inhabitant of this type is \ensuremath{\Varid{id}}, and thus \ensuremath{\alpha\hsdot{\circ }{.\;}\beta\mathrel{=}\Varid{id}}.

Therefore \ensuremath{\alpha} and \ensuremath{\beta} are inverses, and \ensuremath{\Varid{f}\cong \Conid{Id}}.

\hypertarget{deriving-a-profunctor-encoding}{%
\subsubsection{Deriving a profunctor
encoding}\label{deriving-a-profunctor-encoding}}

\begin{hscode}\SaveRestoreHook
\column{B}{@{}>{\hspre}l<{\hspost}@{}}%
\column{E}{@{}>{\hspre}l<{\hspost}@{}}%
\>[B]{}\Varid{adapterToIso}\mathbin{::}\Conid{Adapter}\leadsto \Conid{IsoOptic}\;(\Conid{Functorize}\;\Conid{Adapter}){}\<[E]%
\\
\>[B]{}\Varid{adapterToIso}\;(\Conid{Adapter}\;\Varid{f}\;\Varid{g})\mathrel{=}\Varid{injOptic}\;\Varid{f}\;\Varid{g}{}\<[E]%
\ColumnHook
\end{hscode}\resethooks

\smallskip

\begin{hscode}\SaveRestoreHook
\column{B}{@{}>{\hspre}l<{\hspost}@{}}%
\column{E}{@{}>{\hspre}l<{\hspost}@{}}%
\>[B]{}\Varid{adapterToIso}\;(\Varid{unfunctorize}\;(\Conid{IsoOptic}\;\alpha\;\beta)){}\<[E]%
\\
\>[B]{}\mathrel{=}\Varid{adapterToIso}\;(\Varid{injOptic}\;\alpha\;\beta\ \circ_{op}\ \Varid{enhanceFop}){}\<[E]%
\\
\>[B]{}\mathrel{=}\Varid{adapterToIso}\;(\Conid{Adapter}\;\alpha\;\beta\ \circ_{op}\ \Conid{Adapter}\;\Varid{f}\;\Varid{g}){}\<[E]%
\\
\>[B]{}\mathrel{=}\Varid{adapterToIso}\;(\Conid{Adapter}\;(\Varid{f}\hsdot{\circ }{.\;}\alpha)\;(\beta\hsdot{\circ }{.\;}\Varid{g})){}\<[E]%
\\
\>[B]{}\mathrel{=}\Varid{injOptic}\;(\Varid{f}\hsdot{\circ }{.\;}\alpha)\;(\beta\hsdot{\circ }{.\;}\Varid{g}){}\<[E]%
\\
\>[B]{}\mathrel{=}\Conid{IsoOptic}\;(\Conid{Id}\hsdot{\circ }{.\;}\Varid{f}\hsdot{\circ }{.\;}\alpha)\;(\beta\hsdot{\circ }{.\;}\Varid{g}\hsdot{\circ }{.\;}\Varid{unId}){}\<[E]%
\\
\>[B]{}\mathrel{=}\Conid{IsoOptic}\;(\Varid{g}\hsdot{\circ }{.\;}\Varid{unId}\hsdot{\circ }{.\;}\Conid{Id}\hsdot{\circ }{.\;}\Varid{f}\hsdot{\circ }{.\;}\alpha)\;\beta{}\<[E]%
\\
\>[B]{}\mathrel{=}\Conid{IsoOptic}\;\alpha\;\beta{}\<[E]%
\\[\blanklineskip]%
\>[B]{}\Varid{unfunctorize}\;(\Varid{adapterToIso}\;(\Conid{Adapter}\;\Varid{f}\;\Varid{g})){}\<[E]%
\\
\>[B]{}\mathrel{=}\Varid{unfunctorize}\;(\Varid{injOptic}\;\Varid{f}\;\Varid{g}){}\<[E]%
\\
\>[B]{}\mathrel{=}\Varid{injOptic}\;\Varid{f}\;\Varid{g}{}\<[E]%
\\
\>[B]{}\mathrel{=}\Conid{Adapter}\;\Varid{f}\;\Varid{g}{}\<[E]%
\ColumnHook
\end{hscode}\resethooks

By \textcite{derivation_theorem}, \ensuremath{\Conid{ProfOptic}\;(\Conid{Functorize}\;\Conid{Adapter})} is a
profunctor encoding for \ensuremath{\Conid{Adapter}}.

Since every functor in \ensuremath{\Conid{Functorize}\;\Conid{Adapter}} is isomorphic to \ensuremath{\Conid{Id}}, using
the wedge law we get that every instance of
\ensuremath{\Conid{Enhancing}\;(\Conid{Functorize}\;\Conid{Adapter})} is uniquely determined by
\ensuremath{\Varid{enhance}\;@\Conid{Id}}. Since \ensuremath{\Varid{enhance}\;@\Conid{Id}\mathrel{=}\Varid{dimap}\;\Varid{unId}\;\Conid{Id}},
\ensuremath{\Conid{Enhancing}\;(\Conid{Functorize}\;\Conid{Adapter})} adds no constraints to \ensuremath{\Conid{Profunctor}},
and therefore
\ensuremath{\Conid{Adapter}\;\Varid{a}\;\Varid{b}\;\Varid{s}\;\Varid{t}\cong \mathbf{forall}\;\Varid{p}\hsforall \hsdot{\circ }{.\;}\Conid{Profunctor}\;\Varid{p}\Rightarrow \Varid{p}\;\Varid{a}\;\Varid{b}\to \Varid{p}\;\Varid{s}\;\Varid{t}}.

\hypertarget{lens}{%
\subsection{Lens}\label{lens}}

\hypertarget{definition-3}{%
\subsubsection{Definition}\label{definition-3}}

\begin{hscode}\SaveRestoreHook
\column{B}{@{}>{\hspre}l<{\hspost}@{}}%
\column{E}{@{}>{\hspre}l<{\hspost}@{}}%
\>[B]{}\mathbf{data}\;\Conid{Lens}\;\Varid{a}\;\Varid{b}\;\Varid{s}\;\Varid{t}\mathrel{=}\Conid{Lens}\;\{\mskip1.5mu \Varid{get}\mathbin{::}\Varid{s}\to \Varid{a},\Varid{put}\mathbin{::}\Varid{b}\to \Varid{s}\to \Varid{t}\mskip1.5mu\}{}\<[E]%
\ColumnHook
\end{hscode}\resethooks

\smallskip

\begin{hscode}\SaveRestoreHook
\column{B}{@{}>{\hspre}l<{\hspost}@{}}%
\column{5}{@{}>{\hspre}l<{\hspost}@{}}%
\column{9}{@{}>{\hspre}l<{\hspost}@{}}%
\column{E}{@{}>{\hspre}l<{\hspost}@{}}%
\>[B]{}\mathbf{instance}\;\Conid{OpticFamily}\;\Conid{Lens}\;\mathbf{where}{}\<[E]%
\\
\>[B]{}\hsindent{5}{}\<[5]%
\>[5]{}\Varid{injOptic}\;\Varid{f}\;\Varid{g}\mathrel{=}\Conid{Lens}\;\Varid{f}\;(\Varid{const}\hsdot{\circ }{.\;}\Varid{g}){}\<[E]%
\\
\>[B]{}\hsindent{5}{}\<[5]%
\>[5]{}(\circ_{op})\;(\Conid{Lens}\;\Varid{get1}\;\Varid{put1})\;(\Conid{Lens}\;\Varid{get}_{\mathrm{2}}\;\Varid{put}_{\mathrm{2}})\mathrel{=}{}\<[E]%
\\
\>[5]{}\hsindent{4}{}\<[9]%
\>[9]{}\Conid{Lens}\;(\Varid{get}_{\mathrm{2}}\hsdot{\circ }{.\;}\Varid{get}_{\mathrm{1}})\;(\lambda \Varid{y}\;\Varid{s}\to \Varid{put}_{\mathrm{1}}\;(\Varid{put}_{\mathrm{2}}\;\Varid{y}\;(\Varid{get}_{\mathrm{1}}\;\Varid{s}))\;\Varid{s}){}\<[E]%
\\
\>[B]{}\hsindent{5}{}\<[5]%
\>[5]{}\Varid{mapOptic}\;(\Conid{Lens}\;\Varid{get}\;\Varid{put})\;\Varid{f}\;\Varid{s}\mathrel{=}\Varid{put}\;(\Varid{f}\;(\Varid{get}\;\Varid{s}))\;\Varid{s}{}\<[E]%
\ColumnHook
\end{hscode}\resethooks

\hypertarget{analyzing-the-functorization-1}{%
\subsubsection{Analyzing the
functorization}\label{analyzing-the-functorization-1}}

Let's analyze \ensuremath{\Conid{Functorize}\;\Conid{Lens}}:

Let \ensuremath{\Varid{f}\;\mathbin{\in}\;\Conid{Functorize}\;\Conid{Lens}}. Let \ensuremath{\Conid{Lens}\;\Varid{get}\;\Varid{put}\mathrel{=}\Varid{enhanceFop}\;@\Varid{f}}.

\ensuremath{\Varid{enhanceFop}\;@\Varid{f}} is a lawful lens:

\begin{hscode}\SaveRestoreHook
\column{B}{@{}>{\hspre}l<{\hspost}@{}}%
\column{18}{@{}>{\hspre}l<{\hspost}@{}}%
\column{49}{@{}>{\hspre}l<{\hspost}@{}}%
\column{60}{@{}>{\hspre}l<{\hspost}@{}}%
\column{E}{@{}>{\hspre}l<{\hspost}@{}}%
\>[B]{}\Varid{put}\;(\Varid{get}\;\Varid{s})\;\Varid{s}\mathrel{=}{}\<[18]%
\>[18]{}\Varid{mapOptic}\;(\Varid{enhanceFop}\;@\Varid{f})\;\Varid{id}\mathrel{=}{}\<[49]%
\>[49]{}\Varid{fmap}\;\Varid{id}\mathrel{=}{}\<[60]%
\>[60]{}\Varid{id}{}\<[E]%
\ColumnHook
\end{hscode}\resethooks

\smallskip

\begin{hscode}\SaveRestoreHook
\column{B}{@{}>{\hspre}l<{\hspost}@{}}%
\column{4}{@{}>{\hspre}l<{\hspost}@{}}%
\column{E}{@{}>{\hspre}l<{\hspost}@{}}%
\>[B]{}\lambda \Varid{b}\to \Varid{get}\;(\Varid{put}\;\Varid{b}\;\Varid{fa})\mathbin{::}\mathbf{forall}\;\Varid{b}\hsforall \hsdot{\circ }{.\;}\Varid{b}\to \Varid{b}{}\<[E]%
\\
\>[B]{}\mathrel{=}\mbox{\commentbegin   parametricity   \commentend}{}\<[E]%
\\
\>[B]{}\hsindent{4}{}\<[4]%
\>[4]{}\Varid{id}{}\<[E]%
\ColumnHook
\end{hscode}\resethooks

\smallskip

\begin{hscode}\SaveRestoreHook
\column{B}{@{}>{\hspre}l<{\hspost}@{}}%
\column{10}{@{}>{\hspre}l<{\hspost}@{}}%
\column{48}{@{}>{\hspre}l<{\hspost}@{}}%
\column{E}{@{}>{\hspre}l<{\hspost}@{}}%
\>[B]{}\Varid{put}\;\Varid{a}\mathrel{=}{}\<[10]%
\>[10]{}\Varid{mapOptic}\;(\Varid{enhanceFop}\;@\Varid{f})\;(\Varid{const}\;\Varid{a})\mathrel{=}{}\<[48]%
\>[48]{}\Varid{fmap}\;(\Varid{const}\;\Varid{a}){}\<[E]%
\ColumnHook
\end{hscode}\resethooks

Thus

\begin{hscode}\SaveRestoreHook
\column{B}{@{}>{\hspre}l<{\hspost}@{}}%
\column{18}{@{}>{\hspre}l<{\hspost}@{}}%
\column{46}{@{}>{\hspre}l<{\hspost}@{}}%
\column{64}{@{}>{\hspre}l<{\hspost}@{}}%
\column{E}{@{}>{\hspre}l<{\hspost}@{}}%
\>[B]{}\Varid{put}\;\Varid{a}\hsdot{\circ }{.\;}\Varid{put}\;\Varid{b}\mathrel{=}{}\<[18]%
\>[18]{}\Varid{fmap}\;(\Varid{const}\;\Varid{a}\hsdot{\circ }{.\;}\Varid{const}\;\Varid{b})\mathrel{=}{}\<[46]%
\>[46]{}\Varid{fmap}\;(\Varid{const}\;\Varid{a})\mathrel{=}{}\<[64]%
\>[64]{}\Varid{put}\;\Varid{a}{}\<[E]%
\ColumnHook
\end{hscode}\resethooks

\bigskip

Knowing that lawful lenses have an equivalent \emph{residual}
expression, we can try and write such an equivalence:

\begin{hscode}\SaveRestoreHook
\column{B}{@{}>{\hspre}l<{\hspost}@{}}%
\column{E}{@{}>{\hspre}l<{\hspost}@{}}%
\>[B]{}\Varid{toProduct}\mathbin{::}\Conid{Functorize}\;\Varid{op}\;\Varid{f}\Rightarrow \Varid{f}\;\Varid{a}\to (\Varid{f}\;(),\Varid{a}){}\<[E]%
\\
\>[B]{}\Varid{toProduct}\;\Varid{fa}\mathrel{=}(\Varid{put}\;()\;\Varid{fa},\Varid{get}\;\Varid{fa}){}\<[E]%
\\
\>[B]{}\Varid{fromProduct}\mathbin{::}\Conid{Functorize}\;\Varid{op}\;\Varid{f}\Rightarrow (\Varid{f}\;(),\Varid{a})\to \Varid{f}\;\Varid{a}{}\<[E]%
\\
\>[B]{}\Varid{fromProduct}\;(\Varid{f1},\Varid{a})\mathrel{=}\Varid{put}\;\Varid{a}\;\Varid{f1}{}\<[E]%
\ColumnHook
\end{hscode}\resethooks

Since \ensuremath{\Varid{enhanceFop}\;@\Varid{f}} is lawful, those functions are mutual inverses:

\begin{hscode}\SaveRestoreHook
\column{B}{@{}>{\hspre}l<{\hspost}@{}}%
\column{4}{@{}>{\hspre}l<{\hspost}@{}}%
\column{E}{@{}>{\hspre}l<{\hspost}@{}}%
\>[B]{}\Varid{toProduct}\;(\Varid{fromProduct}\;(\Varid{f1},\Varid{a})){}\<[E]%
\\
\>[B]{}\mathrel{=}{}\<[4]%
\>[4]{}\Varid{toProduct}\;(\Varid{put}\;\Varid{a}\;\Varid{f1}){}\<[E]%
\\
\>[B]{}\mathrel{=}{}\<[4]%
\>[4]{}(\Varid{put}\;()\;(\Varid{put}\;\Varid{a}\;\Varid{f1}),\Varid{get}\;(\Varid{put}\;\Varid{a}\;\Varid{f1})){}\<[E]%
\\
\>[B]{}\mathrel{=}{}\<[4]%
\>[4]{}(\Varid{put}\;()\;\Varid{f1},\Varid{a}){}\<[E]%
\\
\>[B]{}\mathrel{=}{}\<[4]%
\>[4]{}(\Varid{put}\;(\Varid{get}\;\Varid{f1})\;\Varid{f1},\Varid{a}){}\<[E]%
\\
\>[B]{}\mathrel{=}{}\<[4]%
\>[4]{}(\Varid{f1},\Varid{a}){}\<[E]%
\\[\blanklineskip]%
\>[B]{}\Varid{fromProduct}\;(\Varid{toProduct}\;\Varid{fa}){}\<[E]%
\\
\>[B]{}\mathrel{=}{}\<[4]%
\>[4]{}\Varid{fromProduct}\;(\Varid{put}\;()\;\Varid{fa},\Varid{get}\;\Varid{fa}){}\<[E]%
\\
\>[B]{}\mathrel{=}{}\<[4]%
\>[4]{}\Varid{put}\;(\Varid{get}\;\Varid{fa})\;(\Varid{put}\;()\;\Varid{fa}){}\<[E]%
\\
\>[B]{}\mathrel{=}{}\<[4]%
\>[4]{}\Varid{put}\;(\Varid{get}\;\Varid{fa})\;\Varid{fa}{}\<[E]%
\\
\>[B]{}\mathrel{=}{}\<[4]%
\>[4]{}\Varid{fa}{}\<[E]%
\ColumnHook
\end{hscode}\resethooks

We capture this property in the \ensuremath{\Conid{IsProduct}} typeclass:

\begin{hscode}\SaveRestoreHook
\column{B}{@{}>{\hspre}l<{\hspost}@{}}%
\column{5}{@{}>{\hspre}l<{\hspost}@{}}%
\column{E}{@{}>{\hspre}l<{\hspost}@{}}%
\>[B]{}\mathbf{class}\;\Conid{IsProduct}\;\Varid{f}\;\mathbf{where}{}\<[E]%
\\
\>[B]{}\hsindent{5}{}\<[5]%
\>[5]{}\Varid{toProduct}\mathbin{::}\Varid{f}\;\Varid{a}\to (\Varid{f}\;(),\Varid{a}){}\<[E]%
\\
\>[B]{}\hsindent{5}{}\<[5]%
\>[5]{}\Varid{fromProduct}\mathbin{::}(\Varid{f}\;(),\Varid{a})\to \Varid{f}\;\Varid{a}{}\<[E]%
\\
\>[B]{}\hsindent{5}{}\<[5]%
\>[5]{}\mbox{\onelinecomment  toProduct and fromProduct should be inverses}{}\<[E]%
\ColumnHook
\end{hscode}\resethooks

We therefore have that \ensuremath{\Conid{Functorize}\;\Conid{Lens}\cong \Conid{IsProduct}}.

\hypertarget{deriving-a-profunctor-encoding-1}{%
\subsubsection{Deriving a profunctor
encoding}\label{deriving-a-profunctor-encoding-1}}

We can now define the following instance:

\begin{hscode}\SaveRestoreHook
\column{B}{@{}>{\hspre}l<{\hspost}@{}}%
\column{5}{@{}>{\hspre}l<{\hspost}@{}}%
\column{E}{@{}>{\hspre}l<{\hspost}@{}}%
\>[B]{}\mathbf{instance}\;\Conid{Functorize}\;\Conid{Lens}\;((,)\;\Varid{c})\;\mathbf{where}{}\<[E]%
\\
\>[B]{}\hsindent{5}{}\<[5]%
\>[5]{}\Varid{enhanceFop}\mathbin{::}\Conid{Lens}\;\Varid{a}\;\Varid{b}\;(\Varid{c},\Varid{a})\;(\Varid{c},\Varid{b}){}\<[E]%
\\
\>[B]{}\hsindent{5}{}\<[5]%
\>[5]{}\Varid{enhanceFop}\mathrel{=}\Conid{Lens}\;\Varid{snd}\;(\Varid{fmap}\hsdot{\circ }{.\;}\Varid{const}){}\<[E]%
\ColumnHook
\end{hscode}\resethooks

and prove that it is lawful (using \textcite{map_enhance_id}):

\begin{hscode}\SaveRestoreHook
\column{B}{@{}>{\hspre}l<{\hspost}@{}}%
\column{4}{@{}>{\hspre}l<{\hspost}@{}}%
\column{E}{@{}>{\hspre}l<{\hspost}@{}}%
\>[B]{}\Varid{mapOptic}\;\Varid{enhanceFop}\;\Varid{id}\;(\Varid{c},\Varid{a}){}\<[E]%
\\
\>[B]{}\mathrel{=}{}\<[4]%
\>[4]{}\Varid{fmap}\;(\Varid{const}\;(\Varid{id}\;(\Varid{snd}\;(\Varid{c},\Varid{a}))))\;(\Varid{c},\Varid{a}){}\<[E]%
\\
\>[B]{}\mathrel{=}{}\<[4]%
\>[4]{}\Varid{fmap}\;(\Varid{const}\;\Varid{a})\;(\Varid{c},\Varid{a}){}\<[E]%
\\
\>[B]{}\mathrel{=}{}\<[4]%
\>[4]{}(\Varid{c},\Varid{a}){}\<[E]%
\\[\blanklineskip]%
\>[B]{}\Varid{enhanceFop}\ \circ_{op}\ \Varid{injOptic}\;\Varid{f}\;\Varid{g}{}\<[E]%
\\
\>[B]{}\mathrel{=}{}\<[4]%
\>[4]{}\Conid{Lens}\;\Varid{snd}\;(\Varid{fmap}\hsdot{\circ }{.\;}\Varid{const})\ \circ_{op}\ \Conid{Lens}\;\Varid{f}\;(\Varid{const}\hsdot{\circ }{.\;}\Varid{g}){}\<[E]%
\\
\>[B]{}\mathrel{=}{}\<[4]%
\>[4]{}\Conid{Lens}\;(\Varid{f}\hsdot{\circ }{.\;}\Varid{snd})\;(\lambda \Varid{b}\;\Varid{s}\to \Varid{fmap}\;(\Varid{const}\;(\Varid{const}\;(\Varid{g}\;\Varid{b})\;(\Varid{snd}\;\Varid{s})))\;\Varid{s}){}\<[E]%
\\
\>[B]{}\mathrel{=}{}\<[4]%
\>[4]{}\Conid{Lens}\;(\Varid{snd}\hsdot{\circ }{.\;}\Varid{fmap}\;\Varid{f})\;(\lambda \Varid{b}\;\Varid{s}\to \Varid{fmap}\;(\Varid{const}\;(\Varid{g}\;\Varid{b}))\;\Varid{s}){}\<[E]%
\\
\>[B]{}\mathrel{=}{}\<[4]%
\>[4]{}\Conid{Lens}\;(\Varid{snd}\hsdot{\circ }{.\;}\Varid{fmap}\;\Varid{f})\;(\lambda \Varid{b}\;(\Varid{c},\Varid{a})\to (\Varid{c},\Varid{g}\;\Varid{b})){}\<[E]%
\\
\>[B]{}\mathrel{=}{}\<[4]%
\>[4]{}\Conid{Lens}\;(\Varid{snd}\hsdot{\circ }{.\;}\Varid{fmap}\;\Varid{f})\;(\lambda \Varid{b}\to \Varid{fmap}\;\Varid{g}\hsdot{\circ }{.\;}\Varid{fmap}\;(\Varid{const}\;\Varid{b})){}\<[E]%
\\
\>[B]{}\mathrel{=}{}\<[4]%
\>[4]{}\Varid{injOptic}\;(\Varid{fmap}\;\Varid{f})\;(\Varid{fmap}\;\Varid{g})\ \circ_{op}\ \Varid{enhanceFop}{}\<[E]%
\ColumnHook
\end{hscode}\resethooks

\bigskip

Parametricity makes the wedge condition always hold:

Let \ensuremath{\Varid{f}\;\mathbin{\in}\;\Conid{Functorize}\;\Conid{Lens}}. Let \ensuremath{\Conid{Lens}\;\Varid{get}\;\Varid{put}\mathrel{=}\Varid{enhanceFop}\;@\Varid{f}}.

We saw that \ensuremath{\Varid{enhanceFop}\;@\Varid{f}} is lawful, and that
\ensuremath{\Varid{put}\;\Varid{b}\mathrel{=}\Varid{fmap}\;(\Varid{const}\;\Varid{b})}.

Let \ensuremath{\Varid{h}\mathbin{::}\mathbf{forall}\;\Varid{a}\hsforall \hsdot{\circ }{.\;}\Varid{f}\;\Varid{a}\to \Varid{a}}. By parametricity, given \ensuremath{\Varid{f}\mathbin{::}\Varid{a}\to \Varid{b}}, we
have \ensuremath{\Varid{f}\hsdot{\circ }{.\;}\Varid{h}\mathrel{=}\Varid{h}\hsdot{\circ }{.\;}\Varid{fmap}\;\Varid{f}}. Thus
\ensuremath{\Varid{const}\;\Varid{x}\mathrel{=}\Varid{const}\;\Varid{x}\hsdot{\circ }{.\;}\Varid{h}\mathrel{=}\Varid{h}\hsdot{\circ }{.\;}\Varid{fmap}\;(\Varid{const}\;\Varid{x})\mathrel{=}\Varid{h}\hsdot{\circ }{.\;}\Varid{put}\;\Varid{x}}. So
\ensuremath{\Varid{h}\;\Varid{s}\mathrel{=}\Varid{h}\;(\Varid{put}\;(\Varid{get}\;\Varid{s})\;\Varid{s})\mathrel{=}\Varid{const}\;(\Varid{get}\;\Varid{s})\;\Varid{s}\mathrel{=}\Varid{get}\;\Varid{s}}, and \ensuremath{\Varid{h}\mathrel{=}\Varid{get}}.

Now let \ensuremath{\Varid{f},\Varid{g}\;\mathbin{\in}\;\Conid{Functorize}\;\Conid{Lens}} and \ensuremath{\alpha\mathbin{::}\mathbf{forall}\;\Varid{a}\hsforall \hsdot{\circ }{.\;}\Varid{f}\;\Varid{a}\to \Varid{g}\;\Varid{a}}.

Let \ensuremath{\Conid{Lens}\;\Varid{get}\;\Varid{put}\mathrel{=}\Varid{enhanceFop}\;@\Varid{f}}, \ensuremath{\Conid{Lens}\;\Varid{get'}\;\Varid{put'}\mathrel{=}\Varid{enhanceFop}\;@\Varid{g}}.

By parametricity,
\ensuremath{\alpha\hsdot{\circ }{.\;}\Varid{put}\;\Varid{b}\mathrel{=}\alpha\hsdot{\circ }{.\;}\Varid{fmap}\;(\Varid{const}\;\Varid{b})\mathrel{=}\Varid{fmap}\;(\Varid{const}\;\Varid{b})\hsdot{\circ }{.\;}\alpha\mathrel{=}\Varid{put'}\;\Varid{b}}.
Using the previous result, we also have
\ensuremath{\Varid{get'}\hsdot{\circ }{.\;}\alpha\mathbin{::}\mathbf{forall}\;\Varid{a}\hsforall \hsdot{\circ }{.\;}\Varid{f}\;\Varid{a}\to \Varid{a}}, hence \ensuremath{\Varid{get'}\hsdot{\circ }{.\;}\alpha\mathrel{=}\Varid{get}}

Therefore:

\begin{hscode}\SaveRestoreHook
\column{B}{@{}>{\hspre}l<{\hspost}@{}}%
\column{4}{@{}>{\hspre}l<{\hspost}@{}}%
\column{E}{@{}>{\hspre}l<{\hspost}@{}}%
\>[B]{}\Varid{injOptic}\;\Varid{id}\;\alpha\ \circ_{op}\ \Varid{enhanceFop}\;@\Varid{f}{}\<[E]%
\\
\>[B]{}\mathrel{=}{}\<[4]%
\>[4]{}\Varid{injOptic}\;\Varid{id}\;\alpha\ \circ_{op}\ \Conid{Lens}\;\Varid{get}\;\Varid{put}{}\<[E]%
\\
\>[B]{}\mathrel{=}{}\<[4]%
\>[4]{}\Conid{Lens}\;\Varid{get}\;(\lambda \Varid{b}\to \alpha\hsdot{\circ }{.\;}\Varid{put}\;\Varid{b}){}\<[E]%
\\
\>[B]{}\mathrel{=}{}\<[4]%
\>[4]{}\Conid{Lens}\;\Varid{get}\;(\lambda \Varid{b}\to \Varid{put'}\;\Varid{b}\hsdot{\circ }{.\;}\alpha){}\<[E]%
\\
\>[B]{}\mathrel{=}{}\<[4]%
\>[4]{}\Conid{Lens}\;(\Varid{get'}\hsdot{\circ }{.\;}\alpha)\;(\lambda \Varid{b}\to \Varid{put'}\;\Varid{b}\hsdot{\circ }{.\;}\alpha){}\<[E]%
\\
\>[B]{}\mathrel{=}{}\<[4]%
\>[4]{}\Varid{injOptic}\;\alpha\;\Varid{id}\ \circ_{op}\ \Conid{Lens}\;\Varid{get'}\;\Varid{put'}{}\<[E]%
\\
\>[B]{}\mathrel{=}{}\<[4]%
\>[4]{}\Varid{injOptic}\;\alpha\;\Varid{id}\ \circ_{op}\ \Varid{enhanceFop}\;@\Varid{g}{}\<[E]%
\ColumnHook
\end{hscode}\resethooks

\bigskip

We can now define the inverse to \ensuremath{\Varid{unfunctorize}}:

\begin{hscode}\SaveRestoreHook
\column{B}{@{}>{\hspre}l<{\hspost}@{}}%
\column{E}{@{}>{\hspre}l<{\hspost}@{}}%
\>[B]{}\Varid{lensToIso}\;(\Conid{Lens}\;\Varid{get}\;\Varid{put})\mathrel{=}\Conid{IsoOptic}\;(\lambda \Varid{s}\to (\Varid{s},\Varid{get}\;\Varid{s}))\;(\lambda (\Varid{s},\Varid{b})\to \Varid{put}\;\Varid{b}\;\Varid{s}){}\<[E]%
\ColumnHook
\end{hscode}\resethooks

We omit the proof that they are inverses, but note that they only are so
if we require lawful lenses.

We get both a profunctor encoding for \ensuremath{\Conid{Lens}}, and an alternative
expression via the isomorphism encoding:

\begin{hscode}\SaveRestoreHook
\column{B}{@{}>{\hspre}l<{\hspost}@{}}%
\column{15}{@{}>{\hspre}l<{\hspost}@{}}%
\column{E}{@{}>{\hspre}l<{\hspost}@{}}%
\>[B]{}\Conid{Lens}\;\Varid{a}\;\Varid{b}\;\Varid{s}\;\Varid{t}{}\<[15]%
\>[15]{}\cong \mathbf{exists}\;\Varid{f}\hsexists \hsdot{\circ }{.\;}\Conid{Functorize}\;\Conid{Lens}\;\Varid{f}\Rightarrow (\Varid{s}\to \Varid{f}\;\Varid{a},\Varid{f}\;\Varid{b}\to \Varid{t}){}\<[E]%
\\
\>[15]{}\cong \mathbf{exists}\;\Varid{f}\hsexists \hsdot{\circ }{.\;}\Conid{IsProduct}\;\Varid{f}\Rightarrow (\Varid{s}\to \Varid{f}\;\Varid{a},\Varid{f}\;\Varid{b}\to \Varid{t}){}\<[E]%
\ColumnHook
\end{hscode}\resethooks

Since \ensuremath{\Conid{IsProduct}\;\Varid{f}\mathbin{\mathbf{\Leftrightarrow}}\mathbf{exists}\;\Varid{c}\hsexists \hsdot{\circ }{.\;}\Varid{f}\cong (\Varid{c},\mathbin{-})}, we finally get:

\begin{hscode}\SaveRestoreHook
\column{B}{@{}>{\hspre}l<{\hspost}@{}}%
\column{15}{@{}>{\hspre}l<{\hspost}@{}}%
\column{E}{@{}>{\hspre}l<{\hspost}@{}}%
\>[B]{}\Conid{Lens}\;\Varid{a}\;\Varid{b}\;\Varid{s}\;\Varid{t}{}\<[15]%
\>[15]{}\cong \mathbf{exists}\;\Varid{c}\hsexists \Rightarrow (\Varid{s}\to (\Varid{c},\Varid{a}),(\Varid{c},\Varid{b})\to \Varid{t}){}\<[E]%
\ColumnHook
\end{hscode}\resethooks

Similarly, we get the usual lens profunctor encoding through the same
isomorphism.

\hypertarget{setter}{%
\subsection{Setter}\label{setter}}

\hypertarget{definition-4}{%
\subsubsection{Definition}\label{definition-4}}

\begin{hscode}\SaveRestoreHook
\column{B}{@{}>{\hspre}l<{\hspost}@{}}%
\column{5}{@{}>{\hspre}l<{\hspost}@{}}%
\column{E}{@{}>{\hspre}l<{\hspost}@{}}%
\>[B]{}\mathbf{data}\;\Conid{Setter}\;\Varid{a}\;\Varid{b}\;\Varid{s}\;\Varid{t}\mathrel{=}\Conid{Setter}\;\{\mskip1.5mu \Varid{unSetter}\mathbin{::}(\Varid{a}\to \Varid{b})\to (\Varid{s}\to \Varid{t})\mskip1.5mu\}{}\<[E]%
\\[\blanklineskip]%
\>[B]{}\mathbf{instance}\;\Conid{OpticFamily}\;\Conid{Setter}\;\mathbf{where}{}\<[E]%
\\
\>[B]{}\hsindent{5}{}\<[5]%
\>[5]{}\Varid{injOptic}\;\Varid{f}\;\Varid{g}\mathrel{=}\Conid{Setter}\;(\Varid{dimap}\;\Varid{f}\;\Varid{g}){}\<[E]%
\\
\>[B]{}\hsindent{5}{}\<[5]%
\>[5]{}(\circ_{op})\;(\Conid{Setter}\;\Varid{f1})\;(\Conid{Setter}\;\Varid{f2})\mathrel{=}\Conid{Setter}\;(\Varid{f1}\hsdot{\circ }{.\;}\Varid{f2}){}\<[E]%
\\
\>[B]{}\hsindent{5}{}\<[5]%
\>[5]{}\Varid{mapOptic}\;(\Conid{Setter}\;\Varid{f})\mathrel{=}\Varid{f}{}\<[E]%
\ColumnHook
\end{hscode}\resethooks

\hypertarget{analyzing-the-functorization-2}{%
\subsubsection{Analyzing the
functorization}\label{analyzing-the-functorization-2}}

Since \ensuremath{\Varid{mapOptic}\;(\Conid{Setter}\;\Varid{f})\mathrel{=}\Varid{f}} and \ensuremath{\Varid{mapOptic}\;\Varid{enhanceFop}\mathrel{=}\Varid{fmap}},
necessarily \ensuremath{\Varid{enhanceFop}\mathrel{=}\Conid{Setter}\;\Varid{fmap}}.

\begin{hscode}\SaveRestoreHook
\column{B}{@{}>{\hspre}l<{\hspost}@{}}%
\column{5}{@{}>{\hspre}l<{\hspost}@{}}%
\column{E}{@{}>{\hspre}l<{\hspost}@{}}%
\>[B]{}\mathbf{instance}\;\Conid{Functor}\;\Varid{f}\Rightarrow \Conid{Functorize}\;\Conid{Setter}\;\Varid{f}\;\mathbf{where}{}\<[E]%
\\
\>[B]{}\hsindent{5}{}\<[5]%
\>[5]{}\Varid{enhanceFop}\mathrel{=}\Conid{Setter}\;\Varid{fmap}{}\<[E]%
\ColumnHook
\end{hscode}\resethooks

The \ensuremath{\Varid{enhanceFop}} laws are straightforward by parametricity.

Thus \ensuremath{\Conid{Functorize}\;\Conid{Setter}\cong \Conid{Functor}}.

\hypertarget{deriving-a-profunctor-encoding-2}{%
\subsubsection{Deriving a profunctor
encoding}\label{deriving-a-profunctor-encoding-2}}

The choice of a correct functor is not trivial. Since \ensuremath{\Conid{Setter}} has a
known profunctor encoding, the easiest solution is to look at the
implementation of one of the existing profunctor optic libraries. We
therefore define the following functor:

\begin{hscode}\SaveRestoreHook
\column{B}{@{}>{\hspre}l<{\hspost}@{}}%
\column{5}{@{}>{\hspre}l<{\hspost}@{}}%
\column{E}{@{}>{\hspre}l<{\hspost}@{}}%
\>[B]{}\mathbf{newtype}\;\Conid{CPS}\;\Varid{t}\;\Varid{b}\;\Varid{a}\mathrel{=}\Conid{CPS}\;\{\mskip1.5mu \Varid{unCPS}\mathbin{::}(\Varid{a}\to \Varid{b})\to \Varid{t}\mskip1.5mu\}{}\<[E]%
\\[\blanklineskip]%
\>[B]{}\mathbf{instance}\;\Conid{Functor}\;(\Conid{CPS}\;\Varid{t}\;\Varid{b})\;\mathbf{where}{}\<[E]%
\\
\>[B]{}\hsindent{5}{}\<[5]%
\>[5]{}\Varid{fmap}\;\Varid{f}\mathrel{=}\Conid{CPS}\hsdot{\circ }{.\;}(\hsdot{\circ }{.\;}(\hsdot{\circ }{.\;}\Varid{f}))\hsdot{\circ }{.\;}\Varid{unCPS}{}\<[E]%
\ColumnHook
\end{hscode}\resethooks

The inverse to \ensuremath{\Varid{unfunctorize}} is then:

\begin{hscode}\SaveRestoreHook
\column{B}{@{}>{\hspre}l<{\hspost}@{}}%
\column{E}{@{}>{\hspre}l<{\hspost}@{}}%
\>[B]{}\Varid{setterToIso}\;(\Conid{Setter}\;\Varid{f})\mathrel{=}\Conid{IsoOptic}\;(\Conid{CPS}\hsdot{\circ }{.\;}\Varid{flip}\;\Varid{f})\;((\mathbin{\$}\Varid{id})\hsdot{\circ }{.\;}\Varid{unCPS}){}\<[E]%
\ColumnHook
\end{hscode}\resethooks

We get that \ensuremath{\Conid{Adapter}\cong \Conid{ProfOptic}\;\Conid{Functor}}.

The \ensuremath{\Conid{IsoOptic}} representation also gives us
\ensuremath{\Conid{Adapter}\;\Varid{a}\;\Varid{b}\;\Varid{s}\;\Varid{t}\cong \mathbf{exists}\;\Varid{f}\hsexists \hsdot{\circ }{.\;}\Conid{Functor}\;\Varid{f}\Rightarrow (\Varid{s}\to \Varid{f}\;\Varid{a},\Varid{f}\;\Varid{b}\to \Varid{t})}.

\hypertarget{deriving-a-new-profunctor-optic-a-case-study}{%
\section{Deriving a new profunctor optic: a case
study}\label{deriving-a-new-profunctor-optic-a-case-study}}

So far we have looked at optics that have a known profunctor encoding.
To illustrate the overall insight gained on profunctor optics and their
derivations, we explore the derivation of a new profunctor optic.

\hypertarget{definition-5}{%
\subsection{Definition}\label{definition-5}}

In an important part of the literature about lenses, lenses feature a
\ensuremath{\Varid{create}} function along with the usual \ensuremath{\Varid{get}} and \ensuremath{\Varid{put}}. We dub this new
type of lens \emph{achromatic lens}. \bigskip

\begin{defn}[Achromatic lens]

We call achromatic lens a value of the following datatype:

\begin{hscode}\SaveRestoreHook
\column{B}{@{}>{\hspre}l<{\hspost}@{}}%
\column{3}{@{}>{\hspre}l<{\hspost}@{}}%
\column{E}{@{}>{\hspre}l<{\hspost}@{}}%
\>[B]{}\mathbf{data}\;\Conid{AchLens}\;\Varid{a}\;\Varid{b}\;\Varid{s}\;\Varid{t}\mathrel{=}\Conid{AchLens}\;\{\mskip1.5mu {}\<[E]%
\\
\>[B]{}\hsindent{3}{}\<[3]%
\>[3]{}\Varid{get}\mathbin{::}\Varid{s}\to \Varid{a},{}\<[E]%
\\
\>[B]{}\hsindent{3}{}\<[3]%
\>[3]{}\Varid{put}\mathbin{::}\Varid{b}\to \Varid{s}\to \Varid{t},{}\<[E]%
\\
\>[B]{}\hsindent{3}{}\<[3]%
\>[3]{}\Varid{create}\mathbin{::}\Varid{b}\to \Varid{t}{}\<[E]%
\\
\>[B]{}\mskip1.5mu\}{}\<[E]%
\ColumnHook
\end{hscode}\resethooks

Typical laws for such lenses are the following:

\begin{itemize}
\tightlist
\item
  (GetPut) \ensuremath{\Varid{get}\;(\Varid{put}\;\Varid{b}\;\Varid{s})\mathrel{=}\Varid{b}}
\item
  (PutGet) \ensuremath{\Varid{put}\;(\Varid{get}\;\Varid{s})\;\Varid{s}\mathrel{=}\Varid{s}}
\item
  (PutPut) \ensuremath{\Varid{put}\;\Varid{b}\hsdot{\circ }{.\;}\Varid{put}\;\Varid{b'}\mathrel{=}\Varid{put}\;\Varid{b}}
\item
  (GetCreate) \ensuremath{\Varid{get}\hsdot{\circ }{.\;}\Varid{create}\mathrel{=}\Varid{id}}
\end{itemize}

They extend the usual lens very-well-behavedness laws.

\end{defn}

\bigskip

\begin{prop}

\ensuremath{\Conid{AchLens}} is an optic family.

\end{prop}

\begin{proof}

All achromatic lenses are lenses, therefore we can reuse most of the
definition of the \ensuremath{\Conid{OpticFamily}\;\Conid{Lens}} instance.

\begin{hscode}\SaveRestoreHook
\column{B}{@{}>{\hspre}l<{\hspost}@{}}%
\column{5}{@{}>{\hspre}l<{\hspost}@{}}%
\column{9}{@{}>{\hspre}l<{\hspost}@{}}%
\column{11}{@{}>{\hspre}l<{\hspost}@{}}%
\column{13}{@{}>{\hspre}l<{\hspost}@{}}%
\column{E}{@{}>{\hspre}l<{\hspost}@{}}%
\>[B]{}\mathbf{instance}\;\Conid{OpticFamily}\;\Conid{AchLens}\;\mathbf{where}{}\<[E]%
\\
\>[B]{}\hsindent{5}{}\<[5]%
\>[5]{}\Varid{injOptic}\mathbin{::}(\Varid{s}\to \Varid{a})\to (\Varid{b}\to \Varid{t})\to \Conid{AchLens}\;\Varid{a}\;\Varid{b}\;\Varid{s}\;\Varid{t}{}\<[E]%
\\
\>[B]{}\hsindent{5}{}\<[5]%
\>[5]{}\Varid{injOptic}\;\Varid{f}\;\Varid{g}\mathrel{=}\Conid{AchLens}\;\Varid{get}\;\Varid{put}\;\Varid{create}{}\<[E]%
\\
\>[5]{}\hsindent{4}{}\<[9]%
\>[9]{}\mathbf{where}{}\<[E]%
\\
\>[9]{}\hsindent{2}{}\<[11]%
\>[11]{}\Varid{get}\mathrel{=}\Varid{f}{}\<[E]%
\\
\>[9]{}\hsindent{2}{}\<[11]%
\>[11]{}\Varid{put}\;\Varid{b}\;\anonymous \mathrel{=}\Varid{g}\;\Varid{b}{}\<[E]%
\\
\>[9]{}\hsindent{2}{}\<[11]%
\>[11]{}\Varid{create}\mathrel{=}\Varid{g}{}\<[E]%
\\[\blanklineskip]%
\>[B]{}\hsindent{5}{}\<[5]%
\>[5]{}(\circ_{op})\mathbin{::}\Conid{AchLens}\;\Varid{a}\;\Varid{b}\;\Varid{s}\;\Varid{t}\to \Conid{AchLens}\;\Varid{x}\;\Varid{y}\;\Varid{a}\;\Varid{b}\to \Conid{AchLens}\;\Varid{x}\;\Varid{y}\;\Varid{s}\;\Varid{t}{}\<[E]%
\\
\>[B]{}\hsindent{5}{}\<[5]%
\>[5]{}(\circ_{op})\;(\Conid{AchLens}\;\Varid{get}_{\mathrm{1}}\;\Varid{put}_{\mathrm{1}}\;\Varid{create}_{\mathrm{1}})\;(\Conid{AchLens}\;\Varid{get}_{\mathrm{2}}\;\Varid{put}_{\mathrm{2}}\;\Varid{create}_{\mathrm{2}})\mathrel{=}{}\<[E]%
\\
\>[5]{}\hsindent{4}{}\<[9]%
\>[9]{}\Conid{AchLens}\;\Varid{get}\;\Varid{put}\;\Varid{create}{}\<[E]%
\\
\>[9]{}\hsindent{2}{}\<[11]%
\>[11]{}\mathbf{where}{}\<[E]%
\\
\>[11]{}\hsindent{2}{}\<[13]%
\>[13]{}\Varid{get}\mathrel{=}\Varid{get}_{\mathrm{2}}\hsdot{\circ }{.\;}\Varid{get}_{\mathrm{1}}{}\<[E]%
\\
\>[11]{}\hsindent{2}{}\<[13]%
\>[13]{}\Varid{put}\;\Varid{y}\;\Varid{s}\mathrel{=}\Varid{put}_{\mathrm{1}}\;(\Varid{put}_{\mathrm{2}}\;\Varid{y}\;(\Varid{get}_{\mathrm{1}}\;\Varid{s}))\;\Varid{s}{}\<[E]%
\\
\>[11]{}\hsindent{2}{}\<[13]%
\>[13]{}\Varid{create}\mathrel{=}\Varid{create1}\hsdot{\circ }{.\;}\Varid{create}_{\mathrm{2}}{}\<[E]%
\\[\blanklineskip]%
\>[B]{}\hsindent{5}{}\<[5]%
\>[5]{}\Varid{mapOptic}\mathbin{::}\Conid{AchLens}\;\Varid{a}\;\Varid{b}\;\Varid{s}\;\Varid{t}\to (\Varid{a}\to \Varid{b})\to (\Varid{s}\to \Varid{t}){}\<[E]%
\\
\>[B]{}\hsindent{5}{}\<[5]%
\>[5]{}\Varid{mapOptic}\;(\Conid{AchLens}\;\Varid{get}\;\Varid{put}\;\Varid{create})\;\Varid{f}\;\Varid{s}\mathrel{=}\Varid{put}\;(\Varid{f}\;(\Varid{get}\;\Varid{s}))\;\Varid{s}{}\<[E]%
\ColumnHook
\end{hscode}\resethooks

The proof of the laws is essentially the same as for \ensuremath{\Conid{Lens}}, so we omit
it.

\end{proof}

\bigskip

\begin{rmk}

We could have chosen \ensuremath{\Varid{mapOptic}\;(\Conid{AchLens}\;\{\mskip1.5mu \mathinner{\ldotp\ldotp}\mskip1.5mu\})\;\Varid{f}\mathrel{=}\Varid{create}\hsdot{\circ }{.\;}\Varid{f}\hsdot{\circ }{.\;}\Varid{get}}.
However, this would have changed the structure of \ensuremath{\Conid{Functorize}\;\Conid{AchLens}}.
Notably, \ensuremath{\Varid{enhanceFop}} would verify \ensuremath{\Varid{create}\hsdot{\circ }{.\;}\Varid{get}\mathrel{=}\Varid{id}} and
\ensuremath{\Varid{get}\hsdot{\circ }{.\;}\Varid{create}\mathrel{=}\Varid{id}}, thus \ensuremath{\Conid{Functorize}\;\Conid{AchLens}} would be reduced to the
\ensuremath{\Conid{Id}} functor, and \ensuremath{\Conid{AchLens}} would not have a profunctor encoding.

\end{rmk}

\bigskip

\hypertarget{analyzing-the-functorization-3}{%
\subsection{Analyzing the
functorization}\label{analyzing-the-functorization-3}}

Let \ensuremath{\Varid{f}\;\mathbin{\in}\;\Conid{Functorize}\;\Conid{AchLens}}. Let
\ensuremath{\Conid{AchLens}\;\Varid{get}\;\Varid{put}\;\Varid{create}\mathrel{=}\Varid{enhanceFop}\;@\Varid{f}}.

From \textcite{lens}, we know that \ensuremath{\Varid{enhanceFop}\;@\Varid{f}} is a lawful \ensuremath{\Conid{Lens}}.

By parametricity, \ensuremath{\Varid{get}\hsdot{\circ }{.\;}\Varid{create}\mathbin{::}\mathbf{forall}\;\Varid{a}\hsforall \hsdot{\circ }{.\;}\Varid{a}\to \Varid{a}} must be the
identity, therefore \ensuremath{\Varid{enhanceFop}\;@\Varid{f}} is also a lawful \ensuremath{\Conid{AchLens}}.

Since it is a lawful lens, we know that \ensuremath{\Varid{f}\cong (\Varid{f}\;(),\mathbin{-})}.

This is not enough to characterize \ensuremath{\Conid{Functorize}\;\Conid{AchLens}} though.
Achromatic lenses have a little more structure than that: they have a
value \ensuremath{\Varid{create}\;()\mathbin{::}\Varid{f}\;()}.

We can therefore define a subclass of \ensuremath{\Conid{IsProduct}} that contains the
functors \ensuremath{\Varid{f}} that also have a distinguished \ensuremath{\Varid{f}\;()}:

\begin{hscode}\SaveRestoreHook
\column{B}{@{}>{\hspre}l<{\hspost}@{}}%
\column{5}{@{}>{\hspre}l<{\hspost}@{}}%
\column{E}{@{}>{\hspre}l<{\hspost}@{}}%
\>[B]{}\mathbf{class}\;\Conid{IsProduct}\;\Varid{f}\Rightarrow \Conid{IsPointedProduct}\;\Varid{f}\;\mathbf{where}{}\<[E]%
\\
\>[B]{}\hsindent{5}{}\<[5]%
\>[5]{}\Varid{f1}\mathbin{::}\Varid{f}\;(){}\<[E]%
\ColumnHook
\end{hscode}\resethooks

We know that \ensuremath{\Varid{f}\;\mathbin{\in}\;\Conid{IsPointedProduct}}. Conversely, \ensuremath{\Conid{IsPointedProduct}}
embeds into \ensuremath{\Conid{Functorize}\;\Conid{AchLens}}:

\begin{hscode}\SaveRestoreHook
\column{B}{@{}>{\hspre}l<{\hspost}@{}}%
\column{5}{@{}>{\hspre}l<{\hspost}@{}}%
\column{9}{@{}>{\hspre}l<{\hspost}@{}}%
\column{11}{@{}>{\hspre}l<{\hspost}@{}}%
\column{E}{@{}>{\hspre}l<{\hspost}@{}}%
\>[B]{}\mathbf{instance}\;\Conid{IsPointedProduct}\;\Varid{f}\Rightarrow \Conid{Functorize}\;\Conid{AchLens}\;\Varid{f}\;\mathbf{where}{}\<[E]%
\\
\>[B]{}\hsindent{5}{}\<[5]%
\>[5]{}\Varid{enhanceFop}\mathbin{::}\Conid{AchLens}\;\Varid{a}\;\Varid{b}\;(\Varid{f}\;\Varid{a})\;(\Varid{f}\;\Varid{b}){}\<[E]%
\\
\>[B]{}\hsindent{5}{}\<[5]%
\>[5]{}\Varid{enhanceFop}\mathrel{=}\Conid{AchLens}\;\Varid{get}\;\Varid{put}\;\Varid{create}{}\<[E]%
\\
\>[5]{}\hsindent{4}{}\<[9]%
\>[9]{}\mathbf{where}{}\<[E]%
\\
\>[9]{}\hsindent{2}{}\<[11]%
\>[11]{}\Varid{get}\mathrel{=}\Varid{snd}\hsdot{\circ }{.\;}\Varid{toProduct}{}\<[E]%
\\
\>[9]{}\hsindent{2}{}\<[11]%
\>[11]{}\Varid{put}\;\Varid{b}\mathrel{=}\Varid{fromProduct}\hsdot{\circ }{.\;}\Varid{fmap}\;(\Varid{const}\;\Varid{b})\hsdot{\circ }{.\;}\Varid{toProduct}{}\<[E]%
\\
\>[9]{}\hsindent{2}{}\<[11]%
\>[11]{}\Varid{create}\;\Varid{b}\mathrel{=}\Varid{fromProduct}\;(\Varid{f1},\Varid{b}){}\<[E]%
\ColumnHook
\end{hscode}\resethooks

The laws are mostly identical to then lens case. In particular, the
wedge condition similarly holds by parametricity.

Thus \ensuremath{\Conid{Functorize}\;\Conid{AchLens}\cong \Conid{IsPointedProduct}}.

\hypertarget{deriving-a-profunctor-encoding-3}{%
\subsection{Deriving a profunctor
encoding}\label{deriving-a-profunctor-encoding-3}}

By the derivation theorem (\textcite{derivation_theorem}), to get a
profunctor encoding it is both necessary and sufficient to construct an
inverse to \ensuremath{\Varid{unfunctorize}\mathbin{::}\Conid{IsoOptic}\;(\Conid{Functorize}\;\Conid{AchLens})\leadsto \Conid{AchLens}}.

The \ensuremath{\Conid{Lens}\leadsto \Conid{IsoOptic}\;(\Conid{Functorize}\;\Conid{Lens})} transformation is defined using
the \ensuremath{(\Varid{c},\mathbin{-})} functor. Since achromatic lenses are lenses, we expect the
transformation to be similar. However we cannot create an instance for
the \ensuremath{(\Varid{c},\mathbin{-})} functor since we cannot write the
\ensuremath{\Varid{create}\mathbin{::}\mathbf{forall}\;\Varid{a}\hsforall \hsdot{\circ }{.\;}\Varid{a}\to (\Varid{c},\Varid{a})} function for an arbitrary type \ensuremath{\Varid{c}}. We
would need a default value for \ensuremath{\Varid{c}} to fill in the left side of the
product.

To get ``\ensuremath{\Varid{c}} with a default value'', we can simply wrap \ensuremath{\Varid{c}} in a \ensuremath{\Conid{Maybe}}
container, and consider the \ensuremath{(\Conid{Maybe}\;\Varid{c},\mathbin{-})} functor instead.

Since \ensuremath{\Conid{Functorize}\;\Conid{AchLens}\cong \Conid{IsPointedProduct}}, we only have to write
an \ensuremath{\Conid{IsPointedProduct}} instance for this functor. We can also reuse
\ensuremath{\mathbf{instance}\;\Conid{IsProduct}\;((,)\;\Varid{x})} with \ensuremath{\Varid{x}\mathrel{=}\Conid{Maybe}\;\Varid{c}}.

\begin{hscode}\SaveRestoreHook
\column{B}{@{}>{\hspre}l<{\hspost}@{}}%
\column{5}{@{}>{\hspre}l<{\hspost}@{}}%
\column{E}{@{}>{\hspre}l<{\hspost}@{}}%
\>[B]{}\mathbf{instance}\;\Conid{IsPointedProduct}\;((,)\;(\Conid{Maybe}\;\Varid{c}))\;\mathbf{where}{}\<[E]%
\\
\>[B]{}\hsindent{5}{}\<[5]%
\>[5]{}\Varid{f1}\mathrel{=}(\Conid{Nothing},()){}\<[E]%
\ColumnHook
\end{hscode}\resethooks

We automatically get a \ensuremath{\Conid{Functorize}\;\Conid{AchLens}} instance for \ensuremath{(\Conid{Maybe}\;\Varid{c},\mathbin{-})}.
\ensuremath{\Varid{enhanceFop}} is then equal to:

\begin{hscode}\SaveRestoreHook
\column{B}{@{}>{\hspre}l<{\hspost}@{}}%
\column{E}{@{}>{\hspre}l<{\hspost}@{}}%
\>[B]{}\Varid{enhanceFop}\mathbin{::}\Conid{AchLens}\;\Varid{a}\;\Varid{b}\;(\Conid{Maybe}\;\Varid{c},\Varid{a})\;(\Conid{Maybe}\;\Varid{c},\Varid{b}){}\<[E]%
\\
\>[B]{}\Varid{enhanceFop}\mathrel{=}\Conid{AchLens}\;\Varid{snd}\;(\lambda \Varid{b}\to \Varid{fmap}\;(\Varid{const}\;\Varid{b}))\;((,)\;\Conid{Nothing}){}\<[E]%
\ColumnHook
\end{hscode}\resethooks

We are now ready to write \ensuremath{\Varid{achLensToIso}}:

\begin{hscode}\SaveRestoreHook
\column{B}{@{}>{\hspre}l<{\hspost}@{}}%
\column{5}{@{}>{\hspre}l<{\hspost}@{}}%
\column{7}{@{}>{\hspre}l<{\hspost}@{}}%
\column{E}{@{}>{\hspre}l<{\hspost}@{}}%
\>[B]{}\Varid{achLensToIso}\mathbin{::}\Conid{AchLens}\leadsto \Conid{IsoOptic}\;\Conid{IsPointedProduct}{}\<[E]%
\\
\>[B]{}\Varid{achLensToIso}\;(\Conid{AchLens}\;\Varid{get}\;\Varid{put}\;\Varid{create})\mathrel{=}\Conid{IsoOptic}\;\alpha\;\beta{}\<[E]%
\\
\>[B]{}\hsindent{5}{}\<[5]%
\>[5]{}\mathbf{where}{}\<[E]%
\\
\>[5]{}\hsindent{2}{}\<[7]%
\>[7]{}\alpha\mathbin{::}\Varid{s}\to (\Conid{Maybe}\;\Varid{s},\Varid{a}){}\<[E]%
\\
\>[5]{}\hsindent{2}{}\<[7]%
\>[7]{}\alpha\;\Varid{s}\mathrel{=}(\Conid{Just}\;\Varid{s},\Varid{get}\;\Varid{s}){}\<[E]%
\\
\>[5]{}\hsindent{2}{}\<[7]%
\>[7]{}\beta\mathbin{::}(\Conid{Maybe}\;\Varid{s},\Varid{b})\to \Varid{t}{}\<[E]%
\\
\>[5]{}\hsindent{2}{}\<[7]%
\>[7]{}\beta\;(\Conid{Just}\;\Varid{s},\Varid{b})\mathrel{=}\Varid{put}\;\Varid{b}\;\Varid{s}{}\<[E]%
\\
\>[5]{}\hsindent{2}{}\<[7]%
\>[7]{}\beta\;(\Conid{Nothing},\Varid{b})\mathrel{=}\Varid{create}\;\Varid{b}{}\<[E]%
\ColumnHook
\end{hscode}\resethooks

We need \ensuremath{\Varid{achLensToIso}} to be the inverse of \ensuremath{\Varid{unfunctorize}}; see
\textcite{proof_deriving-a-profunctor-encoding-3} for the proof.
\bigskip

Hence
\ensuremath{\Conid{AchLens}\cong \Conid{ProfOptic}\;(\Conid{Functorize}\;\Conid{AchLens})\cong \Conid{ProfOptic}\;\Conid{IsPointedProduct}}.

The associated profunctor encoding is as follows:

\begin{hscode}\SaveRestoreHook
\column{B}{@{}>{\hspre}l<{\hspost}@{}}%
\column{5}{@{}>{\hspre}l<{\hspost}@{}}%
\column{E}{@{}>{\hspre}l<{\hspost}@{}}%
\>[B]{}\mathbf{class}\;\Conid{Distinguished}\;\Varid{a}\;\mathbf{where}{}\<[E]%
\\
\>[B]{}\hsindent{5}{}\<[5]%
\>[5]{}\Varid{distinguished}\mathbin{::}\Varid{a}{}\<[E]%
\\[\blanklineskip]%
\>[B]{}\mathbf{class}\;\Conid{PointedCartesian}\;\Varid{p}\;\mathbf{where}{}\<[E]%
\\
\>[B]{}\hsindent{5}{}\<[5]%
\>[5]{}\Varid{pointedSecond}\mathbin{::}\Conid{Distinguished}\;\Varid{c}\Rightarrow \Varid{p}\;\Varid{a}\;\Varid{b}\to \Varid{p}\;(\Varid{c},\Varid{a})\;(\Varid{c},\Varid{b}){}\<[E]%
\\[\blanklineskip]%
\>[B]{}\Conid{AchLens}\;\Varid{a}\;\Varid{b}\;\Varid{s}\;\Varid{t}\cong \mathbf{forall}\;\Varid{p}\hsforall \hsdot{\circ }{.\;}\Conid{PointedCartesian}\;\Varid{p}\Rightarrow \Varid{p}\;\Varid{a}\;\Varid{b}\to \Varid{p}\;\Varid{s}\;\Varid{t}{}\<[E]%
\ColumnHook
\end{hscode}\resethooks

We also have \ensuremath{\Conid{AchLens}\cong \Conid{IsoOptic}\;\Conid{IsPointedProduct}}.

This means that achromatic lenses, like lenses, have an equivalent
residual expression. Unlike lenses however, the residual has a
distinguished element:

\begin{hscode}\SaveRestoreHook
\column{B}{@{}>{\hspre}l<{\hspost}@{}}%
\column{E}{@{}>{\hspre}l<{\hspost}@{}}%
\>[B]{}\Conid{AchLens}\;\Varid{a}\;\Varid{b}\;\Varid{s}\;\Varid{t}\cong \mathbf{exists}\;\Varid{c}\hsexists \hsdot{\circ }{.\;}(\Varid{c},\Varid{s}\to \Varid{c}\mathbin{\mathbf{\times}}\Varid{a},\Varid{c}\mathbin{\mathbf{\times}}\Varid{b}\to \Varid{t}){}\<[E]%
\ColumnHook
\end{hscode}\resethooks

\hypertarget{conclusion}{%
\chapter{Conclusion}\label{conclusion}}

\hypertarget{discussion}{%
\section{Discussion}\label{discussion}}

This thesis provides two major insights into the nature of profunctor
optics:

First, profunctor optics are best understood through the lens
\footnote{pun intended} of isomorphism optics. Iso optics are a
reasonably intuitive description of optics and, as demonstrated by the
number of theorems derived in \textcite{isomorphism-optics}, are easy to
reason about formally.

An isomorphism optic encoding can also provide insights into the
structure of an optic family, like the residual representation of
lenses.

Secondly, given an optic family, we know which functor monoid is of
interest to try and derive a profunctor encoding. Moreover, the
derivation theorem greatly limits the difficulty of doing so. Most of
the remaining difficulty is in choosing the right functor when defining
the inverse to \ensuremath{\Varid{unfunctorize}}. Properties of the encoding (preservation
of dimap, composition and map) follow for free.

Finally, we have tested the usefulness of those theorems by successfully
using them to derive properties of several optic families.

\hypertarget{related-work}{%
\section{Related work}\label{related-work}}

Even though lenses have been extensively studied in the bidirectional
transformations
literature\autocite{first_lenses}\autocite{mlenses}\autocite{coalg_update_lenses},
work on optics (a.k.a \emph{functional references}) has mostly been done
informally through blog posts and IRC chats, and often by
non-researchers. The only published paper about optics is
\autocite{profunctor_optics}.

Despite this lack of academic attention, the research on optics is quite
active.

Isomorphism optics have been described as an alternative to other
encodings for lenses by T. Van Laarhoven\autocite{isomorphism_lenses},
and extended to the other common optics by R.
O'Connor\autocite{r6_profunctor_hierarchy}.

On profunctor optics, two representation theorems have been proved in
categorical terms for lenses, adapters and
prisms\autocite{repr_theorem_for_sopf}\autocite{bartosz_profops_categorically}.
The second one\autocite{bartosz_profops_categorically} notably proves
the isomorphism between iso optics and profunctor optics.

As far as the author knows however, this paper is the first attempt at a
truly general characterization of profunctor optics. Moreover, most
derivations usually gloss over important properties of the encoding like
preservation of composition.

\hypertarget{future-work}{%
\section{Future work}\label{future-work}}

The representation theorem opens up the way to deriving a lot of
properties about profunctor optics. Among the areas of interest would be
a general characterization of lawfulness for optic families, that could
be easily checked for a given profunctor optic.

The representation theorem may also be extended to describe ``one-way''
optics (Folds, Getters, \ldots).

The derivation theorem gives insights into the derivation of new
profunctor optics, but the story is not complete. Ideally, the
appropriate reverse transformation should be derived generically,
assuming some stronger axioms on optic families.

Finally, profunctor optics seem to have deep algebraic structure. For
example, \ensuremath{\Conid{Enhancing}} profunctors can be described as coalgebras of an
appropriate comonad. Profunctor optics also form a semi-lattice, with
composition acting as join\autocite{profunctor_optics}. Discovering the
full algebraic picture would pave the way to more general and robust
abstractions in the form of powerful tools for the functional
programmer.

\printbibliography

\begin{appendices}

\hypertarget{working-implementation}{%
\chapter{Working implementation}\label{working-implementation}}

To make the reasoning clearer, a number of simplifications have been
made in the paper regarding the actual definitions and implementations
presented. This appendix includes a working version of the notions
presented. It can be compiled using GHC 8.0.2 and the
\text{\ttfamily constraints}\autocite{constraints} package.

\input{code.tex}

\hypertarget{proofs}{%
\chapter{Proofs}\label{proofs}}

Some longer proofs are included here instead of the main body.

\hypertarget{section}{%
\section{\texorpdfstring{\Cref{lens_is_profop_isproduct}}{}}\label{section}}

\label{proof_lens_is_profop_isproduct}

Direction \ensuremath{\Conid{Cartesian}\;\Varid{p}\Rightarrow \Conid{Enhancing}\;\Conid{IsProduct}\;\Varid{p}\Rightarrow \Conid{Cartesian}\;\Varid{p}}:

\begin{hscode}\SaveRestoreHook
\column{B}{@{}>{\hspre}l<{\hspost}@{}}%
\column{3}{@{}>{\hspre}l<{\hspost}@{}}%
\column{E}{@{}>{\hspre}l<{\hspost}@{}}%
\>[B]{}\Varid{second'}{}\<[E]%
\\
\>[B]{}\mathrel{=}\mbox{\commentbegin   \ensuremath{\mathbf{instance}\;\Conid{Enhancing}\;\Conid{IsProduct}\;\Varid{p}\Rightarrow \Conid{Cartesian}\;\Varid{p}}   \commentend}{}\<[E]%
\\
\>[B]{}\hsindent{3}{}\<[3]%
\>[3]{}\Varid{enhance}\;@(\Varid{c},\mathbin{-}){}\<[E]%
\\
\>[B]{}\mathrel{=}\mbox{\commentbegin   \ensuremath{\mathbf{instance}\;\Conid{Cartesian}\;\Varid{p}\Rightarrow \Conid{Enhancing}\;\Conid{IsProduct}\;\Varid{p}}   \commentend}{}\<[E]%
\\
\>[B]{}\hsindent{3}{}\<[3]%
\>[3]{}\Varid{dimap}\;(\Varid{fromProduct}\;@(\Varid{c},\mathbin{-}))\;(\Varid{toProduct}\;@(\Varid{c},\mathbin{-}))\hsdot{\circ }{.\;}\Varid{second}{}\<[E]%
\\
\>[B]{}\mathrel{=}\mbox{\commentbegin   \ensuremath{\Conid{Profunctor}} law   \commentend}{}\<[E]%
\\
\>[B]{}\hsindent{3}{}\<[3]%
\>[3]{}\Varid{dimap}\;(\Varid{fromProduct}\;@(\Varid{c},\mathbin{-}))\;\Varid{id}\hsdot{\circ }{.\;}\Varid{dimap}\;\Varid{id}\;(\Varid{toProduct}\;@(\Varid{c},\mathbin{-}))\hsdot{\circ }{.\;}\Varid{second}{}\<[E]%
\\
\>[B]{}\mathrel{=}\mbox{\commentbegin   \ensuremath{\mathbf{instance}\;\Conid{IsProduct}\;(\Varid{c},\mathbin{-})}   \commentend}{}\<[E]%
\\
\>[B]{}\hsindent{3}{}\<[3]%
\>[3]{}\Varid{dimap}\;(\Varid{fromProduct}\;@(\Varid{c},\mathbin{-}))\;\Varid{id}\hsdot{\circ }{.\;}\Varid{dimap}\;\Varid{id}\;(\lambda (\Varid{r},\Varid{a})\to ((\Varid{r},()),\Varid{a}))\hsdot{\circ }{.\;}\Varid{second}{}\<[E]%
\\
\>[B]{}\mathrel{=}\mbox{\commentbegin   definition of \ensuremath{\Varid{first}}   \commentend}{}\<[E]%
\\
\>[B]{}\hsindent{3}{}\<[3]%
\>[3]{}\Varid{dimap}\;(\Varid{fromProduct}\;@(\Varid{c},\mathbin{-}))\;\Varid{id}\hsdot{\circ }{.\;}\Varid{dimap}\;\Varid{id}\;(\Varid{first}\;(\lambda \Varid{r}\to (\Varid{r},())))\hsdot{\circ }{.\;}\Varid{second}{}\<[E]%
\\
\>[B]{}\mathrel{=}\mbox{\commentbegin   \ensuremath{\Conid{Cartesian}} law   \commentend}{}\<[E]%
\\
\>[B]{}\hsindent{3}{}\<[3]%
\>[3]{}\Varid{dimap}\;(\Varid{fromProduct}\;@(\Varid{c},\mathbin{-}))\;\Varid{id}\hsdot{\circ }{.\;}\Varid{dimap}\;(\Varid{first}\;(\lambda \Varid{r}\to (\Varid{r},())))\;\Varid{id}\hsdot{\circ }{.\;}\Varid{second}{}\<[E]%
\\
\>[B]{}\mathrel{=}\mbox{\commentbegin   definition of \ensuremath{\Varid{first}}   \commentend}{}\<[E]%
\\
\>[B]{}\hsindent{3}{}\<[3]%
\>[3]{}\Varid{dimap}\;(\Varid{fromProduct}\;@(\Varid{c},\mathbin{-}))\;\Varid{id}\hsdot{\circ }{.\;}\Varid{dimap}\;(\lambda (\Varid{r},\Varid{a})\to ((\Varid{r},()),\Varid{a}))\;\Varid{id}\hsdot{\circ }{.\;}\Varid{second}{}\<[E]%
\\
\>[B]{}\mathrel{=}\mbox{\commentbegin   \ensuremath{\mathbf{instance}\;\Conid{IsProduct}\;(\Varid{c},\mathbin{-})}   \commentend}{}\<[E]%
\\
\>[B]{}\hsindent{3}{}\<[3]%
\>[3]{}\Varid{dimap}\;(\Varid{fromProduct}\;@(\Varid{c},\mathbin{-}))\;\Varid{id}\hsdot{\circ }{.\;}\Varid{dimap}\;(\Varid{toProduct}\;@(\Varid{c},\mathbin{-}))\;\Varid{id}\hsdot{\circ }{.\;}\Varid{second}{}\<[E]%
\\
\>[B]{}\mathrel{=}\mbox{\commentbegin   \ensuremath{\Conid{Profunctor}} law   \commentend}{}\<[E]%
\\
\>[B]{}\hsindent{3}{}\<[3]%
\>[3]{}\Varid{dimap}\;(\Varid{toProduct}\;@(\Varid{c},\mathbin{-})\hsdot{\circ }{.\;}\Varid{fromProduct}\;@(\Varid{c},\mathbin{-}))\;\Varid{id}\hsdot{\circ }{.\;}\Varid{second}{}\<[E]%
\\
\>[B]{}\mathrel{=}\mbox{\commentbegin   \ensuremath{\Conid{IsProduct}} law   \commentend}{}\<[E]%
\\
\>[B]{}\hsindent{3}{}\<[3]%
\>[3]{}\Varid{dimap}\;\Varid{id}\;\Varid{id}\hsdot{\circ }{.\;}\Varid{second}{}\<[E]%
\\
\>[B]{}\mathrel{=}\mbox{\commentbegin   \ensuremath{\Conid{Profunctor}} law   \commentend}{}\<[E]%
\\
\>[B]{}\hsindent{3}{}\<[3]%
\>[3]{}\Varid{second}{}\<[E]%
\ColumnHook
\end{hscode}\resethooks

Direction
\ensuremath{\Conid{Enhancing}\;\Conid{IsProduct}\;\Varid{p}\Rightarrow \Conid{Cartesian}\;\Varid{p}\Rightarrow \Conid{Enhancing}\;\Conid{IsProduct}\;\Varid{p}}:

\begin{hscode}\SaveRestoreHook
\column{B}{@{}>{\hspre}l<{\hspost}@{}}%
\column{3}{@{}>{\hspre}l<{\hspost}@{}}%
\column{E}{@{}>{\hspre}l<{\hspost}@{}}%
\>[B]{}\Varid{enhance'}\;@\Varid{f}{}\<[E]%
\\
\>[B]{}\mathrel{=}\mbox{\commentbegin   \ensuremath{\mathbf{instance}\;\Conid{Cartesian}\;\Varid{p}\Rightarrow \Conid{Enhancing}\;\Conid{IsProduct}\;\Varid{p}}   \commentend}{}\<[E]%
\\
\>[B]{}\hsindent{3}{}\<[3]%
\>[3]{}\Varid{dimap}\;(\Varid{fromProduct}\;@\Varid{f})\;(\Varid{toProduct}\;@\Varid{f})\hsdot{\circ }{.\;}\Varid{second}{}\<[E]%
\\
\>[B]{}\mathrel{=}\mbox{\commentbegin   \ensuremath{\mathbf{instance}\;\Conid{Enhancing}\;\Conid{IsProduct}\;\Varid{p}\Rightarrow \Conid{Cartesian}\;\Varid{p}}   \commentend}{}\<[E]%
\\
\>[B]{}\hsindent{3}{}\<[3]%
\>[3]{}\Varid{dimap}\;(\Varid{fromProduct}\;@\Varid{f})\;(\Varid{toProduct}\;@\Varid{f})\hsdot{\circ }{.\;}\Varid{enhance}\;@(\Varid{f}\;(),\mathbin{-}){}\<[E]%
\\
\>[B]{}\mathrel{=}\mbox{\commentbegin   \ensuremath{\Conid{Profunctor}} law   \commentend}{}\<[E]%
\\
\>[B]{}\hsindent{3}{}\<[3]%
\>[3]{}\Varid{dimap}\;(\Varid{fromProduct}\;@\Varid{f})\;\Varid{id}\hsdot{\circ }{.\;}\Varid{dimap}\;\Varid{id}\;(\Varid{toProduct}\;@\Varid{f})\hsdot{\circ }{.\;}\Varid{enhance}\;@(\Varid{f}\;(),\mathbin{-}){}\<[E]%
\\
\>[B]{}\mathrel{=}\mbox{\commentbegin   \ensuremath{\Conid{Enhancing}} law   \commentend}{}\<[E]%
\\
\>[B]{}\hsindent{3}{}\<[3]%
\>[3]{}\Varid{dimap}\;(\Varid{fromProduct}\;@\Varid{f})\;\Varid{id}\hsdot{\circ }{.\;}\Varid{dimap}\;(\Varid{toProduct}\;@\Varid{f})\;\Varid{id}\hsdot{\circ }{.\;}\Varid{enhance}\;@\Varid{f}{}\<[E]%
\\
\>[B]{}\mathrel{=}\mbox{\commentbegin   \ensuremath{\Conid{Profunctor}} law   \commentend}{}\<[E]%
\\
\>[B]{}\hsindent{3}{}\<[3]%
\>[3]{}\Varid{dimap}\;(\Varid{toProduct}\;@\Varid{f}\hsdot{\circ }{.\;}\Varid{fromProduct}\;@\Varid{f})\;\Varid{id}\hsdot{\circ }{.\;}\Varid{enhance}\;@\Varid{f}{}\<[E]%
\\
\>[B]{}\mathrel{=}\mbox{\commentbegin   \ensuremath{\Conid{IsProduct}} law   \commentend}{}\<[E]%
\\
\>[B]{}\hsindent{3}{}\<[3]%
\>[3]{}\Varid{dimap}\;\Varid{id}\;\Varid{id}\hsdot{\circ }{.\;}\Varid{enhance}\;@\Varid{f}{}\<[E]%
\\
\>[B]{}\mathrel{=}\mbox{\commentbegin   \ensuremath{\Conid{Profunctor}} law   \commentend}{}\<[E]%
\\
\>[B]{}\hsindent{3}{}\<[3]%
\>[3]{}\Varid{enhance}\;@\Varid{f}{}\<[E]%
\ColumnHook
\end{hscode}\resethooks

\hypertarget{section-1}{%
\section{\texorpdfstring{\Cref{iso_is_opfam}}{}}\label{section-1}}

\label{proof_iso_is_opfam}

\begin{hscode}\SaveRestoreHook
\column{B}{@{}>{\hspre}l<{\hspost}@{}}%
\column{4}{@{}>{\hspre}l<{\hspost}@{}}%
\column{14}{@{}>{\hspre}l<{\hspost}@{}}%
\column{E}{@{}>{\hspre}l<{\hspost}@{}}%
\>[B]{}\Varid{injOptic}\;\Varid{f}\;\Varid{g}\ \circ_{op}\ \Conid{IsoOptic}\;\alpha\;\beta{}\<[E]%
\\
\>[B]{}\mathrel{=}{}\<[4]%
\>[4]{}\Conid{IsoOptic}\;(\Conid{Id}\hsdot{\circ }{.\;}\Varid{f})\;(\Varid{g}\hsdot{\circ }{.\;}\Varid{unId})\ \circ_{op}\ \Conid{IsoOptic}\;\alpha\;\beta{}\<[E]%
\\
\>[B]{}\mathrel{=}{}\<[4]%
\>[4]{}\Conid{IsoOptic}\;{}\<[14]%
\>[14]{}(\Conid{Compose}\hsdot{\circ }{.\;}\Varid{fmap}\;@\Conid{Id}\;\alpha\hsdot{\circ }{.\;}\Conid{Id}\hsdot{\circ }{.\;}\Varid{f})\;{}\<[E]%
\\
\>[14]{}(\Varid{g}\hsdot{\circ }{.\;}\Varid{unId}\hsdot{\circ }{.\;}\Varid{fmap}\;@\Conid{Id}\;\beta\hsdot{\circ }{.\;}\Varid{unCompose}){}\<[E]%
\\
\>[B]{}\mathrel{=}{}\<[4]%
\>[4]{}\Conid{IsoOptic}\;(\Conid{Compose}\hsdot{\circ }{.\;}\Conid{Id}\hsdot{\circ }{.\;}\alpha\hsdot{\circ }{.\;}\Varid{f})\;(\Varid{g}\hsdot{\circ }{.\;}\beta\hsdot{\circ }{.\;}\Varid{unId}\hsdot{\circ }{.\;}\Varid{unCompose}){}\<[E]%
\\
\>[B]{}\mathrel{=}{}\<[4]%
\>[4]{}\Conid{IsoOptic}\;(\Varid{unId}\hsdot{\circ }{.\;}\Varid{unCompose}\hsdot{\circ }{.\;}\Conid{Compose}\hsdot{\circ }{.\;}\Conid{Id}\hsdot{\circ }{.\;}\alpha\hsdot{\circ }{.\;}\Varid{f})\;(\Varid{g}\hsdot{\circ }{.\;}\beta){}\<[E]%
\\
\>[B]{}\mathrel{=}{}\<[4]%
\>[4]{}\Conid{IsoOptic}\;(\alpha\hsdot{\circ }{.\;}\Varid{f})\;(\Varid{g}\hsdot{\circ }{.\;}\beta){}\<[E]%
\ColumnHook
\end{hscode}\resethooks

Thus

\begin{hscode}\SaveRestoreHook
\column{B}{@{}>{\hspre}l<{\hspost}@{}}%
\column{E}{@{}>{\hspre}l<{\hspost}@{}}%
\>[B]{}\Varid{injOptic}\;\Varid{id}\;\Varid{id}\ \circ_{op}\ \Varid{l}\mathrel{=}\Varid{l}{}\<[E]%
\ColumnHook
\end{hscode}\resethooks

and

\begin{hscode}\SaveRestoreHook
\column{B}{@{}>{\hspre}l<{\hspost}@{}}%
\column{E}{@{}>{\hspre}l<{\hspost}@{}}%
\>[B]{}\Varid{injOptic}\;\Varid{f}\;\Varid{g}\ \circ_{op}\ \Varid{injOptic}\;\Varid{f'}\;\Varid{g'}\mathrel{=}\Varid{injOptic}\;(\Varid{f'}\hsdot{\circ }{.\;}\Varid{f})\;(\Varid{g}\hsdot{\circ }{.\;}\Varid{g'}){}\<[E]%
\ColumnHook
\end{hscode}\resethooks

\smallskip

\begin{hscode}\SaveRestoreHook
\column{B}{@{}>{\hspre}l<{\hspost}@{}}%
\column{4}{@{}>{\hspre}l<{\hspost}@{}}%
\column{14}{@{}>{\hspre}l<{\hspost}@{}}%
\column{E}{@{}>{\hspre}l<{\hspost}@{}}%
\>[B]{}\Conid{IsoOptic}\;\alpha\;\beta\ \circ_{op}\ \Varid{injOptic}\;\Varid{id}\;\Varid{id}{}\<[E]%
\\
\>[B]{}\mathrel{=}{}\<[4]%
\>[4]{}\Conid{IsoOptic}\;\alpha\;\beta\ \circ_{op}\ \Conid{IsoOptic}\;\Conid{Id}\;\Varid{unId}{}\<[E]%
\\
\>[B]{}\mathrel{=}{}\<[4]%
\>[4]{}\Conid{IsoOptic}\;{}\<[14]%
\>[14]{}(\Conid{Compose}\hsdot{\circ }{.\;}\Varid{fmap}\;@\Varid{f}\;\Conid{Id}\hsdot{\circ }{.\;}\alpha)\;(\beta\hsdot{\circ }{.\;}\Varid{fmap}\;@\Varid{f}\;\Varid{unId}\hsdot{\circ }{.\;}\Varid{unCompose}){}\<[E]%
\\
\>[B]{}\mathrel{=}\mbox{\commentbegin   equality of iso optics   \commentend}{}\<[E]%
\\
\>[B]{}\hsindent{4}{}\<[4]%
\>[4]{}\Conid{IsoOptic}\;(\Varid{fmap}\;@\Varid{f}\;\Varid{unId}\hsdot{\circ }{.\;}\Varid{unCompose}\hsdot{\circ }{.\;}\Conid{Compose}\hsdot{\circ }{.\;}\Varid{fmap}\;@\Varid{f}\;\Conid{Id}\hsdot{\circ }{.\;}\alpha)\;\beta{}\<[E]%
\\
\>[B]{}\mathrel{=}{}\<[4]%
\>[4]{}\Conid{IsoOptic}\;\alpha\;\beta{}\<[E]%
\ColumnHook
\end{hscode}\resethooks

Thus

\begin{hscode}\SaveRestoreHook
\column{B}{@{}>{\hspre}l<{\hspost}@{}}%
\column{E}{@{}>{\hspre}l<{\hspost}@{}}%
\>[B]{}\Varid{l}\ \circ_{op}\ \Varid{injOptic}\;\Varid{id}\;\Varid{id}\mathrel{=}\Varid{l}{}\<[E]%
\ColumnHook
\end{hscode}\resethooks

\smallskip

\begin{hscode}\SaveRestoreHook
\column{B}{@{}>{\hspre}l<{\hspost}@{}}%
\column{4}{@{}>{\hspre}l<{\hspost}@{}}%
\column{E}{@{}>{\hspre}l<{\hspost}@{}}%
\>[B]{}\Varid{mapOptic}\;(\Varid{injOptic}\;\Varid{f}\;\Varid{g})\;\Varid{h}{}\<[E]%
\\
\>[B]{}\mathrel{=}{}\<[4]%
\>[4]{}\Varid{g}\hsdot{\circ }{.\;}\Varid{unId}\hsdot{\circ }{.\;}\Varid{fmap}\;\Varid{h}\hsdot{\circ }{.\;}\Conid{Id}\hsdot{\circ }{.\;}\Varid{f}{}\<[E]%
\\
\>[B]{}\mathrel{=}{}\<[4]%
\>[4]{}\Varid{g}\hsdot{\circ }{.\;}\Varid{h}\hsdot{\circ }{.\;}\Varid{f}{}\<[E]%
\\[\blanklineskip]%
\>[B]{}\Varid{mapOptic}\;(\Conid{IsoOptic}\;\alpha\;\beta\ \circ_{op}\ \Conid{IsoOptic}\;\alpha'\;\beta')\;\Varid{h}{}\<[E]%
\\
\>[B]{}\mathrel{=}{}\<[4]%
\>[4]{}\beta\hsdot{\circ }{.\;}\Varid{fmap}\;\beta'\hsdot{\circ }{.\;}\Varid{unCompose}\hsdot{\circ }{.\;}\Varid{fmap}\;\Varid{h}\hsdot{\circ }{.\;}\Conid{Compose}\hsdot{\circ }{.\;}\Varid{fmap}\;\alpha'\hsdot{\circ }{.\;}\alpha{}\<[E]%
\\
\>[B]{}\mathrel{=}{}\<[4]%
\>[4]{}\beta\hsdot{\circ }{.\;}\Varid{fmap}\;\beta'\hsdot{\circ }{.\;}\Varid{fmap}\;(\Varid{fmap}\;\Varid{h})\hsdot{\circ }{.\;}\Varid{fmap}\;\alpha'\hsdot{\circ }{.\;}\alpha{}\<[E]%
\\
\>[B]{}\mathrel{=}{}\<[4]%
\>[4]{}\beta\hsdot{\circ }{.\;}\Varid{fmap}\;(\beta'\hsdot{\circ }{.\;}\Varid{fmap}\;\Varid{h}\hsdot{\circ }{.\;}\alpha')\hsdot{\circ }{.\;}\alpha{}\<[E]%
\\
\>[B]{}\mathrel{=}{}\<[4]%
\>[4]{}\Varid{mapOptic}\;(\Conid{IsoOptic}\;\alpha\;\beta)\;(\Varid{mapOptic}\;(\Conid{IsoOptic}\;\alpha'\;\beta')\;\Varid{h}){}\<[E]%
\ColumnHook
\end{hscode}\resethooks

Therefore \ensuremath{\Conid{IsoOptic}\;\sigma} is a valid optic family.

\hypertarget{section-2}{%
\section{\texorpdfstring{\Cref{enhancetoarrow_is_opfam_morphism}}{}}\label{section-2}}

\label{proof_enhancetoarrow_is_opfam_morphism}

Preservation of \ensuremath{\Varid{injOptic}}:

\begin{hscode}\SaveRestoreHook
\column{B}{@{}>{\hspre}l<{\hspost}@{}}%
\column{4}{@{}>{\hspre}l<{\hspost}@{}}%
\column{E}{@{}>{\hspre}l<{\hspost}@{}}%
\>[B]{}\Varid{enhanceToArrow}\;\Varid{enhanceOp}\;(\Varid{injOptic}\;\Varid{f}\;\Varid{g}){}\<[E]%
\\
\>[B]{}\mathrel{=}{}\<[4]%
\>[4]{}\Varid{injOptic}\;(\Conid{Id}\hsdot{\circ }{.\;}\Varid{f})\;(\Varid{g}\hsdot{\circ }{.\;}\Varid{unId})\ \circ_{op}\ \Varid{enhanceOp}\;@\Conid{Id}{}\<[E]%
\\
\>[B]{}\mathrel{=}\mbox{\commentbegin   \ensuremath{\Varid{injOptic}} law   \commentend}{}\<[E]%
\\
\>[B]{}\hsindent{4}{}\<[4]%
\>[4]{}\Varid{injOptic}\;\Varid{f}\;\Varid{g}\ \circ_{op}\ \Varid{injOptic}\;\Conid{Id}\;\Varid{unId}\ \circ_{op}\ \Varid{enhanceOp}\;@\Conid{Id}{}\<[E]%
\\
\>[B]{}\mathrel{=}\mbox{\commentbegin   \ensuremath{\Conid{Enhanceable}} law   \commentend}{}\<[E]%
\\
\>[B]{}\hsindent{4}{}\<[4]%
\>[4]{}\Varid{injOptic}\;\Varid{f}\;\Varid{g}\ \circ_{op}\ \Varid{injOptic}\;\Conid{Id}\;\Varid{unId}\ \circ_{op}\ \Varid{injOptic}\;\Varid{unId}\;\Conid{Id}{}\<[E]%
\\
\>[B]{}\mathrel{=}\mbox{\commentbegin   \ensuremath{\Varid{injOptic}} law   \commentend}{}\<[E]%
\\
\>[B]{}\hsindent{4}{}\<[4]%
\>[4]{}\Varid{injOptic}\;\Varid{f}\;\Varid{g}\ \circ_{op}\ \Varid{injOptic}\;(\Varid{unId}\hsdot{\circ }{.\;}\Conid{Id})\;(\Varid{unId}\hsdot{\circ }{.\;}\Conid{Id}){}\<[E]%
\\
\>[B]{}\mathrel{=}{}\<[4]%
\>[4]{}\Varid{injOptic}\;\Varid{f}\;\Varid{g}\ \circ_{op}\ \Varid{injOptic}\;\Varid{id}\;\Varid{id}{}\<[E]%
\\
\>[B]{}\mathrel{=}\mbox{\commentbegin   \ensuremath{\Conid{OpticFamily}} law   \commentend}{}\<[E]%
\\
\>[B]{}\hsindent{4}{}\<[4]%
\>[4]{}\Varid{injOptic}\;\Varid{f}\;\Varid{g}{}\<[E]%
\ColumnHook
\end{hscode}\resethooks

Preservation of composition:

\begin{hscode}\SaveRestoreHook
\column{B}{@{}>{\hspre}l<{\hspost}@{}}%
\column{4}{@{}>{\hspre}l<{\hspost}@{}}%
\column{7}{@{}>{\hspre}l<{\hspost}@{}}%
\column{E}{@{}>{\hspre}l<{\hspost}@{}}%
\>[B]{}\Varid{enhanceToArrow}\;\Varid{enhanceOp}\;(\Conid{IsoOptic}\;\alpha\;\beta)\ \circ_{op}\ {}\<[E]%
\\
\>[B]{}\hsindent{7}{}\<[7]%
\>[7]{}\Varid{enhanceToArrow}\;\Varid{enhanceOp}\;(\Conid{IsoOptic}\;\alpha'\;\beta'){}\<[E]%
\\
\>[B]{}\mathrel{=}{}\<[4]%
\>[4]{}\Varid{injOptic}\;\alpha\;\beta\ \circ_{op}\ \Varid{enhanceOp}\;@\Varid{f}\ \circ_{op}\ \Varid{injOptic}\;\alpha'\;\beta'\ \circ_{op}\ \Varid{enhanceOp}\;@\Varid{g}{}\<[E]%
\\
\>[B]{}\mathrel{=}\mbox{\commentbegin   \ensuremath{\Varid{enhanceOp}} law   \commentend}{}\<[E]%
\\
\>[B]{}\hsindent{4}{}\<[4]%
\>[4]{}\Varid{injOptic}\;\alpha\;\beta\ \circ_{op}\ \Varid{injOptic}\;(\Varid{fmap}\;\alpha')\;(\Varid{fmap}\;\beta')\ \circ_{op}\ {}\<[E]%
\\
\>[4]{}\hsindent{3}{}\<[7]%
\>[7]{}\Varid{enhanceOp}\;@\Varid{f}\ \circ_{op}\ \Varid{enhanceOp}\;@\Varid{g}{}\<[E]%
\\
\>[B]{}\mathrel{=}\mbox{\commentbegin   \ensuremath{\Varid{injOptic}} law   \commentend}{}\<[E]%
\\
\>[B]{}\hsindent{4}{}\<[4]%
\>[4]{}\Varid{injOptic}\;(\Varid{fmap}\;\alpha'\hsdot{\circ }{.\;}\alpha)\;(\beta\hsdot{\circ }{.\;}\Varid{fmap}\;\beta')\ \circ_{op}\ {}\<[E]%
\\
\>[4]{}\hsindent{3}{}\<[7]%
\>[7]{}\Varid{enhanceOp}\;@\Varid{f}\ \circ_{op}\ \Varid{enhanceOp}\;@\Varid{g}{}\<[E]%
\\
\>[B]{}\mathrel{=}\mbox{\commentbegin   \ensuremath{\Varid{injOptic}} law   \commentend}{}\<[E]%
\\
\>[B]{}\hsindent{4}{}\<[4]%
\>[4]{}\Varid{injOptic}\;(\Conid{Compose}\hsdot{\circ }{.\;}\Varid{fmap}\;\alpha'\hsdot{\circ }{.\;}\alpha)\;(\beta\hsdot{\circ }{.\;}\Varid{fmap}\;\beta'\hsdot{\circ }{.\;}\Varid{unCompose})\ \circ_{op}\ {}\<[E]%
\\
\>[4]{}\hsindent{3}{}\<[7]%
\>[7]{}\Varid{injOptic}\;\Varid{unCompose}\;\Conid{Compose}\ \circ_{op}\ \Varid{enhanceOp}\;@\Varid{f}\ \circ_{op}\ \Varid{enhanceOp}\;@\Varid{g}{}\<[E]%
\\
\>[B]{}\mathrel{=}\mbox{\commentbegin   \ensuremath{\Varid{enhanceOp}} law   \commentend}{}\<[E]%
\\
\>[B]{}\hsindent{4}{}\<[4]%
\>[4]{}\Varid{injOptic}\;(\Conid{Compose}\hsdot{\circ }{.\;}\Varid{fmap}\;\alpha'\hsdot{\circ }{.\;}\alpha)\;(\beta\hsdot{\circ }{.\;}\Varid{fmap}\;\beta'\hsdot{\circ }{.\;}\Varid{unCompose})\ \circ_{op}\ {}\<[E]%
\\
\>[4]{}\hsindent{3}{}\<[7]%
\>[7]{}\Varid{enhanceOp}\;@(\Conid{Compose}\;\Varid{f}\;\Varid{g}){}\<[E]%
\\
\>[B]{}\mathrel{=}{}\<[4]%
\>[4]{}\Varid{enhanceToArrow}\;\Varid{enhanceOp}\;(\Conid{IsoOptic}\;\alpha\;\beta\ \circ_{op}\ \Conid{IsoOptic}\;\alpha'\;\beta'){}\<[E]%
\ColumnHook
\end{hscode}\resethooks

Preservation of \ensuremath{\Varid{mapOptic}}:

\begin{hscode}\SaveRestoreHook
\column{B}{@{}>{\hspre}l<{\hspost}@{}}%
\column{4}{@{}>{\hspre}l<{\hspost}@{}}%
\column{E}{@{}>{\hspre}l<{\hspost}@{}}%
\>[B]{}\Varid{mapOptic}\;(\Varid{enhanceToArrow}\;\Varid{enhanceOp}\;(\Conid{IsoOptic}\;\alpha\;\beta)){}\<[E]%
\\
\>[B]{}\mathrel{=}{}\<[4]%
\>[4]{}\Varid{mapOptic}\;(\Varid{injOptic}\;\alpha\;\beta\ \circ_{op}\ \Varid{enhanceOp}){}\<[E]%
\\
\>[B]{}\mathrel{=}\mbox{\commentbegin   \ensuremath{\Varid{mapOptic}} law   \commentend}{}\<[E]%
\\
\>[B]{}\hsindent{4}{}\<[4]%
\>[4]{}\Varid{mapOptic}\;(\Varid{injOptic}\;\alpha\;\beta)\hsdot{\circ }{.\;}\Varid{mapOptic}\;\Varid{enhanceOp}{}\<[E]%
\\
\>[B]{}\mathrel{=}\mbox{\commentbegin   \ensuremath{\Varid{mapOptic}} law   \commentend}{}\<[E]%
\\
\>[B]{}\hsindent{4}{}\<[4]%
\>[4]{}(\lambda \Varid{f}\to \beta\hsdot{\circ }{.\;}\Varid{f}\hsdot{\circ }{.\;}\alpha)\hsdot{\circ }{.\;}\Varid{mapOptic}\;\Varid{enhanceOp}{}\<[E]%
\\
\>[B]{}\mathrel{=}{}\<[4]%
\>[4]{}\lambda \Varid{f}\to \beta\hsdot{\circ }{.\;}\Varid{mapOptic}\;\Varid{enhanceOp}\;\Varid{f}\hsdot{\circ }{.\;}\alpha{}\<[E]%
\\
\>[B]{}\mathrel{=}\mbox{\commentbegin   \ensuremath{\Varid{enhanceOp}} law   \commentend}{}\<[E]%
\\
\>[B]{}\hsindent{4}{}\<[4]%
\>[4]{}\lambda \Varid{f}\to \beta\hsdot{\circ }{.\;}\Varid{fmap}\;\Varid{f}\hsdot{\circ }{.\;}\alpha{}\<[E]%
\\
\>[B]{}\mathrel{=}{}\<[4]%
\>[4]{}\Varid{mapOptic}\;(\Conid{IsoOptic}\;\alpha\;\beta){}\<[E]%
\ColumnHook
\end{hscode}\resethooks

\hypertarget{section-3}{%
\section{\texorpdfstring{\Cref{functorize_is_fmonoid}}{}}\label{section-3}}

\label{proof_functorize_is_fmonoid}

\begin{hscode}\SaveRestoreHook
\column{B}{@{}>{\hspre}l<{\hspost}@{}}%
\column{4}{@{}>{\hspre}l<{\hspost}@{}}%
\column{E}{@{}>{\hspre}l<{\hspost}@{}}%
\>[B]{}\Varid{mapOptic}\;(\Varid{enhanceFop}\;@\Conid{Id})\;\Varid{id}{}\<[E]%
\\
\>[B]{}\mathrel{=}{}\<[4]%
\>[4]{}\Varid{mapOptic}\;(\Varid{injOp}\;\Varid{unId}\;\Conid{Id})\;\Varid{id}{}\<[E]%
\\
\>[B]{}\mathrel{=}\mbox{\commentbegin   \ensuremath{\Varid{mapOptic}} law   \commentend}{}\<[E]%
\\
\>[B]{}\hsindent{4}{}\<[4]%
\>[4]{}\Conid{Id}\hsdot{\circ }{.\;}\Varid{id}\hsdot{\circ }{.\;}\Varid{unId}{}\<[E]%
\\
\>[B]{}\mathrel{=}{}\<[4]%
\>[4]{}\Varid{id}{}\<[E]%
\\[\blanklineskip]%
\>[B]{}\Varid{enhanceFop}\;@\Conid{Id}\ \circ_{op}\ \Varid{injOptic}\;\Varid{f}\;\Varid{g}{}\<[E]%
\\
\>[B]{}\mathrel{=}{}\<[4]%
\>[4]{}\Varid{injOptic}\;(\Varid{f}\hsdot{\circ }{.\;}\Varid{unId})\;(\Conid{Id}\hsdot{\circ }{.\;}\Varid{g}){}\<[E]%
\\
\>[B]{}\mathrel{=}\mbox{\commentbegin   \ensuremath{\mathbf{instance}\;\Conid{Functor}\;\Conid{Id}}   \commentend}{}\<[E]%
\\
\>[B]{}\hsindent{4}{}\<[4]%
\>[4]{}\Varid{injOptic}\;(\Varid{unId}\hsdot{\circ }{.\;}\Varid{fmap}\;\Varid{f})\;(\Varid{fmap}\;\Varid{g}\hsdot{\circ }{.\;}\Conid{Id}){}\<[E]%
\\
\>[B]{}\mathrel{=}\Varid{injOptic}\;(\Varid{fmap}\;\Varid{f})\;(\Varid{fmap}\;\Varid{g})\ \circ_{op}\ \Varid{enhanceFop}\;@\Conid{Id}{}\<[E]%
\\[\blanklineskip]%
\>[B]{}\Varid{injOptic}\;\Varid{id}\;\alpha\ \circ_{op}\ \Varid{enhanceFop}\;@\Conid{Id}{}\<[E]%
\\
\>[B]{}\mathrel{=}\Varid{injOptic}\;\Varid{unId}\;(\alpha\hsdot{\circ }{.\;}\Conid{Id}){}\<[E]%
\ColumnHook
\end{hscode}\resethooks

\smallskip

\begin{hscode}\SaveRestoreHook
\column{B}{@{}>{\hspre}l<{\hspost}@{}}%
\column{4}{@{}>{\hspre}l<{\hspost}@{}}%
\column{7}{@{}>{\hspre}l<{\hspost}@{}}%
\column{9}{@{}>{\hspre}l<{\hspost}@{}}%
\column{E}{@{}>{\hspre}l<{\hspost}@{}}%
\>[B]{}\Varid{mapOptic}\;(\Varid{enhanceFop}\;@(\Conid{Compose}\;\Varid{f}\;\Varid{g}))\;\Varid{id}{}\<[E]%
\\
\>[B]{}\mathrel{=}\mbox{\commentbegin   \ensuremath{\Varid{mapOptic}} law   \commentend}{}\<[E]%
\\
\>[B]{}\hsindent{4}{}\<[4]%
\>[4]{}(\Varid{mapOptic}\;(\Varid{injOptic}\;\Varid{unCompose}\;\Conid{Compose}){}\<[E]%
\\
\>[4]{}\hsindent{3}{}\<[7]%
\>[7]{}\hsdot{\circ }{.\;}\Varid{mapOptic}\;(\Varid{enhanceFop}\;@\Varid{f}){}\<[E]%
\\
\>[4]{}\hsindent{3}{}\<[7]%
\>[7]{}\hsdot{\circ }{.\;}\Varid{mapOptic}\;(\Varid{enhanceFop}\;@\Varid{g}))\;\Varid{id}{}\<[E]%
\\
\>[B]{}\mathrel{=}\mbox{\commentbegin   \ensuremath{\Varid{enhanceFop}} law   \commentend}{}\<[E]%
\\
\>[B]{}\hsindent{4}{}\<[4]%
\>[4]{}(\Varid{mapOptic}\;(\Varid{injOptic}\;\Varid{unCompose}\;\Conid{Compose}))\;{}\<[E]%
\\
\>[4]{}\hsindent{3}{}\<[7]%
\>[7]{}(\Varid{mapOptic}\;(\Varid{enhanceFop}\;@\Varid{f})\;\Varid{id}){}\<[E]%
\\
\>[B]{}\mathrel{=}\mbox{\commentbegin   \ensuremath{\Varid{enhanceFop}} law   \commentend}{}\<[E]%
\\
\>[B]{}\hsindent{4}{}\<[4]%
\>[4]{}\Varid{mapOptic}\;(\Varid{injOptic}\;\Varid{unCompose}\;\Conid{Compose})\;\Varid{id}{}\<[E]%
\\
\>[B]{}\mathrel{=}\mbox{\commentbegin   \ensuremath{\Varid{mapOptic}} law   \commentend}{}\<[E]%
\\
\>[B]{}\hsindent{4}{}\<[4]%
\>[4]{}\Conid{Compose}\hsdot{\circ }{.\;}\Varid{id}\hsdot{\circ }{.\;}\Varid{unCompose}{}\<[E]%
\\
\>[B]{}\mathrel{=}\Varid{id}{}\<[E]%
\\[\blanklineskip]%
\>[B]{}\Varid{enhanceFop}\;@(\Conid{Compose}\;\Varid{f}\;\Varid{g})\ \circ_{op}\ \Varid{injOptic}\;\Varid{f}\;\Varid{g}{}\<[E]%
\\
\>[B]{}\mathrel{=}{}\<[4]%
\>[4]{}\Varid{injOptic}\;\Varid{unCompose}\;\Conid{Compose}\ \circ_{op}\ \Varid{enhanceFop}\;@\Varid{f}\ \circ_{op}\ {}\<[E]%
\\
\>[4]{}\hsindent{5}{}\<[9]%
\>[9]{}\Varid{enhanceFop}\;@\Varid{g}\ \circ_{op}\ \Varid{injOptic}\;\Varid{f}\;\Varid{g}{}\<[E]%
\\
\>[B]{}\mathrel{=}{}\<[4]%
\>[4]{}\Varid{injOptic}\;\Varid{unCompose}\;\Conid{Compose}\ \circ_{op}\ {}\<[E]%
\\
\>[4]{}\hsindent{3}{}\<[7]%
\>[7]{}\Varid{injOptic}\;(\Varid{fmap}\;(\Varid{fmap}\;\Varid{f}))\;(\Varid{fmap}\;(\Varid{fmap}\;\Varid{g}))\ \circ_{op}\ {}\<[E]%
\\
\>[4]{}\hsindent{3}{}\<[7]%
\>[7]{}\Varid{enhanceFop}\;@\Varid{f}\ \circ_{op}\ \Varid{enhanceFop}\;@\Varid{g}{}\<[E]%
\\
\>[B]{}\mathrel{=}\mbox{\commentbegin   \ensuremath{\mathbf{instance}\;\Conid{Functor}\;\Conid{Compose}}   \commentend}{}\<[E]%
\\
\>[B]{}\hsindent{4}{}\<[4]%
\>[4]{}\Varid{injOptic}\;(\Varid{fmap}\;\Varid{f})\;(\Varid{fmap}\;\Varid{g})\ \circ_{op}\ \Varid{enhanceFop}\;@(\Conid{Compose}\;\Varid{f}\;\Varid{g}){}\<[E]%
\ColumnHook
\end{hscode}\resethooks

\hypertarget{section-4}{%
\section{\texorpdfstring{\Cref{deriving-a-profunctor-encoding-3}}{}}\label{section-4}}

\label{proof_deriving-a-profunctor-encoding-3}

\begin{hscode}\SaveRestoreHook
\column{B}{@{}>{\hspre}l<{\hspost}@{}}%
\column{4}{@{}>{\hspre}l<{\hspost}@{}}%
\column{8}{@{}>{\hspre}l<{\hspost}@{}}%
\column{10}{@{}>{\hspre}l<{\hspost}@{}}%
\column{11}{@{}>{\hspre}l<{\hspost}@{}}%
\column{27}{@{}>{\hspre}l<{\hspost}@{}}%
\column{28}{@{}>{\hspre}l<{\hspost}@{}}%
\column{E}{@{}>{\hspre}l<{\hspost}@{}}%
\>[B]{}\Varid{unfunctorize}\;(\Conid{IsoOptic}\;\alpha\;\beta){}\<[E]%
\\
\>[B]{}\mathrel{=}{}\<[4]%
\>[4]{}\Varid{injOptic}\;\alpha\;\beta\ \circ_{op}\ \Varid{enhanceFop}{}\<[E]%
\\
\>[B]{}\mathrel{=}{}\<[4]%
\>[4]{}\Varid{injOptic}\;\alpha\;\beta\ \circ_{op}\ (\Conid{AchLens}\;\Varid{snd}\;(\lambda \Varid{b}\to \Varid{fmap}\;(\Varid{const}\;\Varid{b}))\;((,)\;\Conid{Nothing})){}\<[E]%
\\
\>[B]{}\mathrel{=}{}\<[4]%
\>[4]{}\Conid{AchLens}\;(\Varid{snd}\hsdot{\circ }{.\;}\alpha)\;(\lambda \Varid{b}\to \beta\hsdot{\circ }{.\;}\Varid{fmap}\;(\Varid{const}\;\Varid{b})\hsdot{\circ }{.\;}\alpha)\;(\beta\hsdot{\circ }{.\;}(,)\;\Conid{Nothing}){}\<[E]%
\\[\blanklineskip]%
\>[B]{}\Varid{unfunctorize}\;(\Varid{achLensToIso}\;(\Conid{AchLens}\;\Varid{get}\;\Varid{put}\;\Varid{create})){}\<[E]%
\\
\>[B]{}\mathrel{=}{}\<[4]%
\>[4]{}\Varid{unfunctorize}\;(\Conid{IsoOptic}\;{}\<[28]%
\>[28]{}(\lambda \Varid{s}\to (\Conid{Just}\;\Varid{s},\Varid{get}\;\Varid{s}))\;{}\<[E]%
\\
\>[28]{}(\lambda (\Varid{ms},\Varid{b})\to \Varid{maybe}\;(\Varid{create}\;\Varid{b})\;(\Varid{put}\;\Varid{b})\;\Varid{ms}){}\<[E]%
\\
\>[B]{}\mathrel{=}{}\<[4]%
\>[4]{}\Conid{AchLens}\;{}\<[E]%
\\
\>[4]{}\hsindent{4}{}\<[8]%
\>[8]{}(\Varid{snd}\hsdot{\circ }{.\;}(\lambda \Varid{s}\to (\Conid{Just}\;\Varid{s},\Varid{get}\;\Varid{s})))\;{}\<[E]%
\\
\>[4]{}\hsindent{4}{}\<[8]%
\>[8]{}(\lambda \Varid{b}\to (\lambda (\Varid{ms},\Varid{b})\to \Varid{maybe}\;(\Varid{create}\;\Varid{b})\;(\Varid{put}\;\Varid{b})\;\Varid{ms})\hsdot{\circ }{.\;}{}\<[E]%
\\
\>[8]{}\hsindent{3}{}\<[11]%
\>[11]{}\Varid{fmap}\;(\Varid{const}\;\Varid{b})\hsdot{\circ }{.\;}(\lambda \Varid{s}\to (\Conid{Just}\;\Varid{s},\Varid{get}\;\Varid{s})))\;{}\<[E]%
\\
\>[4]{}\hsindent{4}{}\<[8]%
\>[8]{}((\lambda (\Varid{ms},\Varid{b})\to \Varid{maybe}\;(\Varid{create}\;\Varid{b})\;(\Varid{put}\;\Varid{b})\;\Varid{ms})\hsdot{\circ }{.\;}(,)\;\Conid{Nothing}){}\<[E]%
\\
\>[B]{}\mathrel{=}{}\<[4]%
\>[4]{}\Conid{AchLens}\;{}\<[E]%
\\
\>[4]{}\hsindent{4}{}\<[8]%
\>[8]{}\Varid{get}\;{}\<[E]%
\\
\>[4]{}\hsindent{4}{}\<[8]%
\>[8]{}(\lambda \Varid{b}\to (\lambda (\Varid{ms},\Varid{b})\to \Varid{maybe}\;(\Varid{create}\;\Varid{b})\;(\Varid{put}\;\Varid{b})\;\Varid{ms})\hsdot{\circ }{.\;}(\lambda \Varid{s}\to (\Conid{Just}\;\Varid{s},\Varid{b})))\;{}\<[E]%
\\
\>[4]{}\hsindent{4}{}\<[8]%
\>[8]{}\Varid{create}{}\<[E]%
\\
\>[B]{}\mathrel{=}{}\<[4]%
\>[4]{}\Conid{AchLens}\;\Varid{get}\;(\lambda \Varid{b}\to (\lambda \Varid{s}\to \Varid{put}\;\Varid{b}\;\Varid{s}))\;\Varid{create}{}\<[E]%
\\
\>[B]{}\mathrel{=}{}\<[4]%
\>[4]{}\Conid{AchLens}\;\Varid{get}\;\Varid{put}\;\Varid{create}{}\<[E]%
\\[\blanklineskip]%
\>[B]{}\Varid{achLensToIso}\;(\Varid{unfunctorize}\;(\Conid{IsoOptic}\;\alpha\;\beta)){}\<[E]%
\\
\>[B]{}\mathrel{=}{}\<[4]%
\>[4]{}\Varid{achLensToIso}\;(\Conid{AchLens}\;{}\<[27]%
\>[27]{}(\Varid{snd}\hsdot{\circ }{.\;}\alpha)\;{}\<[E]%
\\
\>[27]{}(\lambda \Varid{b}\to \beta\hsdot{\circ }{.\;}\Varid{fmap}\;(\Varid{const}\;\Varid{b})\hsdot{\circ }{.\;}\alpha)\;{}\<[E]%
\\
\>[27]{}(\beta\hsdot{\circ }{.\;}(,)\;\Conid{Nothing})){}\<[E]%
\\
\>[B]{}\mathrel{=}{}\<[4]%
\>[4]{}\Conid{IsoOptic}\;{}\<[E]%
\\
\>[4]{}\hsindent{4}{}\<[8]%
\>[8]{}(\lambda \Varid{s}\to (\Conid{Just}\;\Varid{s},\Varid{snd}\;(\alpha\;\Varid{s})))\;{}\<[E]%
\\
\>[4]{}\hsindent{4}{}\<[8]%
\>[8]{}(\lambda \mathbf{case}{}\<[E]%
\\
\>[8]{}\hsindent{2}{}\<[10]%
\>[10]{}(\Conid{Just}\;\Varid{s},\Varid{b})\to \beta\hsdot{\circ }{.\;}\Varid{fmap}\;(\Varid{const}\;\Varid{b})\hsdot{\circ }{.\;}\alpha\mathbin{\$}\Varid{s})\;{}\<[E]%
\\
\>[8]{}\hsindent{2}{}\<[10]%
\>[10]{}(\Conid{Nothing},\Varid{b})\to \beta\;(\Conid{Nothing},\Varid{b}){}\<[E]%
\\
\>[4]{}\hsindent{4}{}\<[8]%
\>[8]{}){}\<[E]%
\\
\>[B]{}\mathrel{=}{}\<[4]%
\>[4]{}\Conid{IsoOptic}\;{}\<[E]%
\\
\>[4]{}\hsindent{4}{}\<[8]%
\>[8]{}(\lambda \Varid{s}\to (\Conid{Just}\;\Varid{s},\Varid{snd}\;(\alpha\;\Varid{s})))\;{}\<[E]%
\\
\>[4]{}\hsindent{4}{}\<[8]%
\>[8]{}(\beta\hsdot{\circ }{.\;}\lambda \mathbf{case}{}\<[E]%
\\
\>[8]{}\hsindent{2}{}\<[10]%
\>[10]{}(\Conid{Just}\;\Varid{s},\Varid{b})\to (\Varid{fst}\;(\alpha\;\Varid{s}),\Varid{b}){}\<[E]%
\\
\>[8]{}\hsindent{2}{}\<[10]%
\>[10]{}(\Conid{Nothing},\Varid{b})\to (\Conid{Nothing},\Varid{b}){}\<[E]%
\\
\>[4]{}\hsindent{4}{}\<[8]%
\>[8]{}){}\<[E]%
\\
\>[B]{}\mathrel{=}{}\<[4]%
\>[4]{}\Conid{IsoOptic}\;{}\<[E]%
\\
\>[4]{}\hsindent{4}{}\<[8]%
\>[8]{}(\lambda \Varid{s}\to (\Conid{Just}\;\Varid{s},\Varid{snd}\;(\alpha\;\Varid{s})))\;{}\<[E]%
\\
\>[4]{}\hsindent{4}{}\<[8]%
\>[8]{}(\beta\hsdot{\circ }{.\;}\Varid{first}\;(\lambda \mathbf{case}{}\<[E]%
\\
\>[8]{}\hsindent{2}{}\<[10]%
\>[10]{}\Conid{Just}\;\Varid{s}\to \Varid{fst}\;(\alpha\;\Varid{s}){}\<[E]%
\\
\>[8]{}\hsindent{2}{}\<[10]%
\>[10]{}\Conid{Nothing}\to \Conid{Nothing}{}\<[E]%
\\
\>[4]{}\hsindent{4}{}\<[8]%
\>[8]{})){}\<[E]%
\\
\>[B]{}\mathrel{=}{}\<[4]%
\>[4]{}\Conid{IsoOptic}\;(\lambda \Varid{s}\to (\Conid{Just}\;\Varid{s},\Varid{snd}\;(\alpha\;\Varid{s})))\;(\beta\hsdot{\circ }{.\;}\Varid{first}\;(\Varid{maybe}\;\Conid{Nothing}\;(\Varid{fst}\hsdot{\circ }{.\;}\alpha))){}\<[E]%
\\
\>[B]{}\mathrel{=}\mbox{\commentbegin   equality of iso optics   \commentend}{}\<[E]%
\\
\>[B]{}\hsindent{4}{}\<[4]%
\>[4]{}\Conid{IsoOptic}\;(\Varid{first}\;(\Varid{maybe}\;\Conid{Nothing}\;(\Varid{fst}\hsdot{\circ }{.\;}\alpha))\hsdot{\circ }{.\;}\lambda \Varid{s}\to (\Conid{Just}\;\Varid{s},\Varid{snd}\;(\alpha\;\Varid{s})))\;\beta{}\<[E]%
\\
\>[B]{}\mathrel{=}{}\<[4]%
\>[4]{}\Conid{IsoOptic}\;(\lambda \Varid{s}\to (\Varid{maybe}\;\Conid{Nothing}\;(\Varid{fst}\hsdot{\circ }{.\;}\alpha)\;(\Conid{Just}\;\Varid{s}),\Varid{snd}\;(\alpha\;\Varid{s})))\;\beta{}\<[E]%
\\
\>[B]{}\mathrel{=}{}\<[4]%
\>[4]{}\Conid{IsoOptic}\;(\lambda \Varid{s}\to \Varid{fst}\;(\alpha\;\Varid{s}),\Varid{snd}\;(\alpha\;\Varid{s})))\;\beta{}\<[E]%
\\
\>[B]{}\mathrel{=}{}\<[4]%
\>[4]{}\Conid{IsoOptic}\;\alpha\;\beta{}\<[E]%
\ColumnHook
\end{hscode}\resethooks

\end{appendices}

\end{document}

%% file: code.tex
%auto-ignore
% Generated by pandoc syntax highlighting from impl.hs
\begin{Shaded}
\begin{Highlighting}[]
\OtherTok{\{{-}\# LANGUAGE ConstraintKinds \#{-}\}}
\OtherTok{\{{-}\# LANGUAGE RankNTypes \#{-}\}}
\OtherTok{\{{-}\# LANGUAGE ExistentialQuantification \#{-}\}}
\OtherTok{\{{-}\# LANGUAGE TypeOperators \#{-}\}}
\OtherTok{\{{-}\# LANGUAGE MultiParamTypeClasses \#{-}\}}
\OtherTok{\{{-}\# LANGUAGE KindSignatures \#{-}\}}
\OtherTok{\{{-}\# LANGUAGE FlexibleContexts \#{-}\}}
\OtherTok{\{{-}\# LANGUAGE FlexibleInstances \#{-}\}}
\OtherTok{\{{-}\# LANGUAGE ScopedTypeVariables \#{-}\}}
\OtherTok{\{{-}\# LANGUAGE InstanceSigs \#{-}\}}
\OtherTok{\{{-}\# LANGUAGE DeriveFunctor \#{-}\}}
\OtherTok{\{{-}\# LANGUAGE UndecidableSuperClasses \#{-}\}}
\OtherTok{\{{-}\# LANGUAGE TypeApplications \#{-}\}}
\OtherTok{\{{-}\# LANGUAGE AllowAmbiguousTypes \#{-}\}}

\KeywordTok{module} \DataTypeTok{ComposeRecords} \KeywordTok{where}

\KeywordTok{import} \DataTypeTok{Data.Constraint}

\KeywordTok{type}\NormalTok{ (}\OperatorTok{\textasciitilde{}\textasciitilde{}>}\NormalTok{) p q }\OtherTok{=} \KeywordTok{forall}\NormalTok{ a b s t}\OperatorTok{.}\NormalTok{ p a b s t }\OtherTok{{-}>}\NormalTok{ q a b s t}

\KeywordTok{type}\NormalTok{ (}\OperatorTok{+}\NormalTok{) }\OtherTok{=} \DataTypeTok{Either}

\OtherTok{lmap ::}\NormalTok{ (a }\OtherTok{{-}>}\NormalTok{ a\textquotesingle{}) }\OtherTok{{-}>}\NormalTok{ a }\OperatorTok{+}\NormalTok{ b }\OtherTok{{-}>}\NormalTok{ a\textquotesingle{} }\OperatorTok{+}\NormalTok{ b}
\NormalTok{lmap f }\OtherTok{=} \FunctionTok{either}\NormalTok{ (}\DataTypeTok{Left} \OperatorTok{.}\NormalTok{ f) }\DataTypeTok{Right}

\KeywordTok{newtype} \DataTypeTok{Id}\NormalTok{ a }\OtherTok{=} \DataTypeTok{Id}\NormalTok{ \{}\OtherTok{ unId ::}\NormalTok{ a \}}

\KeywordTok{newtype} \DataTypeTok{Compose}\NormalTok{ f g x }\OtherTok{=} \DataTypeTok{Compose}\NormalTok{ \{}\OtherTok{ unCompose ::}\NormalTok{ f (g x) \}}

\KeywordTok{class}\NormalTok{ c }\DataTypeTok{Id} \OtherTok{=>} \DataTypeTok{FunctorMonoid}\NormalTok{ c }\KeywordTok{where}
\OtherTok{    functorProof ::} \KeywordTok{forall}\NormalTok{ f}\OperatorTok{.}\NormalTok{ c f }\OperatorTok{:{-}} \DataTypeTok{Functor}\NormalTok{ f}
\OtherTok{    fmapFMonoid ::} \KeywordTok{forall}\NormalTok{ f x y}\OperatorTok{.}\NormalTok{ c f }\OtherTok{=>}\NormalTok{ (x }\OtherTok{{-}>}\NormalTok{ y) }\OtherTok{{-}>}\NormalTok{ f x }\OtherTok{{-}>}\NormalTok{ f y}
\NormalTok{    fmapFMonoid }\OtherTok{=} \FunctionTok{fmap}\NormalTok{ \textbackslash{}\textbackslash{} (}\OtherTok{functorProof ::}\NormalTok{ c f }\OperatorTok{:{-}} \DataTypeTok{Functor}\NormalTok{ f)}

\OtherTok{    composeProof ::} \KeywordTok{forall}\NormalTok{ f g}\OperatorTok{.}\NormalTok{ (c f, c g) }\OperatorTok{:{-}}\NormalTok{ c (}\DataTypeTok{Compose}\NormalTok{ f g)}

\KeywordTok{instance} \DataTypeTok{Functor} \DataTypeTok{Id} \KeywordTok{where}
    \FunctionTok{fmap}\NormalTok{ f }\OtherTok{=} \DataTypeTok{Id} \OperatorTok{.}\NormalTok{ f }\OperatorTok{.}\NormalTok{ unId}

\KeywordTok{instance}\NormalTok{ (}\DataTypeTok{Functor}\NormalTok{ f, }\DataTypeTok{Functor}\NormalTok{ g) }\OtherTok{=>} \DataTypeTok{Functor}\NormalTok{ (}\DataTypeTok{Compose}\NormalTok{ f g) }\KeywordTok{where}
    \FunctionTok{fmap}\NormalTok{ f }\OtherTok{=} \DataTypeTok{Compose} \OperatorTok{.} \FunctionTok{fmap}\NormalTok{ (}\FunctionTok{fmap}\NormalTok{ f) }\OperatorTok{.}\NormalTok{ unCompose}

\KeywordTok{instance} \DataTypeTok{FunctorMonoid} \DataTypeTok{Functor} \KeywordTok{where}
\NormalTok{    functorProof }\OtherTok{=} \DataTypeTok{Sub} \DataTypeTok{Dict}
\NormalTok{    composeProof }\OtherTok{=} \DataTypeTok{Sub} \DataTypeTok{Dict}

\KeywordTok{class} \DataTypeTok{Profunctor}\NormalTok{ p }\KeywordTok{where}
\OtherTok{    dimap ::}\NormalTok{ (a\textquotesingle{} }\OtherTok{{-}>}\NormalTok{ a) }\OtherTok{{-}>}\NormalTok{ (b }\OtherTok{{-}>}\NormalTok{ b\textquotesingle{}) }\OtherTok{{-}>}\NormalTok{ (p a b }\OtherTok{{-}>}\NormalTok{ p a\textquotesingle{} b\textquotesingle{})}

\KeywordTok{instance} \DataTypeTok{Profunctor}\NormalTok{ (}\OtherTok{{-}>}\NormalTok{) }\KeywordTok{where}
\NormalTok{    dimap f g h }\OtherTok{=}\NormalTok{ g }\OperatorTok{.}\NormalTok{ h }\OperatorTok{.}\NormalTok{ f}

\KeywordTok{class}\NormalTok{ (}\DataTypeTok{FunctorMonoid}\NormalTok{ c, }\DataTypeTok{Profunctor}\NormalTok{ p) }\OtherTok{=>} \DataTypeTok{Enhancing}\NormalTok{ c p }\KeywordTok{where}
\OtherTok{    enhance ::} \KeywordTok{forall}\NormalTok{ f a b}\OperatorTok{.}\NormalTok{ c f }\OtherTok{=>}\NormalTok{ p a b }\OtherTok{{-}>}\NormalTok{ p (f a) (f b)}

\KeywordTok{instance} \DataTypeTok{FunctorMonoid}\NormalTok{ c }\OtherTok{=>} \DataTypeTok{Enhancing}\NormalTok{ c (}\OtherTok{{-}>}\NormalTok{) }\KeywordTok{where}
\NormalTok{    enhance }\OtherTok{=}\NormalTok{ fmapFMonoid }\OperatorTok{@}\NormalTok{c}

\KeywordTok{class} \DataTypeTok{OpticFamily}\NormalTok{ (}\OtherTok{op ::} \OperatorTok{*} \OtherTok{{-}>} \OperatorTok{*} \OtherTok{{-}>} \OperatorTok{*} \OtherTok{{-}>} \OperatorTok{*} \OtherTok{{-}>} \OperatorTok{*}\NormalTok{) }\KeywordTok{where}
\OtherTok{    idOptic ::}\NormalTok{ op a b a b}
\NormalTok{    idOptic }\OtherTok{=}\NormalTok{ injOptic }\FunctionTok{id} \FunctionTok{id}
\OtherTok{    composeOptic ::}\NormalTok{ op a b s t }\OtherTok{{-}>}\NormalTok{ op x y a b }\OtherTok{{-}>}\NormalTok{ op x y s t}
\OtherTok{    injOptic ::}\NormalTok{ (s }\OtherTok{{-}>}\NormalTok{ a) }\OtherTok{{-}>}\NormalTok{ (b }\OtherTok{{-}>}\NormalTok{ t) }\OtherTok{{-}>}\NormalTok{ op a b s t}
\OtherTok{    mapOptic ::}\NormalTok{ op a b s t }\OtherTok{{-}>}\NormalTok{ (a }\OtherTok{{-}>}\NormalTok{ b) }\OtherTok{{-}>}\NormalTok{ (s }\OtherTok{{-}>}\NormalTok{ t)}

\OtherTok{(.:.) ::} \DataTypeTok{OpticFamily}\NormalTok{ op }\OtherTok{=>}\NormalTok{ op a b s t }\OtherTok{{-}>}\NormalTok{ op x y a b }\OtherTok{{-}>}\NormalTok{ op x y s t}
\NormalTok{(}\OperatorTok{.:.}\NormalTok{) }\OtherTok{=}\NormalTok{ composeOptic}

\KeywordTok{instance} \DataTypeTok{OpticFamily}\NormalTok{ op }\OtherTok{=>} \DataTypeTok{Profunctor}\NormalTok{ (op a b) }\KeywordTok{where}
\NormalTok{    dimap f g op }\OtherTok{=}\NormalTok{ injOptic f g }\OperatorTok{.:.}\NormalTok{ op}

\KeywordTok{data} \DataTypeTok{IsoOptic}\NormalTok{ c a b s t }\OtherTok{=} \KeywordTok{forall}\NormalTok{ f}\OperatorTok{.}\NormalTok{ c f }\OtherTok{=>} \DataTypeTok{IsoOptic}\NormalTok{ (s }\OtherTok{{-}>}\NormalTok{ f a) (f b }\OtherTok{{-}>}\NormalTok{ t)}

\OtherTok{enhanceIso ::}\NormalTok{ c f }\OtherTok{=>} \DataTypeTok{IsoOptic}\NormalTok{ c a b (f a) (f b)}
\NormalTok{enhanceIso }\OtherTok{=} \DataTypeTok{IsoOptic} \FunctionTok{id} \FunctionTok{id}

\OtherTok{composeIso ::} \KeywordTok{forall}\NormalTok{ c a b ta tb tta ttb}\OperatorTok{.} \DataTypeTok{FunctorMonoid}\NormalTok{ c}
      \OtherTok{=>} \DataTypeTok{IsoOptic}\NormalTok{ c ta tb tta ttb }\OtherTok{{-}>} \DataTypeTok{IsoOptic}\NormalTok{ c a b ta tb}
          \OtherTok{{-}>} \DataTypeTok{IsoOptic}\NormalTok{ c a b tta ttb}
\NormalTok{composeIso}
\NormalTok{    (}\DataTypeTok{IsoOptic}\NormalTok{ (}\OtherTok{alpha1 ::}\NormalTok{ tta }\OtherTok{{-}>}\NormalTok{ f ta) (}\OtherTok{beta1 ::}\NormalTok{ f tb }\OtherTok{{-}>}\NormalTok{ ttb))}
\NormalTok{    (}\DataTypeTok{IsoOptic}\NormalTok{ (}\OtherTok{alpha2 ::}\NormalTok{ ta }\OtherTok{{-}>}\NormalTok{ g a) (}\OtherTok{beta2 ::}\NormalTok{ g b }\OtherTok{{-}>}\NormalTok{ tb)) }\OtherTok{=}
\NormalTok{      (}\DataTypeTok{IsoOptic}\NormalTok{ \textbackslash{}\textbackslash{} (}\OtherTok{composeProof ::}\NormalTok{ (c f, c g) }\OperatorTok{:{-}}\NormalTok{ c (}\DataTypeTok{Compose}\NormalTok{ f g))) alpha beta}
        \KeywordTok{where}
\OtherTok{            alpha ::}\NormalTok{ tta }\OtherTok{{-}>} \DataTypeTok{Compose}\NormalTok{ f g a}
\NormalTok{            alpha }\OtherTok{=} \DataTypeTok{Compose} \OperatorTok{.}\NormalTok{ fmapFMonoid }\OperatorTok{@}\NormalTok{c alpha2 }\OperatorTok{.}\NormalTok{ alpha1}
\OtherTok{            beta ::} \DataTypeTok{Compose}\NormalTok{ f g b }\OtherTok{{-}>}\NormalTok{ ttb}
\NormalTok{            beta }\OtherTok{=}\NormalTok{ beta1 }\OperatorTok{.}\NormalTok{ fmapFMonoid }\OperatorTok{@}\NormalTok{c beta2 }\OperatorTok{.}\NormalTok{ unCompose}

\KeywordTok{instance} \DataTypeTok{FunctorMonoid}\NormalTok{ c }\OtherTok{=>} \DataTypeTok{OpticFamily}\NormalTok{ (}\DataTypeTok{IsoOptic}\NormalTok{ c) }\KeywordTok{where}
\NormalTok{    composeOptic }\OtherTok{=}\NormalTok{ composeIso}
\NormalTok{    injOptic f g }\OtherTok{=} \DataTypeTok{IsoOptic}\NormalTok{ (}\DataTypeTok{Id} \OperatorTok{.}\NormalTok{ f) (g }\OperatorTok{.}\NormalTok{ unId)}
\NormalTok{    mapOptic (}\DataTypeTok{IsoOptic}\NormalTok{ alpha beta) f }\OtherTok{=}\NormalTok{ beta }\OperatorTok{.}\NormalTok{ fmapFMonoid }\OperatorTok{@}\NormalTok{c f }\OperatorTok{.}\NormalTok{ alpha}

\KeywordTok{type} \DataTypeTok{ProfOptic}\NormalTok{ c a b s t }\OtherTok{=} \KeywordTok{forall}\NormalTok{ p}\OperatorTok{.} \DataTypeTok{Enhancing}\NormalTok{ c p }\OtherTok{=>}\NormalTok{ p a b }\OtherTok{{-}>}\NormalTok{ p s t}

\KeywordTok{newtype} \DataTypeTok{WrapProfOptic}\NormalTok{ c a b s t }\OtherTok{=} \DataTypeTok{WPO}\NormalTok{ \{}\OtherTok{ unWPO ::} \DataTypeTok{ProfOptic}\NormalTok{ c a b s t \}}

\KeywordTok{instance} \DataTypeTok{FunctorMonoid}\NormalTok{ c }\OtherTok{=>} \DataTypeTok{OpticFamily}\NormalTok{ (}\DataTypeTok{WrapProfOptic}\NormalTok{ c) }\KeywordTok{where}
\NormalTok{    injOptic f g }\OtherTok{=} \DataTypeTok{WPO} \OperatorTok{$}\NormalTok{ dimap f g}
\NormalTok{    composeOptic (}\DataTypeTok{WPO}\NormalTok{ l1) (}\DataTypeTok{WPO}\NormalTok{ l2) }\OtherTok{=} \DataTypeTok{WPO} \OperatorTok{$}\NormalTok{ l1 }\OperatorTok{.}\NormalTok{ l2}
\NormalTok{    mapOptic (}\DataTypeTok{WPO}\NormalTok{ l) }\OtherTok{=}\NormalTok{ l}

\KeywordTok{class} \DataTypeTok{Functor}\NormalTok{ f }\OtherTok{=>} \DataTypeTok{Functorize}\NormalTok{ op f }\KeywordTok{where}
\OtherTok{    enhanceFop ::}\NormalTok{ op a b (f a) (f b)}

\KeywordTok{instance} \DataTypeTok{OpticFamily}\NormalTok{ op }\OtherTok{=>} \DataTypeTok{Functorize}\NormalTok{ op }\DataTypeTok{Id} \KeywordTok{where}
\NormalTok{    enhanceFop }\OtherTok{=}\NormalTok{ injOptic unId }\DataTypeTok{Id}

\KeywordTok{instance}\NormalTok{ (}\DataTypeTok{OpticFamily}\NormalTok{ op, }\DataTypeTok{Functorize}\NormalTok{ op f, }\DataTypeTok{Functorize}\NormalTok{ op g) }\OtherTok{=>}
        \DataTypeTok{Functorize}\NormalTok{ op (}\DataTypeTok{Compose}\NormalTok{ f g) }\KeywordTok{where}
\NormalTok{    enhanceFop }\OtherTok{=}\NormalTok{ injOptic unCompose }\DataTypeTok{Compose} \OperatorTok{.:.}\NormalTok{ enhanceFop }\OperatorTok{.:.}\NormalTok{ enhanceFop}

\KeywordTok{instance} \DataTypeTok{OpticFamily}\NormalTok{ op }\OtherTok{=>} \DataTypeTok{FunctorMonoid}\NormalTok{ (}\DataTypeTok{Functorize}\NormalTok{ op) }\KeywordTok{where}
\NormalTok{    functorProof }\OtherTok{=} \DataTypeTok{Sub} \DataTypeTok{Dict}
\NormalTok{    composeProof }\OtherTok{=} \DataTypeTok{Sub} \DataTypeTok{Dict}

\KeywordTok{class}\NormalTok{ (}\DataTypeTok{OpticFamily}\NormalTok{ op, }\DataTypeTok{FunctorMonoid}\NormalTok{ c) }\OtherTok{=>} \DataTypeTok{Enhanceable}\NormalTok{ c op }\KeywordTok{where}
\OtherTok{    enhanceOp ::}\NormalTok{ c f }\OtherTok{=>}\NormalTok{ op a b (f a) (f b)}

\KeywordTok{instance} \DataTypeTok{FunctorMonoid}\NormalTok{ c }\OtherTok{=>} \DataTypeTok{Enhanceable}\NormalTok{ c (}\DataTypeTok{IsoOptic}\NormalTok{ c) }\KeywordTok{where}
\NormalTok{    enhanceOp }\OtherTok{=}\NormalTok{ enhanceIso}

\KeywordTok{instance} \DataTypeTok{FunctorMonoid}\NormalTok{ c }\OtherTok{=>} \DataTypeTok{Enhanceable}\NormalTok{ c (}\DataTypeTok{WrapProfOptic}\NormalTok{ c) }\KeywordTok{where}
\NormalTok{    enhanceOp }\OtherTok{=} \DataTypeTok{WPO}\NormalTok{ (enhance }\OperatorTok{@}\NormalTok{c)}

\KeywordTok{instance} \DataTypeTok{OpticFamily}\NormalTok{ op }\OtherTok{=>} \DataTypeTok{Enhanceable}\NormalTok{ (}\DataTypeTok{Functorize}\NormalTok{ op) op }\KeywordTok{where}
\NormalTok{    enhanceOp }\OtherTok{=}\NormalTok{ enhanceFop}

\OtherTok{enhanceToArrow ::} \KeywordTok{forall}\NormalTok{ c op}\OperatorTok{.} \DataTypeTok{Enhanceable}\NormalTok{ c op }\OtherTok{=>} \DataTypeTok{IsoOptic}\NormalTok{ c }\OperatorTok{\textasciitilde{}\textasciitilde{}>}\NormalTok{ op}
\NormalTok{enhanceToArrow (}\DataTypeTok{IsoOptic}\NormalTok{ alpha beta) }\OtherTok{=}\NormalTok{ injOptic alpha beta }\OperatorTok{.:.}\NormalTok{ enhanceOp }\OperatorTok{@}\NormalTok{c}

\KeywordTok{instance} \DataTypeTok{Enhanceable}\NormalTok{ c op }\OtherTok{=>} \DataTypeTok{Enhancing}\NormalTok{ c (op a b) }\KeywordTok{where}
\NormalTok{    enhance l }\OtherTok{=}\NormalTok{ enhanceOp }\OperatorTok{@}\NormalTok{c }\OperatorTok{.:.}\NormalTok{ l}

\OtherTok{profEnhanceToArrow ::} \DataTypeTok{Enhanceable}\NormalTok{ c op }\OtherTok{=>} \DataTypeTok{ProfOptic}\NormalTok{ c }\OperatorTok{\textasciitilde{}\textasciitilde{}>}\NormalTok{ op}
\NormalTok{profEnhanceToArrow l }\OtherTok{=}\NormalTok{ l idOptic}

\OtherTok{isoToProf ::} \KeywordTok{forall}\NormalTok{ c}\OperatorTok{.} \DataTypeTok{IsoOptic}\NormalTok{ c }\OperatorTok{\textasciitilde{}\textasciitilde{}>} \DataTypeTok{ProfOptic}\NormalTok{ c}
\NormalTok{isoToProf }\OtherTok{=}\NormalTok{ unWPO }\OperatorTok{@}\NormalTok{c }\OperatorTok{.}\NormalTok{ enhanceToArrow}

\OtherTok{profToIso ::} \KeywordTok{forall}\NormalTok{ c}\OperatorTok{.} \DataTypeTok{FunctorMonoid}\NormalTok{ c }\OtherTok{=>} \DataTypeTok{ProfOptic}\NormalTok{ c }\OperatorTok{\textasciitilde{}\textasciitilde{}>} \DataTypeTok{IsoOptic}\NormalTok{ c}
\NormalTok{profToIso }\OtherTok{=}\NormalTok{ profEnhanceToArrow }\OperatorTok{@}\NormalTok{c}

\OtherTok{unfunctorize ::} \DataTypeTok{OpticFamily}\NormalTok{ op }\OtherTok{=>} \DataTypeTok{IsoOptic}\NormalTok{ (}\DataTypeTok{Functorize}\NormalTok{ op) }\OperatorTok{\textasciitilde{}\textasciitilde{}>}\NormalTok{ op}
\NormalTok{unfunctorize }\OtherTok{=}\NormalTok{ enhanceToArrow}

\KeywordTok{class}\NormalTok{ (}\DataTypeTok{OpticFamily}\NormalTok{ op, }\DataTypeTok{FunctorMonoid}\NormalTok{ c) }\OtherTok{=>} \DataTypeTok{ProfEnc}\NormalTok{ c op }\KeywordTok{where}
\OtherTok{    encodeProf ::}\NormalTok{ op }\OperatorTok{\textasciitilde{}\textasciitilde{}>} \DataTypeTok{ProfOptic}\NormalTok{ c}
\OtherTok{    decodeProf ::} \DataTypeTok{ProfOptic}\NormalTok{ c }\OperatorTok{\textasciitilde{}\textasciitilde{}>}\NormalTok{ op}

\KeywordTok{class} \DataTypeTok{OpticFamily}\NormalTok{ op }\OtherTok{=>} \DataTypeTok{ProfEncF}\NormalTok{ op }\KeywordTok{where}
\OtherTok{    concreteToIso ::}\NormalTok{ op }\OperatorTok{\textasciitilde{}\textasciitilde{}>} \DataTypeTok{IsoOptic}\NormalTok{ (}\DataTypeTok{Functorize}\NormalTok{ op)}

\KeywordTok{instance} \DataTypeTok{ProfEncF}\NormalTok{ op }\OtherTok{=>} \DataTypeTok{ProfEnc}\NormalTok{ (}\DataTypeTok{Functorize}\NormalTok{ op) op }\KeywordTok{where}
\NormalTok{    encodeProf }\OtherTok{=}\NormalTok{ isoToProf }\OperatorTok{.}\NormalTok{ concreteToIso}
\NormalTok{    decodeProf l }\OtherTok{=}\NormalTok{ enhanceToArrow }\OperatorTok{$}\NormalTok{ profToIso }\OperatorTok{@}\NormalTok{(}\DataTypeTok{Functorize}\NormalTok{ op) l}

\CommentTok{{-}{-} Adapter}
\KeywordTok{data} \DataTypeTok{Adapter}\NormalTok{ a b s t }\OtherTok{=} \DataTypeTok{Adapter}\NormalTok{ (s }\OtherTok{{-}>}\NormalTok{ a) (b }\OtherTok{{-}>}\NormalTok{ t)}

\KeywordTok{instance} \DataTypeTok{OpticFamily} \DataTypeTok{Adapter} \KeywordTok{where}
\NormalTok{    injOptic f g }\OtherTok{=} \DataTypeTok{Adapter}\NormalTok{ f g}
\NormalTok{    composeOptic (}\DataTypeTok{Adapter}\NormalTok{ f1 g1) (}\DataTypeTok{Adapter}\NormalTok{ f2 g2) }\OtherTok{=}
        \DataTypeTok{Adapter}\NormalTok{ (f2 }\OperatorTok{.}\NormalTok{ f1) (g1 }\OperatorTok{.}\NormalTok{ g2)}
\NormalTok{    mapOptic (}\DataTypeTok{Adapter}\NormalTok{ f g) }\OtherTok{=}\NormalTok{ dimap f g}

\KeywordTok{instance} \DataTypeTok{ProfEncF} \DataTypeTok{Adapter} \KeywordTok{where}
\NormalTok{    concreteToIso (}\DataTypeTok{Adapter}\NormalTok{ f g) }\OtherTok{=}\NormalTok{ injOptic f g}

\CommentTok{{-}{-} Lens}
\KeywordTok{data} \DataTypeTok{Lens}\NormalTok{ a b s t }\OtherTok{=} \DataTypeTok{Lens}\NormalTok{ \{}\OtherTok{ lget ::}\NormalTok{ s }\OtherTok{{-}>}\NormalTok{ a,}\OtherTok{ lput ::}\NormalTok{ b }\OtherTok{{-}>}\NormalTok{ s }\OtherTok{{-}>}\NormalTok{ t \}}

\KeywordTok{instance} \DataTypeTok{OpticFamily} \DataTypeTok{Lens} \KeywordTok{where}
\NormalTok{    injOptic f g }\OtherTok{=} \DataTypeTok{Lens}\NormalTok{ f (}\FunctionTok{const} \OperatorTok{.}\NormalTok{ g)}
\NormalTok{    composeOptic (}\DataTypeTok{Lens}\NormalTok{ get1 put1) (}\DataTypeTok{Lens}\NormalTok{ get2 put2) }\OtherTok{=}
        \DataTypeTok{Lens}\NormalTok{ (get2 }\OperatorTok{.}\NormalTok{ get1) (\textbackslash{}y s }\OtherTok{{-}>}\NormalTok{ put1 (put2 y (get1 s)) s)}
\NormalTok{    mapOptic (}\DataTypeTok{Lens}\NormalTok{ get put) f s }\OtherTok{=}\NormalTok{ put (f (get s)) s}

\KeywordTok{instance} \DataTypeTok{Functorize} \DataTypeTok{Lens}\NormalTok{ ((,) c) }\KeywordTok{where}
\OtherTok{    enhanceFop ::} \DataTypeTok{Lens}\NormalTok{ a b (c, a) (c, b)}
\NormalTok{    enhanceFop }\OtherTok{=} \DataTypeTok{Lens} \FunctionTok{snd}\NormalTok{ (}\FunctionTok{fmap} \OperatorTok{.} \FunctionTok{const}\NormalTok{)}

\KeywordTok{instance} \DataTypeTok{ProfEncF} \DataTypeTok{Lens} \KeywordTok{where}
\NormalTok{    concreteToIso (}\DataTypeTok{Lens}\NormalTok{ get put) }\OtherTok{=}
        \DataTypeTok{IsoOptic}\NormalTok{ (\textbackslash{}s }\OtherTok{{-}>}\NormalTok{ (s, get s)) (\textbackslash{}(s, b) }\OtherTok{{-}>}\NormalTok{ put b s)}

\OtherTok{getProf ::} \DataTypeTok{ProfOptic}\NormalTok{ (}\DataTypeTok{Functorize} \DataTypeTok{Lens}\NormalTok{) a b s t }\OtherTok{{-}>}\NormalTok{ (s }\OtherTok{{-}>}\NormalTok{ a)}
\NormalTok{getProf l }\OtherTok{=}\NormalTok{ get}
    \KeywordTok{where} \DataTypeTok{Lens}\NormalTok{ get \_ }\OtherTok{=}\NormalTok{ decodeProf }\OperatorTok{@}\NormalTok{(}\DataTypeTok{Functorize} \DataTypeTok{Lens}\NormalTok{) l}

\OtherTok{first ::} \DataTypeTok{Lens}\NormalTok{ a b (a, c) (b, c)}
\NormalTok{first }\OtherTok{=} \DataTypeTok{Lens} \FunctionTok{fst}\NormalTok{ (\textbackslash{}x (\_, y) }\OtherTok{{-}>}\NormalTok{ (x, y))}

\OtherTok{firstOf4 ::} \DataTypeTok{Lens}\NormalTok{ a a\textquotesingle{} (((a, b), c), d) (((a\textquotesingle{}, b), c), d)}
\NormalTok{firstOf4 }\OtherTok{=}\NormalTok{ first }\OperatorTok{.:.}\NormalTok{ first }\OperatorTok{.:.}\NormalTok{ first}

\OtherTok{first\textquotesingle{} ::} \DataTypeTok{ProfOptic}\NormalTok{ (}\DataTypeTok{Functorize} \DataTypeTok{Lens}\NormalTok{) a b (a, c) (b, c)}
\NormalTok{first\textquotesingle{} }\OtherTok{=}\NormalTok{ encodeProf }\OperatorTok{@}\NormalTok{(}\DataTypeTok{Functorize} \DataTypeTok{Lens}\NormalTok{) first}

\OtherTok{firstOf4\textquotesingle{} ::} \DataTypeTok{ProfOptic}\NormalTok{ (}\DataTypeTok{Functorize} \DataTypeTok{Lens}\NormalTok{) a a\textquotesingle{} (((a, b), c), d) (((a\textquotesingle{}, b), c), d)}
\NormalTok{firstOf4\textquotesingle{} }\OtherTok{=}\NormalTok{ first\textquotesingle{} }\OperatorTok{.}\NormalTok{ first\textquotesingle{} }\OperatorTok{.}\NormalTok{ first\textquotesingle{}}

\OtherTok{getFirstOf4 ::}\NormalTok{ (((a, b), c), d) }\OtherTok{{-}>}\NormalTok{ a}
\NormalTok{getFirstOf4 }\OtherTok{=}\NormalTok{ getProf firstOf4\textquotesingle{}}

\CommentTok{{-}{-} Prism}
\KeywordTok{data} \DataTypeTok{Prism}\NormalTok{ a b s t }\OtherTok{=} \DataTypeTok{Prism}\NormalTok{ \{}\OtherTok{ pmatch ::}\NormalTok{ s }\OtherTok{{-}>}\NormalTok{ t }\OperatorTok{+}\NormalTok{ a,}\OtherTok{ pbuild ::}\NormalTok{ b }\OtherTok{{-}>}\NormalTok{ t \}}

\KeywordTok{instance} \DataTypeTok{OpticFamily} \DataTypeTok{Prism} \KeywordTok{where}
\NormalTok{    injOptic f g }\OtherTok{=} \DataTypeTok{Prism}\NormalTok{ (}\DataTypeTok{Right} \OperatorTok{.}\NormalTok{ f) g}
\NormalTok{    composeOptic (}\DataTypeTok{Prism}\NormalTok{ m1 b1) (}\DataTypeTok{Prism}\NormalTok{ m2 b2) }\OtherTok{=}
        \DataTypeTok{Prism}\NormalTok{ (}\FunctionTok{either} \DataTypeTok{Left} \FunctionTok{id} \OperatorTok{.} \FunctionTok{fmap}\NormalTok{ (lmap b1 }\OperatorTok{.}\NormalTok{ m2) }\OperatorTok{.}\NormalTok{ m1) (b1 }\OperatorTok{.}\NormalTok{ b2)}
\NormalTok{    mapOptic (}\DataTypeTok{Prism}\NormalTok{ match build) f }\OtherTok{=} \FunctionTok{either} \FunctionTok{id} \FunctionTok{id} \OperatorTok{.} \FunctionTok{fmap}\NormalTok{ (build }\OperatorTok{.}\NormalTok{ f) }\OperatorTok{.}\NormalTok{ match}

\KeywordTok{instance} \DataTypeTok{Functorize} \DataTypeTok{Prism}\NormalTok{ ((}\OperatorTok{+}\NormalTok{) c) }\KeywordTok{where}
\OtherTok{    enhanceFop ::} \DataTypeTok{Prism}\NormalTok{ a b (c }\OperatorTok{+}\NormalTok{ a) (c }\OperatorTok{+}\NormalTok{ b)}
\NormalTok{    enhanceFop }\OtherTok{=} \DataTypeTok{Prism}\NormalTok{ (}\FunctionTok{either}\NormalTok{ (}\DataTypeTok{Left} \OperatorTok{.} \DataTypeTok{Left}\NormalTok{) }\DataTypeTok{Right}\NormalTok{) }\DataTypeTok{Right}

\KeywordTok{instance} \DataTypeTok{ProfEncF} \DataTypeTok{Prism} \KeywordTok{where}
\NormalTok{    concreteToIso (}\DataTypeTok{Prism}\NormalTok{ match build) }\OtherTok{=} \DataTypeTok{IsoOptic}\NormalTok{ match (}\FunctionTok{either} \FunctionTok{id}\NormalTok{ build)}

\CommentTok{{-}{-} Setter}
\KeywordTok{data} \DataTypeTok{Setter}\NormalTok{ a b s t }\OtherTok{=} \DataTypeTok{Setter}\NormalTok{ ((a }\OtherTok{{-}>}\NormalTok{ b) }\OtherTok{{-}>}\NormalTok{ (s }\OtherTok{{-}>}\NormalTok{ t))}

\KeywordTok{instance} \DataTypeTok{OpticFamily} \DataTypeTok{Setter} \KeywordTok{where}
\NormalTok{    injOptic f g }\OtherTok{=} \DataTypeTok{Setter}\NormalTok{ (dimap f g)}
\NormalTok{    composeOptic (}\DataTypeTok{Setter}\NormalTok{ f1) (}\DataTypeTok{Setter}\NormalTok{ f2) }\OtherTok{=} \DataTypeTok{Setter}\NormalTok{ (f1 }\OperatorTok{.}\NormalTok{ f2)}
\NormalTok{    mapOptic (}\DataTypeTok{Setter}\NormalTok{ f) }\OtherTok{=}\NormalTok{ f}

\KeywordTok{newtype} \DataTypeTok{CPS}\NormalTok{ t b a }\OtherTok{=} \DataTypeTok{CPS}\NormalTok{ \{}\OtherTok{ unCPS ::}\NormalTok{ (a }\OtherTok{{-}>}\NormalTok{ b) }\OtherTok{{-}>}\NormalTok{ t \}}

\KeywordTok{instance} \DataTypeTok{Functor}\NormalTok{ (}\DataTypeTok{CPS}\NormalTok{ t b) }\KeywordTok{where}
    \FunctionTok{fmap}\NormalTok{ f }\OtherTok{=} \DataTypeTok{CPS} \OperatorTok{.}\NormalTok{ (}\OperatorTok{.}\NormalTok{ (}\OperatorTok{.}\NormalTok{ f)) }\OperatorTok{.}\NormalTok{ unCPS}

\KeywordTok{instance} \DataTypeTok{Functor}\NormalTok{ f }\OtherTok{=>} \DataTypeTok{Functorize} \DataTypeTok{Setter}\NormalTok{ f }\KeywordTok{where}
\NormalTok{    enhanceFop }\OtherTok{=} \DataTypeTok{Setter} \FunctionTok{fmap}

\KeywordTok{instance} \DataTypeTok{ProfEncF} \DataTypeTok{Setter} \KeywordTok{where}
\NormalTok{    concreteToIso (}\DataTypeTok{Setter}\NormalTok{ f) }\OtherTok{=} \DataTypeTok{IsoOptic}\NormalTok{ (}\DataTypeTok{CPS} \OperatorTok{.} \FunctionTok{flip}\NormalTok{ f) ((}\OperatorTok{$} \FunctionTok{id}\NormalTok{) }\OperatorTok{.}\NormalTok{ unCPS)}
\end{Highlighting}
\end{Shaded}